\PassOptionsToPackage{dvipsnames,table,xcdraw}{xcolor}
\RequirePackage{pgf}
\documentclass[sn-mathphys-num]{sn-jnl}

\usepackage{graphicx}%
\usepackage{amsmath,amssymb,amsfonts}%
\usepackage{amsthm}%
\usepackage{mathrsfs}%
\usepackage{multirow}
\usepackage[title]{appendix}%
\usepackage{nicematrix}
\usepackage{textcomp}%
\usepackage{manyfoot}%
\usepackage{booktabs}%
\usepackage{algorithm}%
\usepackage{algorithmicx}%
\usepackage{algpseudocode}%
\usepackage{listings}%
\usepackage{cleveref}
\usepackage[vskip=1.05mm,font=itshape]{quoting}
\usepackage{changepage}   
\usepackage{float}
\usepackage{enumitem}

\theoremstyle{thmstyleone}%
\newtheorem{theorem}{Theorem}
\theoremstyle{thmstyletwo}%
\newtheorem{remark}{Remark}%
\newtheorem{lemma}{Lemma}%
\newtheorem{corollary}{Corollary}%

\theoremstyle{thmstylethree}%
\newtheorem{definition}{Definition}%

\newcommand{\stronglyhard}[1]{\textcolor{BrickRed}{Strongly {\sf NP}-hard [#1]}}
\newcommand{\complete}[1]{\textcolor{BrickRed}{{\sf NP}-complete [#1]}}
\newcommand{\easy}[1]{\textcolor{ForestGreen}{#1}}

\newcommand{\asym}{\text{Asym}}
\newcommand{\minmax}{\text{MinMax}}
\newcommand{\maxmin}{\text{MaxMin}}
\newcommand{\decision}{\text{Dec}}

\newcommand{\splititemempty}{$\textsc{SplitItem}$}
\newcommand{\splittingempty}{$\textsc{Splitting}$}
\newcommand{\intervalempty}{$\textsc{Inter}$}
\newcommand{\partitionempty}{$\textsc{Part}$}

\newcommand{\partitionunbounded}[2]{#1$\textsc{-Part}($#2$)$}

\newcommand{\splititem}[3]{#1$\textsc{-SplitItem}[$#2$]($#3$)$}
\newcommand{\splitting}[3]{#1$\textsc{-Splitting}[$#2$]($#3$)$}
\newcommand{\interval}[3]{#1$\textsc{-Inter}[$#2$]($#3$)$}
\newcommand{\partition}[3]{#1$\textsc{-Part}[$#2$]($#3$)$}

\newcommand{\sumratios}{A}

\makeatletter
\DeclareRobustCommand*\cal{\@fontswitch\relax\mathcal}
\makeatother

\begin{document}

\title[Asymmetric Number Partitioning with Splitting and Interval Targets]{Asymmetric Number Partitioning with Splitting and Interval Targets \footnote{This is an extension of a paper that has been presented at the 35th International Symposium on Algorithms and Computation, (ISAAC) and appeared in its proceedings \cite{DBLP:conf/isaac/BismuthMSS24}. The earlier version focused solely on positive results concerning identical bins. We present here a generalization of these results, for asymmetric bins. This version also contains all proofs
omitted from the previous version.}}


\author*[1]{\fnm{Samuel} \sur{Bismuth}}\email{samuelbismuth101@gmail.com}
\equalcont{These authors contributed equally to this work.}


\author[1]{\fnm{Erel} \sur{Segal-Halevi}}\email{erelsgl@gmail.com}
\equalcont{These authors contributed equally to this work.}

\author[1]{\fnm{Dana} \sur{Shapira}}\email{shapird@g.ariel.ac.il}
\equalcont{These authors contributed equally to this work.}

\affil*[1]{\orgdiv{Department of Computer Science}, \orgname{Ariel University}, \orgaddress{\city{Ariel}, \postcode{40700},  \country{Israel}}}



\abstract{The $n$-way number partitioning problem, a fundamental challenge in combinatorial optimization, has significant implications for applications such as fair division and machine scheduling. Despite the fact that these problems are NP-hard, numerous approximation techniques have been developed.
We consider three closely related kinds of approximations, and various objectives such as decision, min-max, max-min, and even a generalized objective, in which the bins are not considered identical anymore, but rather asymmetric (used to solve fair division to asymmetric agents or uniform machine scheduling problems).

The first two variants optimize the partition such that: in the first variant some fixed number $s$ of items can be \emph{split} between two or more bins and in the second variant we allow at most a fixed number $t$ of \emph{splittings}. The third variant is a decision problem: the largest bin sum must be within a pre-specified interval, parameterized by a fixed rational number $u$ times the largest item size.

When the number of bins $n$ is unbounded, we show that every variant is strongly {\sf NP}-complete. 
When the number of bins $n$ is fixed, the running time depends on the fixed parameters $s,t,u$. For each variant, we give a complete picture of its running time.

For $n=2$, the running time is easy to identify. 
Our main results consider any fixed integer $n \geq 3$. 
Using a two-way polynomial-time reduction between the first and the third variant,
we show that $n$-way number-partitioning with $s$ split items can be solved in polynomial time if $s \geq n-2$, and it is {\sf NP}-complete otherwise. Also, $n$-way number-partitioning with $t$ splittings can be solved in polynomial time if $t \geq n-1$, and it is {\sf NP}-complete otherwise. Finally, we show that the third variant can be solved in polynomial time if $u \geq (n-2)/n$, and it is {\sf NP}-complete otherwise.
Our positive results for the optimization problems consider both asymmetric min-max and asymmetric max-min versions.

}




\maketitle

\section{Introduction}
The $n$-way number partitioning problem is a fundamental problem in operations research, and has been extensively studied across diverse contexts and applications \cite{DBLP:journals/anor/StorerFW96}, \cite{DBLP:journals/eor/PedrosoK10}, \cite{DBLP:journals/orf/KalczynskiGD23}, \cite{DBLP:journals/ior/AnilyF91}, \cite{DBLP:journals/cor/FariaSS21}.

In the classic setting, the inputs are a list ${\cal X}=(x_1, \ldots, x_m)$ of $m$ non-negative integers and a number $n$ of bins, and the required output is an $n$-way partition (a partition of the integers into $n$ bins) in which the sums of integers in the bins are \emph{balanced} in some sense.  

In the \emph{decision} version of the problem, the objective is to decide whether there exists an $n$-way partition of ${\cal X}$ such that every bin sum is exactly equal to $\sum_j x_j / n$ (we call it a \emph{perfect partition}). In the \emph{min-max} optimization version, the objective is to find an $n$-way partition such that the maximum bin sum is minimized, while in the \emph{max-min} optimization version, the goal is to maximize the smallest bin sum. 

Whereas in the classic objectives, all the bins are considered identical, we also consider a generalized objective that we call \emph{asymmetric},
parameterized by a list of positive integer ${\cal R} = (r_1, \dots, r_n)$, where $r_i$ is the ratio of bin $i$.
The goal in asymmetric min-max is to minimize the maximum \emph{relative sum} of a bin, which is the bin sum divided by its ratio. Asymmetric max-min is defined analogously.
The classic problems correspond to the special case in which all ratios equal 1.

For tightness, we show negative results for the  classic objectives (which are special cases of the asymmetric objectives), and show positive results for the asymmetric objectives (which are generalizations of the classic objectives).


For each problem of this paper, the problem objective is mentioned first, the fixed parameters are in square brackets, and the problem inputs are in parenthesis. Let us formally define the decision, min-max and asymmetric min-max versions of the $n$-way number partitioning problem, where $n$ is a fixed parameter:
\begin{quoting}
\partition{\decision}{$n$}{${\cal X}$}:
\quad
Decide if there exists a partition of ${\cal X}$ among $n$ bins with sums $b_1,\ldots,b_n$ such that $\max(b_1, \dots, b_n) \leq S$ where $S := \sum_j x_j / n$.
\end{quoting}
\begin{quoting}
\partition{\minmax}{$n$}{${\cal X}$}:
\quad
Minimize $\max(b_1,\ldots, b_n)$,
where $b_1,\ldots,b_n$ are sums of bins in an $n$-way partition of ${\cal X}$.
\end{quoting}
\begin{quoting}
\partition{\asym\minmax}{$n$}{${\cal X}, {\cal R}$}:
\quad
Minimize $\max(b_1,\ldots, b_n)$,
where $b_1,\ldots,b_n$ are \emph{relative} sums (the sum of bin $i$ divided by its ratio $r_i$) of bins with ratios ${\cal R}$ in an $n$-way partition of ${\cal X}$.
\end{quoting}
For every fixed $n \geq 2$ \partition{\decision}{$n$}{${\cal X}$} is known to be {\sf NP}-hard \cite{gareycomputers}; hence; the min-max and max-min variants are {\sf NP}-hard too. 
When $n$ is unbounded (part of the input), \partitionunbounded{\decision}{$n, {\cal X}$} 
is known to be strongly {\sf NP}-hard (it is equivalent to 3-partition) \cite{DBLP:journals/siamcomp/GareyJ75}. In addition, many instances of the decision version are negative (there is no perfect partition). These negative results motivate us to investigate variants of the $n$-way number partitioning problem, for which the running time complexity is better, and the number of positive instances (admitting a perfect partition) is larger. 
We present three variants, that relax the initial problem to solve our concerns.
The first two variants, \splititemempty{} and \splittingempty{} allow ``divisible'' items, bounded by some natural numbers $s$ and $t$. We define the decision and the asymmetric min-max versions as follows:
\begin{quoting}
\splititem{\decision}{$n, s$}{${\cal X}$}:
\quad
Decide if there exists a partition of ${\cal X}$ among $n$ bins with sums $b_1,\ldots, b_n$ with at most $s$ split items, such that $\max(b_1,\ldots,b_n)\leq S$ where $S := \sum_j x_j / n$.
\end{quoting}
\begin{quoting}
\splititem{\asym\minmax}{$n, s$}{${\cal X}, {\cal R}$}:
\quad
Minimize $\max(b_1,\ldots, b_n)$,
where $b_1,\ldots,b_n$ are relative sums (the sum of bin $i$ divided by its ratio $r_i$) of bins with ratios ${\cal R}$ in an $n$-way partition of ${\cal X}$ in which at most $s$ items are split.
\end{quoting}
\begin{quoting}
\splitting{\decision}{$n, t$}{${\cal X}$}:
\quad
Decide if there exists a partition of ${\cal X}$ among $n$ bins with sums $b_1,\ldots, b_n$ with at most $t$ splittings, such that $\max(b_1,\ldots,b_n)\leq S$ where $S := \sum_j x_j / n$.
\end{quoting}
\begin{quoting}
\splitting{\asym\minmax}{$n, t$}{${\cal X}, {\cal R}$}:
\quad
Minimize $\max(b_1,\ldots, b_n)$,
where $b_1,\ldots,b_n$ are relative sums (the sum of bin $i$ divided by its ratio $r_i$) of bins with ratios ${\cal R}$ in an $n$-way partition of ${\cal X}$ in which at most $t$ splittings are allowed.
\end{quoting}
The number of splittings is at least the number of split items but it might be larger. For example, a single item split into $10$ different bins counts as $9$ splittings. 
Note that the problem definitions do not determine in advance which items will be split, but only bound their number, or bound the number of splittings. The solver may decide which items to split after receiving the input. 

Our motivating application for the variants comes from fair division and machine scheduling. For fair division, some $m$ items with market values $x_1,\ldots,x_m$ have to be divided among $n$ asymmetric agents. A \emph{perfect partition} is one in which each agent $i$ receives a total value of exactly $r_i \cdot S$. When the original instance does not admit a perfect partition, we may want to \emph{split} one or more items among two or more agents. Or, we can allow some \emph{splittings}. Divisible items are widespread in fair division applications: the ownership of one or more items may be split to attain a perfectly fair partition. However, divisible items may be inconvenient or expensive. Therefore, the number of split items or splittings should be bounded. The same is true for uniform machine scheduling, in which agents are considered as machines, and items as jobs. A job can be divided and processed by two machines simultaneously. The following examples tackle real-life fair division to asymmetric agents or uniform machines scheduling problems. 

    (1) Consider a firm that owns $m=2$ houses  with values ${\cal X} = (13,3)$. 
    The firm belongs to $n=2$ partners with rights ${\cal R} = (3, 1)$. 
    The firm dissolves, and the partners need to divide the property among them.
    Without splitting, one partner necessarily gets the house with a value of $13$, which is worth more than $3/4$ of the total estate. So, if all houses are considered discrete, then an equal division is not possible. 
    If all houses can be split, then an equal division is easy to attain by giving 75\% of every house to the partner with the right of 3 and 25\% of every house to the second partner, but it is inconvenient since it requires all houses to be jointly managed. A solution often used in practice is to decide in advance that \emph{a single} house can be split (or only one splitting is allowed). In this case, after receiving the input, we can determine that splitting the house with a value of $13$ (such that the partners get $(12, 3+1)$) lets us attain a division in which each partner receives the same relative sum of $4$. 

    (2) Consider a food factory with $n=3$ chopping machines with speeds ${\cal R} = (2,1,1)$, who has to cut $m=4$ vegetables with processing times ${\cal X} = (22, 7, 4, 3)$ minutes. 
    Each job is divisible as one vegetable may be cut in different machines, but splitting a job is inconvenient since it requires washing more dishes. 
    Without splitting, the minimum total processing time is 11 minutes: $(22)/2, (7)/1, (4+3)/1$. By splitting the vegetable with processing time $22$ into three different machines, the processing time is 9 minutes: $(18)/2, (7+2)/1, (4+3+2)/1$.
    
The third variant \intervalempty{} only admits a decision version, parameterized by a rational number $u \geq 0$:
\begin{quoting}
\interval{\asym\decision}{$n, u$}{${\cal X}, {\cal R}$}:
\quad
Decide if there exists a partition of ${\cal X}$ among $n$ bins with ratios {\cal R} and relative sums $b_1,\ldots, b_n$ such that $S \leq \max(b_1,\ldots,b_n) \leq S + u \cdot M$, where $S := (\sum_j x_j)/(\sum_i r_i)$ and  $M := (\max_j x_j)/(\sum_i r_i)$. 
\end{quoting}
In general, the runtime complexity of this problem depends on the size of the allowed interval (i.e., the interval $[S, S + u \cdot M$]): the problem is {\sf NP}-complete when the interval is ``small'' and in {\sf P} when the interval is ``large''. 
We notice that, if we could solve \intervalempty{} for any interval length in polynomial-time, then by binary search we could solve \partitionempty{} in polynomial-time; which is not possible unless {\sf P}={\sf NP}. So in \intervalempty{}, we look for the smallest interval for which we can decide in polynomial time whether it contains a solution.

As an application example, consider the fair allocation of indivisible items among two asymmetric agents. Suppose there is a small amount of money, that can be used to compensate for a small deviation from equality in the allocation. But if the deviation is too big, the agents prefer to find another solution. We can use \intervalempty{} check the feasibility of compensation, by letting the interval length be the amount of money available.

The \intervalempty{} variant is closely related to 
the well-known {\it Fully Polynomial-Time Approximation Scheme} ({\sf FPTAS}) definition:
\begin{definition}
An {\sf FPTAS} for \partition{\asym\minmax}{$n$}{${\cal X}, {\cal R}$} is an algorithm that finds, for each rational $\epsilon>0$, an $n$-way partition of ${\cal X}$ with 
$OPT \leq \max(b_1,\ldots,b_n) \leq (1+\epsilon) \cdot OPT$, 
where OPT is the smallest possible value of $\max(b_1,\ldots,b_n)$ in the given instance, in time $O(poly(m,1/\epsilon, \log (\sum_j x_j))))$.
\end{definition}
An {\sf FPTAS} finds a solution for which the relative deviation from optimality depends on the optimal \emph{integral} solution, OPT. 
In contrast, the \interval{\asym\decision}{$n, u$}{${\cal X}, {\cal R}$} problem
can equivalently be phrased as: decide if there exists a partition such that $S\leq \max(b_1,\ldots,b_n)\leq (1+\delta)\cdot S$, for $\delta := uM/S$.
That is, the relative deviation from optimality depends on the optimal \emph{fractional} solution, $S$.
Although the problems are different, we will use {\sf FPTAS}-s as tools for solving the \intervalempty{} problem.


\paragraph*{Contribution}
When $s,t,u=0$, \splititemempty{}, \splittingempty{} and \intervalempty{} decision versions are equivalent to the {\sf NP}-hard \partitionempty{} decision version. In contrast, when $s, t \geq n-1$ the problem, even with asymmetric objectives, is easily solvable by the following algorithm: put the items on a line, cut the line into $n$ pieces such that the length of line $i$ is $r_i \cdot S$, and put the corresponding piece in bin $i$. Since $n-1$ cuts are made, at most $n-1$ items need to be split. 
So for $n=2$, the runtime complexity of the \splititem{\asym\decision}{$n, s$}{${\cal X}, {\cal R}$} and \splitting{\asym\decision}{$n, t$}{${\cal X}, {\cal R}$} problem is well-understood (assuming ${\sf P}\neq {\sf NP}$): it is polynomial-time solvable if and only if $s, t \geq 1$. 
The case for \intervalempty{} is slightly different since $u$ is a rational number. We summarize all our results in \Cref{tab:nagents}.
\begin{table*}
\begin{center}
\small
\setlength{\tabcolsep}{3pt}
\begin{NiceTabular}{|l|l|l|l|l|}

\hline
\textbf{Problem} & \textbf{Objective} & \textbf{Nb of bins} & \textbf{Bound} & \textbf{Run-time complexity} \\ \hline \hline

\partitionempty & \decision & Unbounded & $s=t$ & \stronglyhard{\cite{DBLP:journals/siamcomp/GareyJ75}} \\ 
\cmidrule{3-3} \cmidrule{5-5}
&  & Const $n$ & $=u=0$ & \complete{\cite{gareycomputers}} \\
\hline \hline

\splititemempty & \decision & Unbounded & $s$ (any) & \stronglyhard{Cor. \ref{cor:splitstronglyhard}} \\ \cmidrule{3-5}
&  & Const $n \geq 3$ & $s < n-2$ & \complete{\Cref{cor:splithard}}\\  \cmidrule{2-2} \cmidrule{4-5}
& \asym\minmax &  & $s \geq n-2$ & \easy{$O(\emph{poly}(m, \log{(\sum_j x_j)}))$ [Cor. \ref{cor:nsplitparmm}]} \\  \cmidrule{2-2} \cmidrule{4-5}
& \asym\maxmin &  & $s \geq n-2$ & \easy{$O(\emph{poly}(m, \log{(\sum_j x_j)}))$ [Thm. \ref{thm:maxminalgo}]} \\  \cmidrule{2-2} \cmidrule{3-5}
& All & Const $n \geq 2$ & $s \geq n-1$ & \easy{$O(m+n)$ [cut-the-line]} \\ 
\hline \hline

\splittingempty & \decision & Unbounded & $t$ (any) & \stronglyhard{Thm. \ref{thm:splittingstrong}} \\ \cmidrule{3-5}
&  & Const $n$ & $t < n-1$ & \complete{\Cref{thm:sharing}} \\  \cmidrule{2-2} \cmidrule{4-5}
& All & & $t \geq n-1$ & \easy{$O(m+n)$ [cut-the-line]} \\ 
\hline \hline

\intervalempty & \decision & Unbounded & $u$ (any) & \stronglyhard{Thm. \ref{thm:n-input-hard}} \\ \cmidrule{3-5}
&  & Const $n \geq 3$ & $u < n-2$ & \complete{\Cref{thm:np-complete-d-possible}} \\ \cmidrule{2-2} \cmidrule{4-5}
& \asym\decision &  & $u \geq n-2$ & \easy{$O(\emph{poly}(m, \log{(\sum_j x_j)}))$ [Thm. \ref{thm:n>=3}]} 
\\
\cmidrule{3-5}
&  & Const $n\geq 2$ & $u \geq n-1$ & \easy{$O(m+n)$ [cut-the-line+Thm. \ref{thm:from-interval-to-shared}]} 
\\
\cmidrule{3-5}
&  & $n=2$ & $u > 0$ & \easy{$O(\emph{poly}(m, \log{(\sum_j x_j)}, 1/ u))$ [Cor. \ref{cor:u-instead-of-v}]}
\\ 
\hline 
\end{NiceTabular}
\end{center}
\caption{\label{tab:nagents} 
Run-time complexity of the $n$-way number partitioning variants. 
In \splititemempty, $s$ (an integer) is the number of items the algorithm is allowed to split.
In \splittingempty, $t$ (an integer) is the number of splittings the algorithm is allowed to make.
In \intervalempty, $u$ (a rational number) is the ratio between the allowed interval length and $M$.
}
\end{table*}

In \Cref{sub:equivalence} we show a two-way polynomial-time reduction between problems \splititemempty{} and \intervalempty{}.
This reduction is the key for many of our results. We use it to handle the case where the number of split items is smaller than $n-2$. First, we design an {\sf FPTAS} for \splititem{\asym\minmax}{$n, s$}{${\cal X}, {\cal R}$} in \Cref{sec:approximations}. Second, we develop a practical (not polynomial-time) algorithm, for solving \splititem{\minmax}{$n, s$}{${\cal X}$} for any $s \geq 0$. The algorithm can use any practical algorithm for solving the \partition{\minmax}{$n$}{${\cal X}$} problem.
The latest helps us in \Cref{sec:experiments} to conduct some experiences to various randomly generated instances and analyze the effect of $s$ on the quality of the attained solution. 
The supplement provides complementary results and technical proof details omitted from the main text. 


\paragraph*{Techniques}
The main results of our paper are based on {\sf FPTAS} for several number partitioning problems.
Woeginger \cite{woeginger2000does} gives a general method for converting a dynamic program to an {\sf FPTAS}. He shows that his method can be used to design {\sf FPTAS} for hundreds of different combinatorial optimization problems. We apply his method to several different problems.
We remark that the {\sf FPTAS} designed by Woeginger's method for partitioning problems are polynomial in $m$ and $\log{\sum_j x_j}$, but exponential in $n$;     
the same is true for our algorithms.

\section{Related Work} \label{sec:related-work}

Combinatorial optimization problems traditionally distinguish between discrete and continuous variables. E.g., when a problem is modeled by a mixed-integer program, each variable in the program is determined in advance to be either discrete (must get an integer value) or continuous (can get any real value). 
The problems we study belong to a much smaller class of problems, in which all variables are potentially continuous, but there is an upper bound on the number of variables that can be non-discrete. We describe some such problems below.

\textbf{Bounded splitting in fair division:}
The idea of finding fair allocations with a bounded number of split items originated from \cite{Brams1996Fair,brams2000winwin}.
They presented the \emph{Adjusted Winner} ({\sf AW}) procedure for allocating items among two agents with possibly different valuations.
{\sf AW} finds an allocation that is \emph{envy-free} (no agent prefers the bundle of another agent), \emph{equitable} (both agents receive the same subjective value), and \emph{Pareto-optimal} (there is no other allocation where some agent gains and no agent loses), and in addition, at most a single item is split between the agents.
Hence, {\sf AW} solves a problem that is similar to
\splititem{\decision}{$n=2, s=1$}{${\cal X}$} but more general, since {\sf AW} allows the agents to have different valuations to the same items.
{\sf AW} was applied (at least theoretically) to division problems in divorce cases and international disputes \citep{Brams1996Camp,Massoud2000Fair} and was studied empirically by  \cite{Schneider2004Limitations,Daniel2005Fair}.
The 
{\sf AW}
procedure is designed for two agents. 
For $n\geq 3$ agents, the number of split items allowed was studied in an unpublished manuscript of \cite{wilson1998fair} using linear programming techniques.
He proved the existence of an \emph{egalitarian}
allocation of goods (i.e., an allocation in which all agents have the largest possible equal utility \citep{pazner1978egalitarian}), with at most $n-1$ split items; this can be seen as a generalization of
\splititem{\decision}{$n, s=n-1$}{${\cal X}$}.

Recently, \cite{bei2020fair, DBLP:journals/corr/abs-2310-00976} studied an allocation problem where some items are divisible and some are indivisible. In contrast to our setting, in \cite{bei2020fair} the distinction between divisible and indivisible items is given in advance, that is, the algorithm can only divide items that are pre-determined as divisible. In \cite{DBLP:journals/corr/abs-2310-00976}, for each good, some agents may regard it as indivisible, while other agents may regard the good as divisible.
In our setting, only the \emph{number} of divisible items (splittings) is given in advance, but the algorithm is free to choose \emph{which} items to split after receiving the input.

Another paper \cite{segal-halevi2019fair} studies the same problems, but, whereas our paper focuses on asymmetric valuations, they give new results on binary valuations (i.e., each agent values each item as 0 or 1), generalized binary valuations (i.e., each agent values each item as 0 or $x_i$, which can be considered as the price of the item) and negative results on non-degenerate valuations, complementing the results given by \cite{sandomirskiy2019fair}. 

In recent years, researchers have explored the more generalized models of asymmetric agents.
Notable fairness criteria are extended. For instance, \emph{envy-freeness up to one item { \sf (EF1)}} has been generalized to \emph{weighted {\sf EF1} {\sf (WEF1)}} \cite{DBLP:journals/teco/ChakrabortyISZ21}, \emph{proportionality up to one item {\sf (PROP1)}} to \emph{weighted {\sf PROP1} {\sf (WPROP1)}} \cite{DBLP:journals/orl/AzizMS20}, and
\emph{maximin share {\sf (MMS)}} to \emph{weighted {\sf MMS} {\sf (WMMS)}} \cite{DBLP:journals/jair/FarhadiGHLPSSY19}.

Asymmetric agents have also been studied in the context of divisible items concerning proportionality \cite{DBLP:books/daglib/0017730}
\cite{DBLP:journals/talg/CsehF20}
\cite{DBLP:journals/tamm/Su00}
\cite{SEGALHALEVI2019123382}. Our work focuses on mixed settings, with divisible and indivisible items.

\textbf{Splitting in job scheduling:}
There are several variants of job scheduling problems in which it is allowed to break jobs apart.
They can be broadly classified into \emph{preemption} and \emph{splitting}.
In the preemption variants, different parts of a job must be processed at different times. In the three-field notation, they are denoted by "pmtn" and were first studied by \cite{10.2307/2627472}. In the splitting variants, different parts of a job may be processed simultaneously on different machines. They are denoted by "split" and were introduced by \cite{DBLP:journals/dam/XingZ00}.

Various problems have been studied in job scheduling with preemption.
The most closely related to our problem is the generalized multiprocessor scheduling (GMS).
It has two variants. 
In the first variant,
the total number of preemptions is bounded by some fixed integer.
In the second variant, each job $j$ has an associated parameter that bounds the number of times $j$ can be preempted.
In both variants, the goal is to find a schedule that minimizes the makespan subject to the preemption constraints. For identical machines, \cite{ShchepinArticle} proves that with the bound of $n-2$ on the total number of preemptions, the problem is {\sf NP}-hard, whereas \cite{10.2307/2627472} shows a linear-time algorithm with $n-1$ preemptions.
In \Cref{thm:sharing} we prove an analogous result where the bound is on the total number of splittings.
For uniform machines, 
a polynomial algorithm by \cite{DBLP:journals/jacm/GonzalezS78} yields at most $2n-2$ preemptions, and \cite{DBLP:journals/algorithmica/ShachnaiTW05} proved {\sf NP}-hardness for the case where the number of preemption is strictly less than $2n -2$.

In all the works we surveyed, there is no global bound on the number of splitting jobs. As far as we know, bounding the number of splittings or split jobs was not studied before. 

\textbf{Fractional bin-packing: }
Another problem, in which splitting was studied, is the classical bin-packing problem. Bin-packing with fragmented items is first introduced by \cite{MANDAL199891}. They called the problem fragmentable object bin-packing problem and proved that the problem 
is {\sf NP}-hard. 
It is later split into two variants. In the first variant called bin-packing with size-increasing fragmentation (BP-SIF), each item may be fragmented; overhead units are added to the size of every fragment. In the second variant called bin-packing with size-preserving fragmentation (BP-SPF) each item has a size and a cost; fragmenting an item increases its cost but does not change its size. Menakerman and Rom \cite{DBLP:conf/wads/MenakermanR01} show that BP-SIF and BP-SPF are {\sf NP}-hard in the strong sense. Despite the hardness, they present several algorithms and investigate their performance. Their algorithms use classic algorithms for bin-packing, like next-fit and first-fit decreasing, as a base for their algorithms.
Bertazzi et al. \cite{DBLP:journals/dam/BertazziGW19} introduced a variant of BP-SIF with $1-x$ split rule: an item is allowed to be split in only one way according to its size. It is useful for the vehicle routing problem for example. In their paper, they provide the worst case performance bound of the variant.

Further related work is surveyed in \Cref{sec:further-related-work}.

\section{Partition with Interval Target}
\label{sec:interval}

In this section, we analyze the problem \intervalempty{}. 

We assume that there are more items than bins, that is, $m > n$. This assumption is because if $m\leq n$, one can compute all the combinations using brute force (note that the running time is polynomial since $2^m \leq 2^n = O(1)$ since $n$ is a fixed parameter).

We give another definition of the \interval{\asym\decision}{$n,u$}{${\cal X}, {\cal R}$} variant, parameterized by a rational number $v \geq 0$:
\begin{quoting}
\interval{\asym\decision}{$n, v$}{${\cal X}, {\cal R}$}:
\quad
Decide if there exists a partition of ${\cal X}$ among $n$ bins with ratios ${\cal R}$ and relative sums $b_1,\ldots, b_n$ such that $S\leq \max(b_1,\ldots,b_n) \leq (1+v)\cdot S$, where $\sumratios{} := (\sum_i r_i)$ and $S := (\sum_j x_j)/\sumratios{}$.
\end{quoting}
Note that when $u:=vS/M$, both definitions are the same. 

Given an instance of \interval{\asym\decision}{$n,v$}{${\cal X}, {\cal R}$}, 
we say that a partition of ${\cal X}$ is \emph{$v$-feasible} if 
$S\leq \max(b_1,\ldots,b_n) \allowbreak \leq (1+v)\cdot S$, where $b_1,\ldots,b_n$ are the bin relative sums and $S$ is the sum of the items divided by $\sumratios{}$ (the sum of the ratios).
The \interval{\asym\decision}{$n,v$}{${\cal X}, {\cal R}$} problem is to decide whether a $v$-feasible partition exists. 

\begin{definition}
Given an instance of \interval{\asym\decision}{$n,v$}{${\cal X}, {\cal R}$}, a rational number $\epsilon>0$, and a partition of ${\cal X}$ among $n$ bins with ratios ${\cal R}$,
an \emph{almost-full bin} is a bin with relative sum larger than $(1+v)\cdot S /(1+\epsilon)$. 
\end{definition}

A known {\sf FPTAS} for the \partition{\asym\minmax}{$n$}{${\cal X}, {\cal R}$} problem \citep{woeginger2000does} gives us valuable information since we can easily verify that if the output of this {\sf FPTAS} is smaller than $(1+v)\cdot S$ then in any $n$-way partition of ${\cal X}$, at least one bin is almost-full. To gain more information on the instance, we apply an {\sf FPTAS} for a constrained variant of \partitionempty{}, with a {\it Critical Coordinate}.
For an integer $n\geq 2$, and a rational number $v>0$, we define the following problem:
\begin{quoting}
\label{pb:min_b2}
\partition{\asym\minmax}{$n,v,i$}{${\cal X},{\cal R}$}: Minimize $\max(b_1, \dots, b_{i-1}, b_{i+1}, \dots, b_n)$ subject to $b_i \leq (1+v) \cdot S$ where $b_1, \dots, b_n$ are bin relative sums in an $n$-way partition of ${\cal X}$.
\end{quoting}
The general technique developed by \cite{woeginger2000does} for converting a dynamic program to an {\sf FPTAS} can be used to design an {\sf FPTAS} for \partition{\asym\minmax}{$n,v,i$}{${\cal X}, {\cal R}$}; we give the details in 
the supplementary information.
We denote by {\sf FPTAS}(\partition{\asym\minmax}{$n,v,i$}{${\cal X}, {\cal R}$}, $\epsilon$) the largest bin relative sum in the solution obtained by the {\sf FPTAS}.
\begin{lemma}
\label{lem:twoalmostfull}
For any $n\geq 2$, $v>0$, $\epsilon>0$,
if, for all $i\in[n]$, 
{\sf FPTAS}(\partition{\asym\minmax}{$n,v,i$}{${\cal X}, {\cal R}$}, $\epsilon) > (1+v)\cdot S$,
then in any $v$-feasible $n$-way partition of ${\cal X}$, at least \emph{two} bins are almost-full.
\end{lemma}
\begin{proof}
Suppose by contradiction that there exists a $v$-feasible partition of ${\cal X}$ with at most one almost-full bin.
Let $i$ be the index of the bin with the largest sum in that partition.
Since bin $i$ has the largest sum, if there is one almost-full bin, it must be bin $i$. 
Hence, bins that are not $i$ are not almost-full, so $\max(b_1, \dots, b_{i-1}, b_{i+1}, \dots, b_n)\leq (1+v)\cdot S/(1+\epsilon)$.
Moreover, $b_i\leq (1+v)\cdot S$ since the partition is $v$-feasible.
Therefore, {\sf FPTAS}(\partition{\asym\minmax}{$n,v,i$}{${\cal X}, {\cal R}$}, $\epsilon) \leq (1+v)\cdot S$ by the definition of {\sf FPTAS}. This contradicts the lemma assumption.
\end{proof}

\subsection{\interval{\asym\decision}{$n,u$}{${\cal X}, {\cal R}$}: an algorithm for $n = 2$ and $u > 0$}

The case when $n = 2$ and $u \ge 1$ is solved by the cut-the-line algorithm combined with \Cref{thm:from-interval-to-shared}. So, \interval{\asym\decision}{$n=2,u\geq 1$}{${\cal X}, {\cal R}$} can be solved in polynomial time, and it is {\sf NP}-complete if $u=0$. The following algorithm completes the picture (when $0< u < 1$) by giving a polynomial time algorithm for the case where $n=2$, and $u > 0$.

Using \Cref{lem:twoalmostfull}, we derive a complete algorithm for \interval{\asym\decision}{$n=2,v$}{${\cal X}, {\cal R}$}.

\begin{algorithm}[!h] \caption{\qquad\interval{\asym\decision}{$n=2,v$}{${\cal X}, {\cal R}$}} \label{alg:UniInterval}
    \begin{algorithmic}[1]
        \For {$i \in \{1, 2\}$} 
            \State 
            Run the {\sf FPTAS} for \partition{\asym\minmax}{$n,v,i$}{${\cal X}, {\cal R}$} with $\epsilon= v/2$.
            \State 
               Let $b_j$ be the returned value, that is, the relative sum of the bin $j \neq i$.
            \State {\sf If} $b_j \le (1+v)\cdot S$, return ``yes''
        \EndFor
        \State Return ``no''
    \end{algorithmic}
\end{algorithm}	

\begin{theorem}\label{thm:intervalv}
For any rational $v>0$,
\Cref{alg:UniInterval} 
solves the 
\interval{\asym\decision}{$n=2,v$}{${\cal X}, {\cal R}$} problem in time $O(\text{poly}(m,\log{(\sumratios{}S)},1/v))$, where $m$ is the number of items in ${\cal X}$ and $S = (\sum_j x_j)/\sumratios{}$.
\end{theorem}

\begin{proof}
The run-time of Algorithm~\ref{alg:UniInterval} is dominated by the run-time of the {\sf FPTAS} for \partition{\asym\minmax}{$n=2,v,i$}{${\cal X}, {\cal R}$}, which is $O(\text{poly}(m, \log{(\sumratios{}S)}, 1/\epsilon)) = O(\text{poly}(m, \log{(\sumratios{}S)}, 1/v))$ (we show in \Cref{sec:running-time1} that the exact run-time is $O(\frac{m}{v} \log{(\sumratios{}S)})$). 
It remains to prove that \Cref{alg:UniInterval} indeed solves \interval{\asym\decision}{$n=2,v$}{${\cal X}, {\cal R}$} correctly.

If $b_j$, the returned bin relative sum of {\sf FPTAS}(\partition{\asym\minmax}{$n=2,v,i$}{${\cal X}, {\cal R}$}, $\epsilon= v/2$), is at most $(1+v)\cdot S$, then the partition found by the {\sf FPTAS} is  $v$-feasible, so \Cref{alg:UniInterval} answers ``yes'' correctly.
Otherwise, by \Cref{lem:twoalmostfull}, in any $v$-feasible partition of ${\cal X}$ into two bins, both bins are almost-full. 
This means that, in any $v$-feasible partition,
both bin relative sums $b_1$ and $b_2$ are larger than 
$(1+v)\cdot S /(1+\epsilon)$,
which is larger than $S$ since $\epsilon=v/2$.
So $r_1b_1+r_2b_2> (r_1+r_2) S = \sumratios{} S$. But this is impossible since the sum of the items is $\sumratios{}S$ by assumption. Hence, no $v$-feasible partition exists, 
and \Cref{alg:UniInterval} answers ``no'' correctly. \footnote{
Instead of an  {\sf FPTAS} for \partition{\asym\minmax}{$n=2,v,i$}{${\cal X}, {\cal R}$},
we could use an {\sf FPTAS} for the {\sf Subset Sum} problem \cite{kellerer2003efficient}, using the same arguments.
The critical coordinate is not needed in the {\sf Subset Sum FPTAS}, since the output is always smaller than the target.
We prefer to use the {\sf FPTAS} for \partition{\asym\minmax}{$n=2,v,i$}{${\cal X}, {\cal R}$}, since it is based on the general technique of \cite{woeginger2000does}, that we use later for solving other problems.
}
\end{proof}

\begin{corollary}\label{cor:u-instead-of-v}
For any rational $u>0$, 
\Cref{alg:UniInterval} 
solves the 
\interval{\asym\decision}{$n=2,u$}{${\cal X}, {\cal R}$} problem in time $O(\text{poly}(m,\log{(\sumratios{}S)},1/u))$.
\end{corollary}
\begin{proof}
    For any rational $u>0$, let $v := uM/S$, that is $v>0$. \Cref{alg:UniInterval} solves the \interval{\asym\decision}{$n=2,v:=uM/S$}{${\cal X}, {\cal R}$} problem in time $O(\text{poly}(m,\log{(\sumratios{}S)},S/uM)))$, where $m$ is the number of items in ${\cal X}$ and $S = (\sum_i x_i)/\sumratios{}$ is the perfect bin relative sum. Since $S \leq mM$, the algorithm runs in time $O(\text{poly}(m,\log{(\sumratios{}S)},1/u))$ for any $u>0$.
\end{proof}

\begin{remark}
The reader may wonder why we cannot use a similar algorithm for $n\geq 3$.  
For example, we could have considered a variant of \partition{\asym\minmax}{$n,v,i$}{${\cal X}, {\cal R}$} with two critical coordinates:
\begin{quoting}
Minimize $\max(b_3,\dots,b_n)$
subject to $b_1 \leq (1+v)\cdot S$  and $b_2 \leq (1+v)\cdot S$,
where $b_1,b_2,b_3,\dots,b_n$ are bin relative sums in an $n$-way partition of ${\cal X}$.
\end{quoting}
If the {\sf FPTAS} for this problem does not find a $t$-feasible partition, then any $t$-feasible partition must have at least three almost-full bins. 
Since not all bins can be almost-full, one could have concluded that there is no $t$-feasible partition into $n=3$ bins.

Unfortunately, the problem with two critical coordinates probably does not have an {\sf FPTAS} even for $n=3$, even for identical bins,
since it is equivalent to 
the {\sf Multiple Subset Sum} problem, which does not have an {\sf FPTAS} unless {\sf P}={\sf NP} 
\citep{caprara2000multiple}.
In the next subsection, we handle the case $n\geq 3$ differently.
\end{remark}

\subsection{\interval{\asym\decision}{$n,u$}{${\cal X}, {\cal R}$}: an algorithm for $n \geq 3$ and $u \ge n-2$}
\label{sec:dinterval}

The case when $u \ge n-1$ is solved by the cut-the-line algorithm combined with \Cref{thm:from-interval-to-shared}. Here, we prove a more general case where $u \ge n-2$.
We recall that given an instance of \interval{\asym\decision}{$n,u$}{${\cal X}, {\cal R}$},
where the sum of the items is $\sumratios{} \cdot S$ and the largest item is $\sumratios{} \cdot M$, 
where $S, M\in \mathbb{Q}$, we say that a partition of ${\cal X}$ is \emph{$u$-possible} if 
$S \leq \max(b_1,\ldots,b_n) \leq S + u \cdot M$, where $b_1,\ldots,b_n$ are the bin relative sums.
The \interval{\asym\decision}{$n,u$}{${\cal X}, {\cal R}$} problem is to decide whether a $u$-possible partition exists.

Similarly to \Cref{alg:UniInterval},  
we would like to gain information on the instance by running {\sf FPTAS}$($\partition{\asym\minmax}{$n, v, i$}{${\cal X}, {\cal R}$}, $\epsilon)$ for every $i \in [n]$. If the {\sf FPTAS} fails to find a $v$-feasible partition, by \Cref{lem:twoalmostfull}, any $v$-feasible partition must have at least two almost-full machines.
We prove some existential results about partitions with two or more almost-full bins.

\subsubsection{Structure of partitions with two or more almost-full bins}

We distinguish between big, medium, and small items defined as follows. 
A \emph{small-item} is an item with size smaller than $2 \epsilon \sumratios{}S$;
a \emph{big-item} is an item with size greater than $(\frac{v}{n-2} - 2 \epsilon) \sumratios{}S$.
All other items are called \emph{medium items}.
We denote by $m_b$ the number of big-items in the input.
Our main structural Lemma is the following.

\begin{lemma}\label{lem:structure}
Suppose that $u\geq n-2$, $v=uM/S$ (and $v<1$)\footnote{The case where $v \geq 1$ is a highly degenerate one. We can solve it for any $n$ using the cut-the-line algorithm combined with \Cref{thm:from-interval-to-shared} since it gives a solution for $u \geq n-1$, so it is a solution for $v \geq M(n-1)/S \geq M(n-1)/mM$ which is smaller than 1 (since $m>n$).},
$\epsilon = \frac{v}{4 n^2 m^2}$ and the following properties hold.
\begin{adjustwidth}{5mm}{}
	\begin{enumerate}[label=(\arabic*)]
		\item 
		\label{prop:no-1-almost-full} 
		There is no $v$-feasible partition with at most $1$ almost-full bin;
		
		\item  
		\label{prop:no-2-almost-full} 
		There is a $v$-feasible partition with at least $2$ almost-full bins.
\end{enumerate}
	\end{adjustwidth}
	Then, there is a $v$-feasible partition with the following properties.

	\begin{enumerate}[label=(\alph*)]
		\item 
		\label{prop:2-almost-full} 
		Exactly two bins (w.l.o.g. bins 1 and 2) are almost-full.
		
		\item 
		\label{prop:not-almost-full} 
		The sum of every not-almost-full bin $i \in\{3,\ldots,n\}$ satisfies
        
    {\small
    \begin{align*}
            \left(r_i + \left(r_i - \frac{\sumratios{}}{n-2} \right)v - r_i 2\epsilon \right) \cdot S
    		\leq
    		~ r_i b_i ~
    		\leq
            \left(r_i + \left(r_ i - \frac{\sumratios{}}{n-2} \right) v + (\sumratios{} - r_i) 2\epsilon\right)\cdot S.
		\end{align*}
    }
		
		\item 
		\label{prop:big-items} 
		Every item in an almost-full bin is a big-item.

		\item 
		\label{prop:split-big-small-items} 
		Every item in a not-almost-full bin is either a small-item, or a big-item larger or equal to every item in bins 1,2.

    \item \label{prop:litems}
    \label{prop:l+1items}
    Denote by $m_i$ the number of big-items in bin $i$. 
    Then,
    \begin{align*}
    	m_i = 
    	\lceil r_i \cdot (m_b + n - 3) / \sumratios{} + 1/\sumratios{}\rceil && \text{for any almost-full bin $i\in\{1,2\}$;}
    \end{align*}
	\end{enumerate}
\end{lemma}

Note that properties \ref{prop:big-items} and \ref{prop:split-big-small-items} together imply that ${\cal X}$ contains no medium items. In other words, if an instance contains some medium items, then the lemma assumptions cannot hold, that is, either there is a $v$-feasible partition with at most one almost-full bin, or there is no $v$-feasible partition at all.


As an example of this situation, consider an instance with $7$ items, all of which have size $1$, with $n=5$ identical bins and $u=3$. Then, there is a $u$-possible partition with two almost-full bins: $(1,1), (1,1), (1), (1), (1)$, and no $u$-possible partition with 1 or 0 almost-full bins. See \Cref{sec:example-lem} for details. 
A full proof (\Cref{sec:appendix-lem-struct}) of the Lemma appears in the appendix, here we provide a sketch proof.

\begin{proof}[Proof Sketch]
We start with an arbitrary $v$-feasible partition with some $q\geq 2$ almost-full bins $1,\ldots, q$, and convert it using a sequence of transformations to another $v$-feasible partition satisfying properties \ref{prop:2-almost-full}--\ref{prop:litems}, as explained below.
Note that the transformations are not part of our algorithm and are only used to prove the lemma. 
First, we note that there must be at least one bin that is not almost-full, since the sum of an almost-full bin is larger than $S$ whereas the sum of all $n$ bins is $\sumratios{}\cdot S$.

\textbf{For \ref{prop:2-almost-full}}, if there are $q\geq 3$ almost-full bins, we move any item from one of the almost-full bins $3,\ldots,q$ to some not-almost-full bin. We prove that, as long as $q\geq 3$, the target bin remains not-almost-full. This transformation is repeated 
until $q=2$ and only bins 1 and 2 remain almost-full.

\textbf{For \ref{prop:not-almost-full}}, for the lower bound, if there is $i \in \{3,\ldots,n\}$ for which $b_i$ is smaller than the lower bound, we move an item from bins $1,2$ to bin $i$. We prove that bin $i$ remains not-almost-full, so by assumption \ref{prop:no-1-almost-full}, bins 1, 2 must remain almost-full. We repeat until $b_i$ satisfies the lower bound. Once all bins satisfy the lower bound, we prove that the upper bound is satisfied too.

\textbf{For \ref{prop:big-items}}, if bin 1 or 2 contains an item that is not big, we move it to some bin $i \in \{3,\ldots,n\}$. 
We prove that bin $i$ remains not-almost-full, so by  assumption \ref{prop:no-1-almost-full}, bins 1, 2 must remain almost-full.
We repeat until bins 1 and 2 contain only big-items.

\textbf{For \ref{prop:split-big-small-items}}, if some bin $i\in\{3,\ldots,n\}$ contains an item $x_j$ that is bigger than $2\sumratios{}S\epsilon$ and smaller than any item $x_{j'}$ in bin 1 or bin 2, we exchange $x_j$ and $x_{j'}$. We prove that, after the exchange, bin $i$ remains not-almost-full, so bins 1, 2 must remain almost-full. We repeat until bins 1,2 contain only the smallest big-items.
Note that transformations \ref{prop:not-almost-full}, \ref{prop:big-items}, \ref{prop:split-big-small-items} increase the sum in the not-almost-full bins $3,\ldots,n$, so the process must end.

\textbf{For \ref{prop:litems}}, we show an upper bound on the almost-full bin sums depending on the number of big-items. We divide the upper bound by the largest item, to get a lower bound on the number of big-items in each bin. We show that the difference between the upper bound and the lower bound for the almost-full bins is less than one, therefore there is at most one possible value.
\end{proof}

\subsubsection{Algorithm sketch} \label{sub:bins12}

The algorithm starts by running 
{\sf FPTAS}$($\partition{\asym\minmax}{$n, v, i$}{${\cal X}, {\cal R}$}, $\epsilon)$.
If the {\sf FPTAS} find a $v$-feasible partition, we return ``yes''.
Otherwise, by \Cref{lem:twoalmostfull}, any $v$-feasible partition must have at least two almost-full bins. So, if a $v$-feasible partition exists, then there exists a $v$-feasible partition satisfying all properties of \Cref{lem:structure}. We can find such a partition (if it exists) in the following way. We loop over all pairs of bins. For simplicity, let us call the current pair of bins: 1 and 2. For each pair, we check whether there exists an allocation in which these bins are the almost-full bins in two steps:

\begin{itemize}
\item \textbf{For bins $1,2$:} 
Let  $B \subseteq {\cal X}$ contain all big-items in ${\cal X}$, $m_b := |B|$, and $m_1, m_2$ given by \Cref{lem:structure}\ref{prop:litems}. Find a $v$-feasible partition of the $m_1 + m_2$ smallest items in $B$ into two bins with $m_i$ items in each bin $i \in \{1,2\}$.
\item \textbf{For bins $3,\ldots,n$:}
Find a $v$-feasible partition of the remaining items in ${\cal X}$ into $n-2$ bins. 
\end{itemize}
For bins $3,\ldots,n$, we use the {\sf FPTAS} for the problem \partition{\asym\minmax}{$n=n-2$}{${\cal X}, {\cal R'}$}, where ${\cal R'}=(r_3,\ldots,r_n)$.
If it returns a $v$-feasible partition, we are done. 
Otherwise, by {\sf FPTAS} definition, every partition into $(n-2)$ bins must have at least one almost-full bin. But by \Cref{lem:structure}\ref{prop:2-almost-full}, all bins $3,\ldots,n$ are not almost-full which is a contradiction. Therefore, if the {\sf FPTAS} does not find a $v$-feasible partition, we answer ``no''. 
Bins 1 and 2 require a more complicated algorithm\footnote{Using a simple {\sf FPTAS} for {\sf Subset Sum} to solve our problem is not possible since the sum of the items may equal $(r_1 + r_2)(S+uM)$ and our goal is to partition the items among two bins of capacity $(r_1 + r_2)(S+uM)$, which means that no approximation is allowed. Therefore, we need to pre-process the items in some way. The technique we use is to inverse the item sizes.
Even after inverting the items, we need an additional cardinality constraint enforcing that there are exactly $m_i$ items in bin $i$.}
that is explained in the following section.

\subsubsection{Algorithm for bins 1 and 2}\label{sec:appendix-algo-12}\label{app:big-items-alg}

We use the following notation:
\begin{itemize}
\item [--] $B_{1:2}$ is the set of $m_1 + m_2$ smallest items in $B$ 
(where $B$ is the set of big-items in ${\cal X}$).
\item [--]
$S_{1:2}:= (\sum_{j\in B_{1:2}} x_j)/(r_1+r_2)$, so that the sum of items in $B_{1:2}$ is $(r_1 + r_2) S_{1:2}$.
\item [--]
$M_{1:2}:= (\max_{j\in B_{1:2}} x_j)/\sumratios{}$, so that the largest item in $B_{1:2}$ is $\sumratios{} M_{1:2}$.
\end{itemize}

Construct a new set, $\overline{B_{1:2}}$, by replacing each item $x \in B_{1:2}$ by its ``inverse'', defined by $\overline{x} := \sumratios{} M_{1:2} - x$.
Since
all items in $B_{1:2}$ are big-items,
every inverse are between $0$ and $2 \sumratios{} S \epsilon$ since $\sumratios{}M - \left(\frac{v}{n-2} - 2 \epsilon\right) \sumratios{}S \leq 2 \sumratios{} S \epsilon$ and $\frac{v\sumratios{}S}{n-2} \geq \sumratios{}M \geq \sumratios{} M_{1:2}$.
Let $\overline{S_{1:2}}$ be the sum of inverses divided by $r_1+r_2$. 

Given a $v$-feasible partition of $B_{1:2}$ with sums $b_1, b_2$, with $m_1$ items in bin $1$ and $m_2$ items in bin $2$, denote
the relative sums of the corresponding partition of $\overline{B_{1:2}}$ by $\overline{b_1}$ and  $\overline{b_2}$, respectively. Since both bins $1$ and $2$ contain $m_1, m_2$ items,
\begin{align}
	r_1\overline{b_1} 
	&
	= m_1 \cdot \sumratios{}M_{1:2} - r_1 b_1 
	\nonumber 
	\text{ and, }
	r_2\overline{b_2} 
	= m_2 \cdot \sumratios{}M_{1:2} - r_2 b_2 \nonumber \\
	(r_1 + r_2) \overline{S_{1:2}} 
	& 
	=  r_1\overline{b_1}  +  r_2\overline{b_2} = (m_1 + m_2) \cdot \sumratios{}M_{1:2} - (r_1 + r_2) S_{1:2}  \nonumber \\
	\overline{S_{1:2}} 
	&
	=  \frac{m_1 + m_2}{r_1 + r_2} \cdot \sumratios{}M_{1:2} - S_{1:2} 
	\nonumber 
	\leq 
	\left(\frac{m_b+n-3}{\sumratios{}} + \frac{2(\sumratios{}-n+1)}{(r_1+r_2)\sumratios{}}\right)\cdot \sumratios{}M_{1:2} - S_{1:2} 
	&& 
	\nonumber 
    \\
    & 
    \nonumber 
	\leq \left(\frac{m_b+n-3}{\sumratios{}}+ \frac{2}{(r_1+r_2)\sumratios{}}\right)\cdot \sumratios{}M_{1:2} - S_{1:2} 
	\\
	& 
	= 
	\left(m_b+n-3 + \frac{2}{r_1+r_2}\right)\cdot M_{1:2} - S_{1:2}.
	~~~~~~
	\text{(by prop \ref{prop:litems})} 
	\label{eq:s12}
\end{align}
Now,
\begin{align*}
	b_1 \leq (1+v)S 
	\iff 
	&
	r_1 \overline{b_1} \geq  m_1 \sumratios{}M_{1:2} - r_1 (1+v)S
	\\
	\iff 
	& 
	r_2 \overline{b_2} \leq  (r_1 + r_2) \overline{S_{1:2}} - m_1 \sumratios{}M_{1:2} + r_1 (1+v)S 
	\iff
	\\
	& 
	\overline{b_2}  \leq  \overline{S_{1:2}} + \frac{r_1 \overline{S_{1:2}} - m_1 \sumratios{} M_{1:2}}{r_2}  + \frac{r_1}{r_2}(1+v)S
	\\
	& 
	~~~ =  \overline{S_{1:2}} +  \frac{1}{r_2} \cdot (r_1 \frac{m_1 + m_2}{r_1 + r_2} \cdot \sumratios{} M_{1:2} - r_1 S_{1:2} - m_1 \sumratios{} M_{1:2} + r_1 (1+v)S)
	\\
	& 
	~~~ = \overline{S_{1:2}} + \frac{1}{r_2} \cdot  ((r_1 \frac{m_1 + m_2}{r_1 + r_2} - m_1 ) \sumratios{} M_{1:2} + r_1 (S+vS-S_{1:2}))
	\\
	& 
	~~~ = \overline{S_{1:2}} + \frac{1}{r_2} \cdot  (\frac{r_1 m_1 + r_1 m_2 - r_1 m_1 - r_2 m_1}{r_1 + r_2} \sumratios{} M_{1:2} + r_1 (S+vS-S_{1:2}))
	\\
	& 
	~~~ = \overline{S_{1:2}} + \frac{1}{r_2} \cdot  (\frac{r_1 m_2 - r_2 m_1}{r_1 + r_2} \sumratios{} M_{1:2} + r_1 (S+vS-S_{1:2}))
	\\
	& 
	~~~ = \overline{S_{1:2}} \cdot (1+ \frac{1}{r_2 \overline{S_{1:2}}} \cdot  (\frac{r_1 m_2 - r_2 m_1}{r_1 + r_2} \sumratios{} M_{1:2} + r_1 (S+vS-S_{1:2}))),
\end{align*}
and similarly, 
$b_2 \leq  (1+v) S$ 
holds iff
$\overline{b_1} \leq  \overline{S_{1:2}} 
\cdot (1 + \frac{1}{r_1 \overline{S_{1:2}}} \cdot  (\frac{r_2 m_1 - r_1 m_2}{r_2 + r_1} \sumratios{} M_{1:2} + r_2 (S+vS-S_{1:2}))$.

So the problem of finding a $v$-feasible partition of $B_{1:2}$ with $m_1$ items in bin $1$ and $m_2$ items in bin $2$ is equivalent to the following problem. For two fixed rationals $v_1, v_2$ and a set of ratios ${\cal R}=(r_1,r_2)$ 
such that $r_1v_1+r_2v_2 >0$,
\begin{quote}
	\interval{\asym\decision}{$n=2, v_1, v_2$}{${\cal X}, {\cal R}, m_1, m_2$}: \quad
	Decide if there exists a partition of ${\cal X}$ among two bins in which the relative sum of every bin $i$ is at most $(1+v_i) S$, where every bin $i$ contains exactly $m_i$ items.
\end{quote}
In \Cref{sec:vv-variant}, we design an algorithm similar to \Cref{alg:UniInterval}, that solves \interval{\asym\decision}{$n=2, \overline{v_1}, \overline{v_2}$}{${\cal X}, {\cal R}, m_1, m_2$} in $O(\text{poly}(m, (r_1 + r_2)/(r_1\overline{v_1}+r_2\overline{v_2}), \log{(\sumratios{}S)}))$ time.

We now show that $(r_1 + r_2)/(r_1\overline{v_1}+r_2\overline{v_2})$ is polynomial in the problem size.
We have, 
\begin{align*}
	\overline{v_1} =  \frac{1}{r_1 \overline{S_{1:2}}}\cdot \left[ \frac{r_2 m_1 - r_1 m_2}{r_2 + r_1} \sumratios{}M_{1:2} + r_2 (S+vS-S_{1:2}) \right],
	\\
	\overline{v_2} = \frac{1}{r_2 \overline{S_{1:2}}}\cdot \left[ \frac{r_1 m_2 - r_2 m_1}{r_1 + r_2} \sumratios{}M_{1:2} + r_1 (S+vS-S_{1:2}) \right].
\end{align*}
By adding the above two equations, we get
\begin{align*}
	\overline{S_{1:2}} \cdot (r_1\overline{v_1}+r_2\overline{v_2}) & := \frac{r_1 m_2 - r_2 m_1}{r_1 + r_2} \sumratios{}M_{1:2} + r_1 (S+vS-S_{1:2}) \\
	& + \frac{r_2 m_1 - r_1 m_2}{r_2 + r_1} \sumratios{}M_{1:2} + r_2 (S+vS-S_{1:2}) \\
	& = (r_1 + r_2)(S+vS-S_{1:2}), 
\end{align*}
so 
$(r_1 + r_2)/(r_1\overline{v_1}+r_2\overline{v_2})=
\overline{S_{1:2}}/(S+vS-S_{1:2})$,
so the run-time is polynomial in:
$
V := \frac{
	\overline{S_{1:2}}
}{
	S+vS-S_{1:2}
}
$

The numerator of $V$ is upper-bounded by \eqref{eq:s12}.
We now give a lower bound to the denominator.

By \Cref{lem:structure}\ref{prop:split-big-small-items}, in the not-almost-full bins $3, \dots, n$, there are $m_3, \dots, m_n$ big-items that are at least as large as $\sumratios{} M_{1:2}$ (in addition to some small-items). Therefore, the sum $\sumratios{} S$ of all items satisfies
$\sumratios{} S \geq (r_1 + r_2) S_{1:2} + \sum_{i=3}^n m_i \cdot \sumratios{}M_{1:2}$, which implies that 
\begin{align*}
	S - S_{1:2} & \geq \frac{(r_1 + r_2) S_{1:2}}{\sumratios{}} + \sum_{i=3}^n m_i M_{1:2} - S_{1:2} 
	\\
	& 
	= 
	\frac{(r_1 + r_2) S_{1:2} - \sumratios{}  S_{1:2}}{\sumratios{}} + \sum_{i=3}^n m_i M_{1:2} 
	\\
	& 
	= 
	\frac{(r_1 + r_2 - \sumratios{}) S_{1:2}}{\sumratios{}} + \sum_{i=3}^n m_i M_{1:2} 
	\\
	&
	= - 
	\frac{\sum_{i=3}^n r_i S_{1:2}}{\sumratios{}} + \sum_{i=3}^n m_i M_{1:2} 
	\\
	& 
	=  
	\frac{\sum_{i=3}^n m_i}{\sumratios{}}\sumratios{} M_{1:2} - \frac{\sum_{i=3}^n r_i}{\sumratios{}} S_{1:2}
\end{align*}
Therefore,
\begin{align*}
	S-S_{1:2}+vS & \geq \frac{\sum_{i=3}^n m_i}{\sumratios{}}\sumratios{} M_{1:2} - \frac{\sum_{i=3}^n r_i}{\sumratios{}} S_{1:2} + vS \\
	& \geq \frac{\sum_{i=3}^n m_i + (n-2)}{\sumratios{}}\sumratios{} M_{1:2} - \frac{\sum_{i=3}^n r_i}{\sumratios{}} S_{1:2} \\
	& ~~~~~~~~~~~~~~~~~~~~~~~~~~~~~~~~~~~ \text{(since $u\geq n-2$, $v S = u M  \geq (n-2) M_{1:2}$)} \\
	& = \frac{\sum_{i=3}^n (m_i + 1)}{\sumratios{}}\sumratios{} M_{1:2} - \frac{\sum_{i=3}^n r_i}{\sumratios{}} S_{1:2}
	\\
	& 
	\geq \frac{\sum_{i=3}^n (\frac{r_i}{\sumratios{}}(m_b+n-3) + 1/\sumratios{})}{\sumratios{}}\sumratios{} M_{1:2} - \frac{\sum_{i=3}^n r_i}{\sumratios{}} S_{1:2}
	~ \text{(by property \ref{prop:litems})}
	\\
	& = \frac{(m_b+n-3)\cdot \sum_{i=3}^n r_i + n-2}{\sumratios{}} M_{1:2} - \frac{\sum_{i=3}^n r_i}{\sumratios{}} S_{1:2}
    \\
	& \geq \frac{(m_b+n-3)\cdot \sum_{i=3}^n r_i + 1}{\sumratios{}} M_{1:2} - \frac{\sum_{i=3}^n r_i}{\sumratios{}} S_{1:2}
	\\
	&
	= 
	\frac{\sum_{i=3}^n r_i}{\sumratios{}}\cdot \left(
	((m_b+n-3) + \sumratios{}/(r_3 + \dots + r_n)) M_{1:2} - S_{1:2}
	\right).
\end{align*}
Note that $\sumratios{}/(r_3 + \dots + r_n) > 1$ while $2/(r_1+r_2) \leq 1$, therefore, $\sumratios{}/(r_3 + \dots + r_n) > 2/(r_1+r_2)$, so:
\begin{align*}
	S-S_{1:2}+vS & > \frac{\sum_{i=3}^n r_i}{\sumratios{}}\cdot \left(
	(m_b+n-3 + 2/(r_1+r_2)) M_{1:2} - S_{1:2}
	\right).
\end{align*}
Now, we substitute the numerator and the lower bound for the denominator:
\begin{align*}
	V \leq &
	\frac{
		\left(m_b+n-3 + \frac{2}{r_1+r_2}\right)\cdot M_{1:2} - S_{1:2}
	}{
		\frac{\sum_{i=3}^n r_i}{\sumratios{}}\cdot\left(\left(m_b+n-3 + \frac{2}{r_1+r_2}\right)\cdot M_{1:2} - S_{1:2}\right)
	}   =   \frac{\sumratios{}}{\sum_{i=3}^n r_i}.
\end{align*}
By assumption, bins 1 and 2 are almost-full, so we have: 
\begin{align*}
    (1+v-2 \epsilon)(r_1+r_2)S & < \sumratios{}S \\
    \iff 
    (r_1+r_2) & < \sumratios{} / ((1+v-2 \epsilon)).
\end{align*}
Substituting in the above upper bound for $V$ gives:
\begin{align*}
V &\leq \sumratios{} / (\sumratios{} - (r_1 + r_2))
\\
&<
\sumratios{} / (\sumratios{} - \sumratios{} / ((1+v-2 \epsilon)))
\\
&=
(1+v-2 \epsilon) / ((1+v-2 \epsilon) - 1)
\\
&=
1 + 1 / (v-2 \epsilon)
\\
&<
1 + 2 / v
\\
&\leq 
1 + 2 S (n-2) / M
\\
& \in 
O(m).
\end{align*}
So $V \in O(m)$ and the sub-problem can be decided in time $O(\text{poly}(m, \allowbreak \log{(\sumratios{}S)}))$.



\subsubsection{Complete algorithm}
We are now ready to present the complete algorithm for \interval{\asym\decision}{$n,u$}{${\cal X}, {\cal R}$}, presented in Algorithm~\ref{alg:tfeasible}.

\begin{algorithm}
\caption{\interval{\asym\decision}{$n,u$}{${\cal X}, {\cal R}$} (complete algorithm)
}{
}\label{alg:tfeasible}
\begin{algorithmic}[1]
\State
$v \longleftarrow u M / S$ and
$\epsilon \longleftarrow v/(4n^2m^2)$.

\For {$i \in \{1, \dots, n\}$} 

    \State \label{item:fptas1}
    {\sf If} {\sf FPTAS}(\partition{\asym\minmax}{$n,v,i$}{${\cal X}, {\cal R}$}, $\epsilon) \leq (1+v)\cdot S$, {\sf return} ``yes''. 

\EndFor   

\State $B \longleftarrow $ 
$\big\{x_i \in {\cal X} \mid x_i
> \sumratios{} S (\frac{v}{n-2} - 2 \epsilon)\big\}$
\Comment{{\scriptsize big-items}}

\State $m_b \longleftarrow |B|$

\For {every pair $i,j$ of bins}
\Comment{{\scriptsize $i$ and $j$ are the almost-full bins}}

\State $m_1 \longleftarrow \lceil r_i \cdot (m_b+n-3) / \sumratios{} + 1/\sumratios{}\rceil$, $m_2 \longleftarrow r_j \cdot (m_b+n-3) / \sumratios{} + 1/\sumratios{}\rceil$.
\Comment{{\scriptsize by property \ref{prop:l+1items}}}

\State {\sf If} $m_1$ or $m_2$ is greater than $r_i \cdot (m_b+n-3) / \sumratios{} + 1/\sumratios{} + (\sumratios{}-n)/\sumratios{}$ or $r_j \cdot (m_b+n-3) / \sumratios{} + 1/\sumratios{} + (\sumratios{}-n)/\sumratios{}$ respectively, skip to next iteration. 
        
\State $B_{1:2} \longleftarrow $ the $m_1 + m_2$ smallest items in $B$.
\Comment{{\scriptsize break ties arbitrarily}}

\State $B_{3..n} \longleftarrow {\cal X} \setminus B_{1:2}$.
\Comment{{\scriptsize 
all the remaining
big-items
and $m - m_b$ small-items}}

\State\label{item:b3n}$(b_3,\ldots,b_n) \longleftarrow$ {\sf FPTAS}
(\partition{\asym\minmax}{$n-2$}{$B_{3..n}, {\cal R}_{3..n}$)}, $\epsilon$)
\Comment{{\scriptsize an $(n-2)-$way partition of $B_{3..n}$}}

\State {\sf If} $\max(b_3,\ldots,b_n) > (1+v)S$, skip to next iteration.

\State Look for a feasible partition for bins $i,j$ as explained in \Cref{app:big-items-alg}.

\State {\sf If} a feasible partition is found, {\sf return} ``yes''.
        
\EndFor
\State {\sf Return} ``no''.
\end{algorithmic}
\end{algorithm}

\begin{theorem}
\label{thm:n>=3}
For any fixed integer $n \geq 3$
and rational number $u\ge n-2$,
\Cref{alg:tfeasible} solves
\interval{\asym\decision}{$n,u$}{${\cal X}, {\cal R}$}
in 
$O(\text{poly}(m, \log{(\sumratios{} S)}))$ time,
where $m$ is the number of items in ${\cal X}$, and $\sumratios{} S$ is the sum of the items.
\end{theorem}

\begin{proof}
If Algorithm~\ref{alg:tfeasible} answers ``yes'', then clearly a $v$-feasible partition exists.
To complete the correctness proof, we have to show that the opposite is true as well.

Suppose there exists a $v$-feasible partition.
If the partition has at most one almost-full bin (say bin $i$), 
then by \Cref{lem:twoalmostfull}, it is found by 
the {\sf FPTAS} in step \ref{item:fptas1} in iteration $i$ of the for loop.
Otherwise, the partition must have at least two almost-full bins, and there exists a $v$-feasible partition satisfying the properties of \Cref{lem:structure}. 

The for loop checks all possible pairs of bins, and checks whether these two bins can be the two almost-full bins in a $v$-feasible partition.
By \ref{prop:l+1items}, the number of big-items in bins 1 and 2 is equal to $r_1 \cdot (m_b+n-2) / \sumratios{}$ and $r_2 \cdot (m_b+n-2) / \sumratios{}$ and it must be an integer. 
By properties \ref{prop:2-almost-full} and \ref{prop:not-almost-full}, there exists a partition of $B_{3..n}$ into $n-2$ bins $3,\ldots,n$ which are not almost-full.
By definition, the {\sf FPTAS} in step \ref{item:b3n} finds a partition with 
	$\max(b_3,\ldots,b_n) \leq (1+v)S$.
The final steps, regarding the partition of $B_{1:2}$, are justified by the discussion at \Cref{app:big-items-alg}.
The complete running time $O(\text{poly}(m, \log{(\sumratios{}S)}))$  of \Cref{alg:tfeasible} is justified by the running time of the {\sf FPTAS} for \partition{\asym\minmax}{$n$}{${\cal X}, {\cal R}$}, \partition{\asym\minmax}{$n,v, i$}{${\cal X}, {\cal R}$} and for \partition{\asym\minmax}{$n=2,v_1,v_2,i$}{${\cal X}, {\cal R}$}. 
Note that $\sumratios{}S/\sumratios{}M \leq m$ since it is the sum of all items divided by the largest item. Therefore, $1/v = S/uM = \sumratios{}S/\sumratios{}M \cdot 1/u \leq m/u = O(m)$ since $u$ is fixed. Also, note that $1/\epsilon$ depends on $n$ (which is fixed) $1/v$ and $m$ so,  $1/\epsilon = O(m)$.
The exact running time, $O(m^4\log{(\sumratios{}S)})$, is detailed in \Cref{sec:running-time1.1}.
\end{proof}

\subsection{Hardness for $n \geq 3$ bins and $u<n-2$}\label{sec:hard_general_m}

The following theorem complements the previous subsection. 
\begin{theorem}
\label{thm:np-complete-d-possible}
Given a fixed integer $n \geq 3$ and a positive rational number $u < n-2$,
the problem	\interval{\decision}{$n,u$}{${\cal X}$}
is {\sf NP}-complete.
\end{theorem}
\begin{proof}
Given an $n$-way partition
of $m$ items, summing the sizes of all elements in each bin allows us to check whether the partition is $u$-possible in linear time. So, the problem is in {\sf NP}. To prove that \interval{\decision}{$n,u$}{${\cal X}$} is {\sf NP}-Hard, we reduce from the 
equal-cardinality partition
problem, proved to be {\sf NP}-hard in \cite{gareycomputers}: given a list with an even number of integers, decide if they can be partitioned into two subsets with the same sum and the same cardinality. 

Given an instance ${\cal X}_1$ of equal-cardinality partition, denote the number of items in ${\cal X}_1$ by $2 m'$.
Define $M$ to be the sum of numbers in ${\cal X}_1$ divided by $2n(1-\frac{u}{n-2})$, so that the sum of items in ${\cal X}_1$ is $2n(1-\frac{u}{n-2}) M$
(where $n$ and $u$ are the parameters in the theorem statement).
We can assume w.l.o.g. that all items in ${\cal X}_1$ are at most $n(1-\frac{u}{n-2}) M$, since if some item is larger than half of the sum, the answer is necessarily ``no''.

Construct an instance ${\cal X}_2$ of the equal-cardinality partition problem by replacing each item $x$ in ${\cal X}_1$ by $nM-x$. 
So ${\cal X}_2$ contains $2 m'$ items between $n(\frac{u}{n-2}) M$ and $n M$. Their sum, which we denote by $2 S'$, satisfies
$
2 S' = 2 m' \cdot n M - 2n\left(1-\frac{u}{n-2}\right) M = 2n \left(m' - 1 + \frac{u}{n-2}\right) M.
$
Clearly, ${\cal X}_1$ has an equal-sum equal-cardinality partition (with bin sums $n\left(1-\frac{u}{n-2}\right) M$) iff ${\cal X}_2$ has an equal-sum equal-cardinality partition (with bin sums $S' = n \left(m' - 1 + \frac{u}{n-2}\right) M$).

Construct an instance $({\cal X}_3, u)$ of
\interval{\decision}{$n,u$}{${\cal X}$}
by adding $(n-2)(m'-1)$ items of size $n M$.
Note that $n M$ is indeed the largest item size in ${\cal X}_3$. 
Denote the sum of item sizes in ${\cal X}_3$ by $n S$. Then
\begin{align*}
n S &=  2S' + (n-2)(m'-1) \cdot n M
=  n \left(2 (m' - 1) + \frac{2u}{n-2} + (n-2)(m'-1) \right) \cdot M 
\\ &
=  n \left(n(m' - 1) + \frac{2u}{n-2} \right) M;
\end{align*}
$$S + u M = \left(n(m' - 1) + \frac{2u}{n-2} + u \right) M
= \left(n(m' - 1) + \frac{n u}{n-2} \right) M
= S',$$
so a partition of ${\cal X}_3$ is $u$-possible if and only if the sum of each of the $n$ bins in the partition is at most $S+uM = S'$.

We now prove that if ${\cal X}_2$ has an equal-sum equal-cardinality partition, then the instance (${\cal X}_3, u$) has a $u$-possible partition, and vice versa. If ${\cal X}_2$ has an equal-sum partition, then the items of ${\cal X}_2$ can be partitioned into two bins of sum $S'$, and the additional $(n-2)(m'-1)$ items can be divided into $n-2$ bins of $m'-1$ items each. The sum of these items is
\begin{align}
\label{eq:(m-1)nM-hard}
(m' -1)\cdot nM =  n(m'-1)M 
= S - \frac{2}{n-2}uM < S+u M = S',
\end{align}
so the resulting partition is a $u$-possible partition of ${\cal X}_3$.
Conversely, suppose ${\cal X}_3$ has a $u$-possible partition. Let us analyze its structure.

Since the partition is $u$-possible, the sum of every two bins is at most $2(S + uM)$.
So the sum of every $n-2$ bins is at least $nS - 2(S + uM) = (n-2)S - 2uM$. 
Since the largest $(n-2)(m'-1)$ items in ${\cal X}_3$ sum up to exactly $(n-2)S - 2uM$ by \eqref{eq:(m-1)nM-hard}, every $n-2$ bins must contain at least $(n-2)(m'-1)$ items.
Since ${\cal X}_3$ has $(n-2)(m'-1)+2m'$ items overall, $n-2$ bins must contain exactly $(n-2)(m'-1)$ items, such that each item size must be $n M$, and their sum must be $(n-2)S - 2uM$.
The other two bins contain together $2m'$ items with a sum of $2(S+uM)$, so each of these bins must have a sum of exactly $S+uM$. Since $(m' -1)\cdot nM < S + uM$ by \eqref{eq:(m-1)nM-hard}, each of these two bins must contain exactly $m'$ items.
These latter two bins are an equal-sum equal-cardinality partition for ${\cal X}_2$. This construction is done in polynomial time, completing the reduction.
\end{proof}


\section{Partition with Split Items} \label{sub:equivalence}
We now deal with the problem \splititemempty{}. We redefine the \splititem{\asym\decision}{$n,s$}{${\cal X}, {\cal R}$} problem.
For a fixed number $n \geq 2$ of bins, given a list ${\cal X}$, some ratios ${\cal R}$, the number of split items $s\in\{0,\ldots, m\}$ and a rational number $v\geq 0$, define:
\begin{quoting}
\splititem{\asym\decision}{$n,s,v$}{${\cal X}, {\cal R}$}: \quad
Decide if there exists a partition of ${\cal X}$ among $n$ bins with ratios ${\cal R}$ and relative sums $b_1,\ldots,b_n$ with at most $s$ split items, such that $\max(b_1,\ldots,b_n)\leq(1+v)S$, where $S:= (\sum_j x_j)/\sumratios{}$.
\end{quoting}
The special case $v=0$ corresponds to the \splititem{\asym\decision}{$n,s$}{${\cal X}, {\cal R}$} problem.
The following Lemma shows that, w.l.o.g., we can consider only the longest items for splitting.

\begin{lemma}\label{largestObjects}
For every partition with $s\in \mathbb{N}$ split items and bin relative sums $b_1,\ldots,b_n$, there exists a partition with the same bin relative sums $b_1,\ldots,b_n$ in which only the $s$ \emph{largest} items are split. 
\end{lemma}

\begin{proof}
Consider a partition in which some item with size $x$ is split between two or more bins, whereas some item with size $y>x$ is allocated entirely to some bin $i$. Construct a new partition as follows: first move item $x$ to bin $i$; second remove from bin $i$, a fraction $\frac{x}{y}$ of item $y$; and finally split that fraction of item $y$ among the other bins, in the same proportions as the previous split of item $x$.
All bin relative sums remain the same.
Repeat the argument until only the longest items are split.
\end{proof}

\begin{theorem}
\label{thm:from-interval-to-shared}
For any fixed integers $n\geq 2$ and $u\geq 0$, there is a polynomial-time reduction from \interval{\decision}{$n,u$}{${\cal X}$} to \splititem{\decision}{$n,s=u,v=0$}{${\cal X}$}.
\end{theorem}

\begin{proof}
Given an instance ${\cal X}$ of 
\interval{\decision}{$n,u$}{${\cal X}$}, we add $u$ items of size $nM$, where $nM$ is the size of the biggest item in ${\cal X}$ to construct an instance
${\cal X}'$ of \splititem{\decision}{$n,s=u,v=0$}{${\cal X}'$}.
	
First, assume that ${\cal X}$ has a $u$-possible partition. Then there are $n$ bins with a sum at most $S+u M$. Take the $u$ added items of size $n M$ and add them to the bins, possibly splitting some items between bins, such that the sum of each bin becomes exactly $S + u M$. This is possible because the sum of the items in ${\cal X}'$ is $nS + unM=n(S+uM)$. The result is a $0$-feasible partition of ${\cal X}'$ with at most $u$ split items.
	
Second, assume that ${\cal X}'$ has a $0$-feasible partition with at most $u$ split items.
Then there are $n$ bins with a sum of exactly $S+u M$. By \Cref{largestObjects}, we can assume the split items are the largest ones, which are the $u$ added items of size $nM$. Remove these items to get a partition of ${\cal X}$. The sum in each bin is now at most $S+uM$, so the partition is $u$-possible. This construction is done in polynomial time, which completes the proof.
\end{proof}

\begin{corollary}\label{cor:splithard}
For every fixed integers $n\geq 3$ and $s \in \{0,\dots,n-3\}$, the problem \splititem{\decision}{$n,s$}{${\cal X}$} is {\sf NP}-complete.
\end{corollary}

\begin{proof}
\Cref{thm:np-complete-d-possible} and
\Cref{thm:from-interval-to-shared} 
imply that \splititem{\decision}{$n,s$}{${\cal X}$} is
{\sf NP}-hard.
The problem \splititem{\decision}{$n,s$}{${\cal X}$} is in {\sf NP}
since given a partition, summing the sizes of the items (or items fractions) in each bin let us check in linear time whether the partition has equal bin sums.
\end{proof}

\begin{theorem}
\label{thm:from-shared-to-interval}
For any fixed integers $n\geq 2, s\geq 0$ and rational $v\geq 0$,
there is a polynomial-time reduction from 
\splititem{\asym\decision}{$n,s,v$}{${\cal X}, {\cal R}$},
to 
\interval{\asym\decision}{$n,u$}{${\cal X}, {\cal R}$}
for some rational number $u \geq s$.
\end{theorem}

\begin{proof}
Given an instance ${\cal X}$ of  \splititem{\asym\decision}{$n,s,v$}{${\cal X}, {\cal R}$},
denote the sum of all items in ${\cal X}$ by $\sumratios{} S$ and the largest item size by $\sumratios{} M$ where $S, M \in \mathbb{Q}$.
Construct an instance ${\cal X}'$ of \interval{\asym\decision}{$n,u$}{${\cal X}', {\cal R}$} by removing the $s$ largest items from ${\cal X}$.
Denote the sum of remaining items by $\sumratios{} S'$ for some $S'\leq S$, and the largest remaining item size by $\sumratios{} M'$ for some $M'\leq M$.
Note that the size of every removed item is between $\sumratios{} M'$ and $\sumratios{} M$, so  $s M' \leq S - S' \leq sM$.
Set $u := (S+vS - S')/M'$, so $S'+uM' = S+vS$. Note that $u\geq (S-S')/M'\geq s$.

First, assume that ${\cal X}$ has a $v$-feasible partition with $s$ split items.
By \Cref{largestObjects}, we can assume that only the $s$ largest items are split.
Therefore, removing the $s$ largest items results in a partition of ${\cal X}'$ with no split items, where the relative sum in each bin is at most $S+vS = S' + u M'$.
This is a $u$-possible partition of ${\cal X}'$.

Second, assume that ${\cal X}'$ has a $u$-possible partition. In this partition, each bin relative sum is at most $S' + uM' = S+vS$, so it is a $v$-feasible partition of ${\cal X}'$.
To get a $v$-feasible partition of ${\cal X}$,
take the $s$ previously removed items and add them to the bins, possibly splitting some items between bins, such that the relative sum in each bin remains at most $S+vS$. This is possible since the sum of the items is $\sumratios{} S \leq \sumratios{} (S+vS)$. This construction is done in polynomial time.
\end{proof}
Combining \Cref{thm:from-shared-to-interval} with \Cref{thm:n>=3} provides a polynomial time algorithm to solve \splititem{\asym\decision}{$n,s,v$}{${\cal X}, {\cal R}$} for any fixed  $n\geq 3, s\geq n-2$ and rational $v\geq 0$.
The latter is used to solve the \splititem{\asym\minmax}{$n,s$}{${\cal X}, {\cal R}$} optimization problem
by using binary search on the parameter $v$ of the \splititem{\asym\decision}{$n,s,v$}{${\cal X}, {\cal R}$} problem. 

\subsection{Binary Search to solve \splititem{\asym\minmax}{$n,s$}{${\cal X}, {\cal R}$}}

We first prove a property of the optimal bin relative sum of the \splititem{\asym\minmax}{$n,s$}{${\cal X}, {\cal R}$} problem.
\begin{lemma}\label{lem:perfectbin}
The optimal value of an \splititem{\asym\minmax}{$n,s$}{${\cal X}, {\cal R}$} instance is either a perfect partition (that is, all bin relative sums are equal to the average bin relative sum $S$) or equal to $q /r_i$, where $q$ is an integer and $r_i \in {\cal R}$.
\end{lemma}

\begin{proof}
Let OPT be the maximum bin relative sum in an optimal partition of ${\cal X}$ with $s$ split items.
If OPT is not equal to $q /r_i$, where $q$ is an integer and $r_i \in {\cal R}$, then the fractional part of OPT
must come from a split item (since the non-split items have integer sizes).
Assuming that the partition is not perfect, there is at least one bin with a relative bin sum smaller than OPT.
Move a small fraction of the split item from all bins with relative sum OPT to a bin with relative sum smaller than OPT. 
This yields a partition with a smaller maximum relative sum and the same set of split items --- a contradiction to optimality.
\end{proof}

Now, we show how to use binary search to solve the \splititem{\asym\minmax}{$n,s$}{${\cal X}, {\cal R}$} problem. The idea is to execute the  \splititem{\asym\decision}{$n,s,v$}{${\cal X}, {\cal R}$} algorithm several times, trimming the parameter $v$ at each execution in a binary search style.
When $v = \sumratios{}-1$, the answer to \splititem{\asym\decision}{$n,s,v$}{${\cal X}, {\cal R}$} is always ``yes'', since even if we put all items in a single bin, its sum is $\sumratios{} S  = (1+v)S$.
Moreover, by \Cref{lem:perfectbin}, the optimal bin relative sum value of \splititem{\asym\minmax}{$n,s$}{${\cal X}, {\cal R}$} is either equal to $q /r_i$ where $q$ is an integer and $r_i \in {\cal R}$ or equal to $S$. Therefore, in addition to considering the case where $v=0$, we only need to consider values of $v$ that differ by at least $1/\sumratios{}$, since if the bin relative sums are not equal to $S$ they are always integers, so they differ by at least $1/\sumratios{}$. 
Therefore, at most $\sumratios{}^2S$ different values of $v$ have to be checked, so the binary search requires at most $\log_2(\sumratios{}^2S)=2\log_2(\sumratios{}S)$ 
executions of \splititem{\asym\decision}{$n,s,v$}{${\cal X}, {\cal R}$}. This is polynomial in the size of the binary representation of the input.

\begin{corollary}
\label{cor:nsplitparmm}
For any fixed integers $n\geq 3$ and $s\geq n-2$,
\splititem{\asym\minmax}{$n,s$}{${\cal X}, {\cal R}$}
can be solved in $O(\text{poly}(m, \log (\sumratios{} S)))$ time.
\end{corollary}

We complete this result by providing a polynomial-time algorithm for the max-min version: \splititem{\maxmin}{$n, s$}{${\cal X}, {\cal R}$} for $s \geq n-2$.

\subsection{\splititem{\asym\maxmin}{$n, s$}{${\cal X}, {\cal R}$}}

We first prove a structural Lemma:
\begin{lemma}\label{lem:maxminstructure}
For any $n\geq 3$ and $s\geq n-2$, in any instance of \splititem{\asym\maxmin}{$n, s$}{${\cal X}, {\cal R}$}, either the output is perfect (all the relative sums are equal), or, the split items are shared only between the $n-1$ bins with the smaller relative sums, and their relative sums are equal.
\end{lemma}
\begin{proof}
    Assume w.l.o.g. that bin $n$ has the largest relative sum. Bin $n$ does not contain any part of any split item, since otherwise, we could move some part of the split item from bin $n$ to all smallest bin relative sums, and get a partition with a larger minimum relative sum.
    So all split items are shared only among the $n-1$ smallest bin relative sums, $1, 2, \dots, n-1$. Since we have $n-1$ bins and we are allowed to split $s\geq n-2$ items, bins $1, 2, \dots, n-1$ relative sums must be equal since otherwise, we could obtain a partition with a larger minimum value by running the cut-the-line algorithm on the contents of bins $1, 2, \dots, n-1$.
\end{proof}

We are ready to design the algorithm. The max-min version of \splititem{\asym\maxmin}{$n, s$}{${\cal X}, {\cal R}$} is based on the min-max version of the same problem. 
Recall that the sum of all items is equal to $\sumratios{}S$.
\begin{algorithm}[H] \caption{\qquad \splititem{\asym\maxmin}{$n, s$}{${\cal X}, {\cal R}$}} \label{alg:maxmin}
\begin{algorithmic}[1]
\State $(b_1,b_2,\dots,b_n) \longleftarrow$ \splititem{\asym\minmax}{$n, s$}{${\cal X}, {\cal R}$}; w.l.o.g. assume $b_1\leq b_2\leq \dots\leq b_n$.
\If{$b_n=S$} 
return $(b_1,b_2,\dots,b_n)$. \Comment{$b_1=b_2=\dots=b_n$ (perfect partition)}
\Else \Comment{$b_n > S$}
\State Divide the items in bins $1, 2, \dots, n-1$ into $n-1$ bins with  $b_1' = b_2' = \dots = b_{n-1}'= \frac{\sumratios{}S - r_nb_n}{\sumratios{}-r_n}$ 
(with at most $n-2$ split item);
\State return $(b_1',b_2',\dots,b_{n-1}',b_n)$
\EndIf
\end{algorithmic}
\end{algorithm}

\begin{theorem}\label{thm:maxminalgo}
\Cref{alg:maxmin} solves \splititem{\asym\maxmin}{$n, s$}{${\cal X}, {\cal R}$} for any fixed integer $n\geq 3$ and $s\geq n-2$,  in $O(\text{poly}(m, \log (\sumratios{} S)))$ time.
\end{theorem}
\begin{proof}
Let $V$ be the value of the smallest bin relative sum in the optimal max-min partition.

If $V=S$, then there is a perfect partition, and we find it in the first if.

Otherwise, $V <S$. By \Cref{lem:maxminstructure}, there is a partition in which bins $1, 2, \dots, n-1$ have relative sum  exactly $V$, and the split items are shared among them. So bin $n$ sum equals $r_nb_n = \sumratios{}S-(\sumratios{}-r_n)V$. So \splititem{\asym\minmax}{$n, s$}{${\cal X}, {\cal R}$} will find an partition with $r_nb_n \leq (\sumratios{}S-(\sumratios{}-r_n)V)$, and therefore $b_1,b_2,\dots,b_{n-1}$ will be at least $\frac{\sumratios{}S - r_nb_n}{\sumratios{}-r_n} \geq V$.

The complete running time $O(\text{poly}(m, \log{(\sumratios{}S)}))$  of \Cref{alg:maxmin} is justified by the running time of the \splititem{\asym\minmax}{$n, s$}{${\cal X}, {\cal R}$} problem, which is $O(\text{poly}(m, \log{(\sumratios{}S)}))$.
\end{proof}

\section{Partition with Splittings}

In this section, we analyze the last variant: \splittingempty{}. We recall that this variant bounds by $t$ the number of splittings (the number of times the items are split). 

\begin{theorem}\label{thm:sharing}
For any fixed integer $n\ge 2$ and fixed $t \in \mathbb{N}$ such that $t \leq n-2$,
the problem \splitting{\decision}{$n,t$}{${\cal X}$}
is {\sf NP}-complete.
\end{theorem} 

\begin{proof}
Given a partition with $n$ bins, $m$ items, and $t$ splittings, summing the size of each item (or fraction of item) in each bin allows us to check whether or not the partition is perfect in linear time. So, the problem is in {\sf NP}.

To prove that \splitting{\decision}{$n,t$}{${\cal X}$} is {\sf NP}-Hard, we apply a reduction from the {\sf Subset Sum} problem. We are given an instance ${\cal X}_1$ of {\sf Subset Sum} with $m$ items summing up to $S$ and target sum $T<S$.
We build an instance ${\cal X}_2$ of
\splitting{\decision}{$n,t$}{${\cal X}_2$}
by adding two items, $x_1,x_2$, such that $x_1=S+T$ and $x_2 = 2S(t+1) - T$ and $n-2-t$ auxiliary items of size $2S$. Notice that the sum of the items in ${\cal X}_2$ equals
\begin{align*}
&
S+(S+T)+2S(t+1)-T + 2S (n-2-t)
=
2S + 2S(t+1) + 2S (n-2-t)
\\&
=
2S\cdot(1+t+1+n-2-t)
=
2Sn.
\end{align*}
The goal is to partition items into $n$ bins with a sum of $2 S$ per bin, and at most $t$ splittings.

First, assume that there is a subset of items $W_1$ in ${\cal X}_1$ with a sum equal to $T$.
Define a set, $W_2$, of items that contains all items in ${\cal X}_1$ that are not in $W_1$, plus $x_1$. The sum of $W_2$ is $(S - T) + x_1 = S + T + S - T = 2S$.
Assign the items of $W_2$ to the first bin. 
Assign each auxiliary item to a different bin.
There are $n - (n-2-t + 1) = t+1$ bins left.
The sum of the remaining items is $2S(t+1)$.
Using the ``cut-the-line'' algorithm described in the introduction, these items can be partitioned into $t+1$ bins of equal sum $2 S$, with at most $t$ splittings.
All in all, there are $n$ bins with a sum of $2 S$ per bin, and the total number of splittings is at most $t$. 

Second, assume that there exists an equal partition for $n$ bins with $t$ splittings. Since $x_2 = 2S(t+1) - T = 2S\cdot t + (2S - T) > 2S\cdot t$, 
this item must be split between $t+1$ bins, which makes the total number of splittings at least $t$. Also, the auxiliary items must be assigned without splittings into $n-2-t$ different bins. There is $n - t - 1 - n + 2 + t = 1$ bin remaining, say bin $i$, containing only whole items, not containing any part of $x_2$, and not containing any auxiliary item. Bin $i$ must contain $x_1$, otherwise its sum is at most $S$ (sum of items in ${\cal X}_1$). Let $W_1$ be the items of ${\cal X}_1$ that are not in bin $i$. 
The sum of $W_1$ is 
$S - (2 S - x_1) = x_1 - S = T$, so it is a solution to ${\cal X}_1$.
\end{proof}

\section{Strong hardness proofs}\label{sec:n-part-of-input}

In this section, we prove that when $n$ is part of the input (and not a fixed parameter), the problems \intervalempty{}, \splititemempty{} and \splittingempty{} are strongly {\sf NP}-hard.

\subsection{The \interval{\decision}{$u$}{$n, {\cal X}$} problem}


\begin{theorem}\label{thm:n-input-hard}
For every fixed rational number $u\geq 0$,
\interval{\decision}{$u$}{$n, {\cal X}$} is strongly {\sf NP}-hard.
\end{theorem}

\begin{proof}
We apply a reduction from the 3-partition problem:
given a list ${\cal X}_1$ of $3m'$ positive integers with sum equal to $m'S'$, decide if they can be partitioned into $m'$ triplets such that the sum of each triplet is $S'$ (we call such a partition a \emph{triplet-partition}).
3-partition is proved to be strongly {\sf NP}-hard in \citep{DBLP:journals/siamcomp/GareyJ75}.
	
Given an instance ${\cal X}_1$ of 3-partition with $3 m'$ integers, define $n:=\lceil2(m'+u)\rceil$, where $u$ is the parameter of the problem \interval{\decision}{$u$}{$n, {\cal X}$}.
Define $M$ to be the sum of numbers in ${\cal X}_1$ divided by $m'n(1-\frac{u}{n-m'})$, so that the sum of numbers in ${\cal X}_1$ is $m'n(1-\frac{u}{n-m'}) M$. 
Since $n=\lceil2(m'+u)\rceil$, it follows that $1-\frac{u}{n-m'} > 0$, so 
$M$
is positive.
	The 3-partition problem decides if the items can be partitioned into $m'$ triplets such that the sum of each triplet is $n(1-\frac{u}{n-m'}) M$.
	We can assume w.l.o.g. that all items in ${\cal X}_1$ are at most $n(1-\frac{u}{n-m'}) M$, since if some item is larger than $n(1-\frac{u}{n-m'}) M$, the answer is necessarily ``no''.
	
	Construct an instance ${\cal X}_2$ of the 3-partition problem by replacing each item $x$ in ${\cal X}_1$ by $nM-x$. 
	So ${\cal X}_2$ contains $3 m'$ items between $n(\frac{u}{n-m'}) M$ and $n M$. Their sum, which we denote by $m' S'$, satisfies
	\begin{align*}
		m' S' = 3 m' \cdot n M - m'n\left(1-\frac{u}{n-m'}\right) M = m'n \left(2 + \frac{u}{n-m'}\right) M.
	\end{align*}
	Clearly, ${\cal X}_1$ has a triplet-partition (with sum $n\left(1-\frac{u}{n-m'}\right) M$) if and only if ${\cal X}_2$ has a triplet-partition (with sum $S' = n \left(2 + \frac{u}{n-m'}\right) M$).
	The construction is done in time polynomial in the size of ${\cal X}_1$.
	
	Construct an instance $(n, {\cal X}_3)$ of
    \interval{\decision}{$u$}{$n, {\cal X}_3$}
	by adding, to the items in ${\cal X}_{2}$, some $2(n-m')$ items equals to $n M$.
Note that $n M$ is indeed the largest item in ${\cal X}_3$. 
Denote the sum of items in ${\cal X}_3$ by $n S$. So
\begin{align*}
n S &=  m'S' + 2(n-m') \cdot n M
=  n \left(2m' + \frac{m'u}{n-m'} + 2(n-m') \right) \cdot M 
\\
&=  n \left(2n + \frac{m'u}{n-m'} \right) M;
\end{align*}
$$S + u M = \left(2n + \frac{m'u}{n-m'} + u \right) M
= \left(2n + \frac{n u}{n-m'} \right) M
= S',$$
so a partition of ${\cal X}_3$ is $u$-possible if and only if the bin sum of each of the $n$ bins in the partition is at most $S+uM = S'$.

We now prove that if ${\cal X}_2$ has a triplet-partition
then the corresponding instance ($n, {\cal X}_3$) has a $u$-possible partition, and vice versa.

If ${\cal X}_2$ has a triplet-partition, then the items of ${\cal X}_2$ can be partitioned into $m'$ bins of sum $S'$, and the additional $2(n-m')$ items can be divided into $n-m'$ bins of 2 items each. Note that the sum of 2 additional items is
\begin{align}
\label{eq:(m-1)nM}
2\cdot nM  
& = \left(2n + \frac{n u}{n-m'} \right) M - \frac{n u}{n-m'} M 
 = S'- \frac{n u}{n-m'} M  < S',
\end{align}
Since $n>m'$.

Conversely, suppose ${\cal X}_3$ has a $u$-possible partition. Let us analyze its structure.
\begin{itemize}
\item [--] Since the partition is $u$-possible, every $m'$ bin obtains together items with a total size of at most $m'(S + uM)$.
\item [--] So every $n-m'$ bins obtain items with total size at least $nS - m'(S + uM) = (n-m')S - m'uM$. 
\item [--] The largest $2(n-m')$ items in ${\cal X}_3$ 
sum up to exactly $(n-m')\cdot 2 nM 
= (n-m')S' - nu M$ by \eqref{eq:(m-1)nM},
which equals
$(n-m')(S + uM) - n u M 
= (n-m') S - m' uM$.
\item [--] Hence,
every $n-m'$ bins must together obtain \emph{at least} $2(n-m')$ items.
\item [--] 
Let $C$ be the set of $n-m'$ bins with the fewest items. 
We claim that bins in $C$ must obtain together \emph{exactly} $2(n-m')$ items.
Suppose by contradiction that bins in $C$ obtained $2(n-m')+1$ or more items. 
By the pigeonhole principle, there was a bin in $C$ that obtained at least $3$ items. By minimality of $C$, the $m'$ bins not in $C$ also obtained at least $3$ items each. The total number of items were $2(n-m') + 1 + 3m'$. This is a contradiction since the total number of items in ${\cal X}_3$ is only $2(n-m')+3m'$.
\item [--]
The sum of items in $C$ must still be at least 
$(n-m')S - m'uM = 2(n-m')\cdot nM$, so each item in $C$ must have the maximum size of $n M$.
\item [--] The other $m'$ bins obtain together $3m'$ items with total size $m'(S+uM)$, so each of these bins must have a bin sum of exactly $S+uM$. Since $2 \cdot nM < S + uM$ by \eqref{eq:(m-1)nM}, each of these $m'$ bins must obtain exactly 3 items.
\end{itemize}
These latter $m'$ bins correspond to the bins in a triplet-partition of ${\cal X}_2$. 
This construction is done in time polynomial in the 
size of ${\cal X}_2$, since we added $2(n-m')$ items to the initial $3m'$ items,
so the new number of items is $m'+2n < m'+4(m'+u)+1 = 5m'+4u+1$.
For every constant $u$, this number is linear in the size of  ${\cal X}_2$.
\end{proof}

\subsection{The \splititem{\decision}{$s$}{$n, {\cal X}$} problem}

\begin{theorem}\label{thm:splitreducinter}
    For any fixed integer $s=u\geq 0$,
    there is a polynomial-time reduction from
    \interval{\decision}{$u$}{$n, {\cal X}$}
    to
    \splititem{\decision}{$s$}{$n, {\cal X}$}.
\end{theorem}

\begin{proof}
Given an instance $(n, {\cal X})$ of 
\interval{\decision}{$u$}{$n, {\cal X}$}, construct an instance
$( n', {\cal X}')$ of \splititem{\decision}{$s$}{$n', {\cal X}'$}
by setting $n'=n$
and adding $u=s$ items of size $nM$, where $nM$ is the size of the largest item in ${\cal X}$.
	
First, assume that ($n, {\cal X})$ has a $u$-possible partition. Then there are $n$ bins with sum at most $S+u M$,
where $S$ is the sum of the items in ${\cal X}$ divided by $n$. Take the $u$ added items of size $n M$ and add them to the bins, possibly splitting some items between bins, such that the sum of each bin becomes exactly $S + u M$. This is possible because the sum of the items in $(n', {\cal X}')$ is $nS + unM=n(S+uM)$. The result is a perfect partition of $(n', {\cal X}')$ with at most $u=s$ split items.
	
Second, assume that $(n', {\cal X}')$ has a perfect partition with at most $s$ split items.
Then there are $n'$ bins with a sum of exactly $S+u M$. By 
\Cref{largestObjects}
we can assume that the split items are the largest ones, which are the $u$ added items of size $nM$. Remove these items to get a partition of ($n, {\cal X})$. The bin sum in each bin is now at most $S+uM$, so the partition is $u$-possible.
	
This construction is done in polynomial time, which completes the proof.
\end{proof}

Combining \Cref{thm:n-input-hard} and \Cref{thm:splitreducinter} gives:

\begin{corollary}\label{cor:splitstronglyhard}
For every fixed integer $s\geq 0$,
the problem \splititem{\decision}{$s$}{$n, {\cal X}$} is strongly {\sf NP}-hard.
\end{corollary}

\subsection{The \splitting{\decision}{$t$}{$n, {\cal X}$} problem}

\begin{theorem}\label{thm:splittingstrong}
For every fixed integer $t \geq 0$, \splitting{\decision}{$t$}{$n, {\cal X}$} is strongly {\sf NP}-hard.
\end{theorem}

\begin{proof}
By reduction from 3-partition. Given a finite multiset $D$ of $3p$ positive integers $d_1,\ldots,d_{3p}$  summing up to $pS'$, we have to decide if they can be partitioned into $p$ triplets such that the sum of each triplet is $S'$. 
We construct an instance of \splitting{\decision}{$t$}{$n, {\cal X}$} with $n=p+t+1$ bins and $m=3p+1$ items with the following values:
\begin{itemize}
\item The value of each item $x_i\in \{1,\ldots, 3p\}$ is $d_{x_i} + S'$;
\item The value of item $3p+1$ is 
$(t+1)4S'$. 
\end{itemize}
The total value of all items is $\sum_{x_i=1}{3p}(d_{x_i} + S') + (t+1)4S'
= 
pS' + 3pS' + (t+1)4S'
=
n\cdot 4S'
$, so the partition is perfect if and only if each bin sum is exactly $4S'$.

Any solution to 3-partition gives us $p$ triplets, each of which is a subset of $\{d_1,\ldots,d_{3p}\}$ with sum equal to $S'$. adding $S'$ to the value of each item in the triplet gives us $p$ bins with sums equal to $4 S'$. The remaining $t+1$ bins are constructed by cutting item $3p+1$ into $t+1$ parts with value $4S'$ each.

Conversely, suppose there is a perfect partition with at most $t$ splittings. Item $3p+1$ must be cut $t$ times since otherwise, one part will be bigger than $4S'$. There are $3p$ remaining items and $p$ remaining bins, and no splitting is allowed anymore. 
Each bin must receive exactly $3$ items since otherwise, one bin sum will be larger than $4S'$.
Removing $S'$ from the value of each item yields a solution to 3-partition. 
\end{proof}

\section{Conclusion and Future Directions}

We presented three variants of the $n$-way number partitioning problem. 

In the language of fair item allocation, we have solved the problem of finding a fair allocation among $n$ agents with asymmetric valuations, when the ownership of some $s$ items may be split between agents. 
When agents may have different valuations, there are various fairness notions, such as \emph{proportionality}, \emph{envy-freeness} or \emph{equitability}. 
A future research direction is to develop algorithms for finding such allocations with a bounded number of shared items. Some results for proportional allocation among three agents with different valuations are based on the algorithms in the present paper.

In the language of machine scheduling, \splititem{\asym\minmax}{$n,s$}{${\cal X}, {\cal R}$} corresponds to finding a schedule that minimizes the makespan on $n$ uniform machines with speeds ${\cal R}$ when $s$ jobs can be split between the machines; 
\interval{\asym\decision}{$n,u$}{${\cal X}, {\cal R}$} corresponds to finding a schedule in which the makespan is in a given interval. 
Studying the more general setting of \emph{unrelated machines} may be interesting.

Our analysis shows the similarities and differences between these variants and the more common notion of {\sf FPTAS}. 
One may view our results as introducing a new kind of approximation that approximates a decision problem by returning ``yes'' if and only if there exists a solution between $PER$ and $(1+v)\cdot PER$, where $PER$ represents the value of a perfect solution. For the $n$-way number partitioning problem, a perfect solution is easy to define: it is a  partition with equal bin sums.
A more general definition of $PER$ could be the solution to the fractional relaxation of an integer linear program representing the problem.
As shown, {\sf NP}-hard decision problems may become tractable when $v$ is sufficiently large.


\backmatter

\bmhead{Acknowledgements}

The paper started from discussions in the stack exchange network:\\
    (1) \href{https://cstheory.stackexchange.com/q/42275}{https://cstheory.stackexchange.com/q/42275}; \\
    (2) \href{https://cs.stackexchange.com/a/141322}{https://cs.stackexchange.com/a/141322}. \\
\splititem{\decision}{$3, 1$}{${\cal X}$} was first solved by Mikhail Rudoy using case analysis. The relation to FPTAS was raised by Chao Xu. We are also grateful to John L. \href{https://cs.stackexchange.com/a/149567}{https://cs.stackexchange.com/a/149567}.

\section*{Declarations}

\begin{itemize}
\item Funding.
Samuel Bismuth: Israel Science Foundation grant no. 712/20.
Erel Segal-Halevi: Israel Science Foundation grant no. 712/20.
\end{itemize}






\begin{appendices}

\section{Additional Results}

In this section, we state some additional results, complementing the results given in the paper.

\subsection{An {\sf FPTAS} for \splititem{\minmax}{$n, s$}{${\cal X}$} when $s<n-2$}\label{sec:approximations}

Although we provide a polynomial-time algorithm for the case $s=n-2$, our results suggest possible solutions for the general case too.

\begin{theorem}\label{thm:FPTAS_split}
	There exists an {\sf FPTAS} for \splititem{\asym\minmax}{$n, s$}{${\cal X}, {\cal R}$} for any $n$ and $s$, given any rational $\epsilon$.
\end{theorem}
\begin{proof}
	Consider the following algorithm:
	\begin{algorithm}[!h] \caption{\qquad {\sf FPTAS}(\splititem{\asym\minmax}{$n, s$}{${\cal X}, {\cal R}$}, $\epsilon$)} \label{alg:FPTAS_split}
		\begin{algorithmic}[1]
			\State Let ${\cal X}_1 := $ the $s$ largest items in ${\cal X}$. Let ${\cal X}_2 := $ the remaining items. 
			\State Solve {\sf FPTAS}(\partition{\asym\minmax}{$n$}{${\cal X}_2, {\cal R}$}, $\epsilon)$.
			\State Split the items in ${\cal X}_1$ among the bins such that the resulting partition is as equal as possible.
		\end{algorithmic}
	\end{algorithm}	 
The run time of the algorithm depends on {\sf FPTAS}(\partition{\asym\minmax}{$n$}{${\cal X}_2, {\cal R}$}, $\epsilon$) so, it is polynomial in the length of the input ${\cal X}$ and $1/\epsilon$.

We now prove that \Cref{alg:FPTAS_split} is correct. Let $OPT$ be the optimal solution of \splititem{\asym\minmax}{$n, s$}{${\cal X}, {\cal R}$} and $OPT_2$ be the optimal solution of \partition{\asym\minmax}{$n$}{${\cal X}_2, {\cal R}$}.
By \Cref{largestObjects}, we can assume that only the $s$ largest items are split.
By definition, {\sf FPTAS}(\partition{\asym\minmax}{$n$}{${\cal X}_2, {\cal R}$}, $\epsilon$) returns a partition of ${\cal X}_2$ where the largest bin relative sum is at most $(1+\epsilon)OPT_2$.
After splitting the $s$ largest items in ${\cal X}$, if the algorithm outputs a perfect partition (with a value equal to  $(\sum_{i\in {\cal X}} x_i)/\sumratios{}$), then we are done.
Otherwise, the output of our algorithm must be equal to  $(1+\epsilon)OPT_2$. This is because, if the partition is not perfect, no part of the split items is allocated to the bin with the largest sum.
Obviously, $OPT_2 \leq OPT$ since \splititem{\asym\minmax}{$n, s$}{${\cal X}, {\cal R}$} has to allocate the same whole items as \partition{\asym\minmax}{$n$}{${\cal X}_2, {\cal R}$}, plus $s$ additional split items.
Therefore, \Cref{alg:FPTAS_split} returns a partition where the largest relative sum is at most $(1+\epsilon)OPT$.
\end{proof}

\subsection{Experiments} \label{sec:experiments}

What is the effect of the number of split items $s$ on the quality of the optimal partition in \splititem{\minmax}{$n,s$}{${\cal X}$}?
We explored this question using experiments on random instances. 
For simplicity, we used the identical bins version of the problem \partition{\minmax}{$n,s$}{${\cal X}$}.
Since the problem is {\sf NP}-hard for $s<n-2$, we used a heuristic algorithm that we describe next.

\subsubsection{\splititem{\minmax}{$n,s$}{${\cal X}$}: a fast algorithm for every $n$ and $s$}

\Cref{thm:from-shared-to-interval} implies that we can solve any instance of \splititem{\minmax}{$n,s$}{${\cal X}$} in the following way (see \Cref{alg:minmax} for further details):
\begin{enumerate}
    \item Let ${\cal X}_1 := $ the $s$ largest items in ${\cal X}$. Let ${\cal X}_2 := $ the remaining items. 
    \item Solve \partition{\minmax}{$n$}{${\cal X}_2$}.
    \item Split the items in ${\cal X}_1$ among the bins such that the resulting partition is as equal as possible.
\end{enumerate} 

While \partition{\minmax}{$n$}{${\cal X}$} is {\sf NP}-hard, several heuristic algorithms can solve medium-sized instances in reasonable time \cite{SchreiberOptimal}. 
We are interested in algorithms based on state-space search, such as the Complete Greedy Algorithm (CGA) \cite{KorfCGA}.
CGA searches through all partitions by adding to each bin the largest number not yet added to any bin so that bins with smaller sums are prioritized. 
This is done depth-first, meaning that the smallest of the input numbers are shuffled between different parts before larger input numbers are. 
In addition to the heuristics that make CGA run fast in practice, we add a new heuristic, specific to our setting: we stop the search whenever it finds a solution in which the maximum sum is at most $S := (\sum_{x\in X} x)/n$.
This is because, by \Cref{largestObjects}, we can then divide the $s$ split items among the bins and get a perfect partition, in which every bin sum equals $S$.

\Cref{alg:minmax} uses the following variant of the Complete Greedy Algorithm (CGA):
for any input list $Y$ and upper bound $S$, it finds an $n$-way partition of $Y$ where the maximum bin sum is at most $S$; if such a partition does not exist, it finds a partition that minimizes the maximum bin sum. 
While its worst-case run-time is $O(2^m)$, it runs fast on practical instances \citep{KorfCGA}.
Note that other practical algorithms for optimal number partitioning could also be used.
Denote the $n$-way partition obtained by the CGA for a given $Y$ and a given upper bound $S$ by CGA($Y, S$). 

\begin{algorithm}
\caption{\qquad \splititem{\minmax}{$n,s$}{${\cal X}$}}\label{alg:minmax}
\begin{algorithmic}[1]
\State 
Order the items such that $x_1 \geq x_2 \geq \cdots \geq x_m$
\State
${\cal X}_1 \longleftarrow \{x_1, \ldots, x_s\}$
\Comment{$s$ largest items in $\mathcal X$}
\State 
$C  \longleftarrow   \sum_{x \in {\cal X}_1} x $
\State 
${\cal X}_2 \longleftarrow X \setminus {\cal X}_1$.
\State 
$P_2 \longleftarrow CGA({\cal X}_2, S)$
\Comment{optimal (min-max) $n$-way partition of ${\cal X}_2$}
\State
The bin sums in $P_2$ are $b_1,\ldots,b_n$, and we order the bins such that $b_1\leq \cdots \leq b_n$
\For{$i \longleftarrow 1$ to $n-1$}
\State
$B  \longleftarrow  i \cdot (b_{i+1} - b_{i})$
\Comment{$B\geq 0$ since $b_{i+1}\geq b_i$}
\If{$C < B$}  
\Comment{Not enough split items to attain $b_{i+1}$}
\State Break
\Else 
\Comment{Use split items to increase $i$ lowest bins to $b_{i+1}$}
\State 
$C  \longleftarrow   C - B$
\State
$b_1 , \cdots , b_{i} \longleftarrow b_{i+1}$
\EndIf
\EndFor
\State
$b_1, \cdots , b_{i} \longleftarrow b_{i} + C/i$
\end{algorithmic}
\end{algorithm}
Based on \Cref{largestObjects},
we may assume that only the $s$ largest items are split. We remove these items from the instance and keep their sum in the variable $C$.
Then we compute a min-max partition of the remaining items.
Then, we try to split the amount $C$ among the bins, such that the maximum bin sum remains as small as possible.
In each iteration $i$, we use some of the split items for increasing the sums in the $i$ lowest-sum such that their sum becomes equal to the $(i+1)$-th bin.
Finally, we divide the remaining amount of split items equally among the lowest-sum bins.

\subsubsection{Experiments}

In \Cref{fig:plot} we show that the possibility of splitting brings the optimal partition closer to the perfect partition, even when $s$ is much smaller than $n-1$.
We ran several instances of the \splititem{\minmax}{$n,s$}{${\cal X}$} optimization problem
on random inputs drawn from an exponential distribution, and from
a uniform distribution with values between $w M$ and $M$, for $w \in \{0, 0.5, 0.9, 0.99\}$.
The number of items was $m\in\{10, 13, 15\}$, and number of bits in the item sizes was $16$ or $32$. Finally, the number of bins was between 2 and 10, each one represented by one curve in \Cref{fig:plot}. We ran our algorithm on seven different instances of each combination 
and plotted the mean of all instances.
The $x$-axis of the plot represents the number of split items, and the $y$-axis denotes the difference between the optimal and the perfect partition in percent: $100 \cdot (OPT - S) / S$. 

\begin{figure}[h!]
\begin{center}
\includegraphics[scale=0.45]{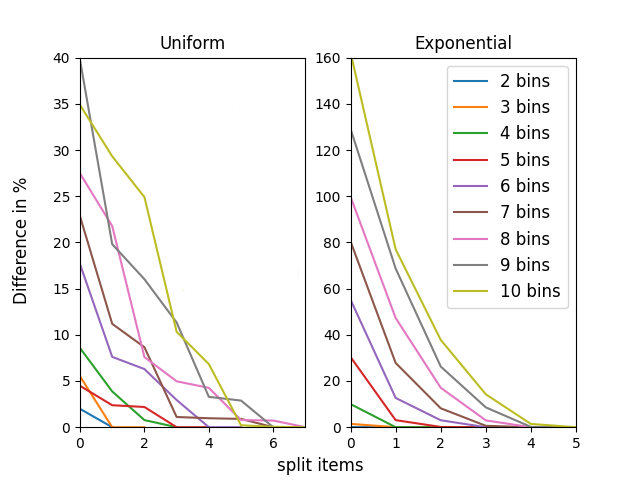}
\end{center}
\caption{Difference between the optimal and the perfect partition, in percent.
The plots show results for uniform distribution with $r=0$ (left plot) and exponential distribution (right plot). In both plots, $m=13$ and the item size has 16 bits.
The results for other values of the parameters are qualitatively similar; we present them in 
the supplementary information.
}
\label{fig:plot}
\end{figure}

\section{Further Related Work}\label{sec:further-related-work}

In this section, we survey some other related work, completing \Cref{sec:related-work}.

\textbf{Bounded splitting in fair division:} Goldberg et al. \cite{goldberg2020consensus} studied the problem of \emph{consensus partitioning}. In this problem, there are $n$ agents with different valuations, and the goal is to partition a list of items into some $k$ subsets (where $k$ and $n$ may be different), 
such that the value of each subset is $1/k$ for each agent.
They prove that a consensus partitioning with at most $n(k-1)$ split items can be found in polynomial time. 
The $n(k-1)$ is tight in the worst case.

Most similar to our paper is the recent work of \cite{sandomirskiy2019fair}.
Their goal is to find an allocation among $n$ agents with different valuations, which is both fair and \emph{fractionally Pareto-optimal {\sf (fPO)}}, a property stronger than Pareto-optimality
(there is no other discrete or fractional allocation
where some agent gains and no agent loses).
This is a very strong requirement: when $n$ is fixed, and the valuations are \emph{non-degenerate} (i.e., for every two agents, no two items have the same value-ratio),
the number of {\sf fPO} allocations is polynomial in $m$, and it is possible to enumerate all such allocations in polynomial time.
Based on this observation, they present an algorithm that finds an allocation with the smallest number of split items, among all allocations that are fair and {\sf fPO}. 
In contrast, in our paper, we do not require {\sf fPO}, which may allow allocations with fewer split items or splittings. However, the number of potential allocations becomes exponential, so enumerating them all is no longer feasible.

\textbf{Splitting in job scheduling:} The generalized multiprocessor scheduling (GMS) has two variants. 
In the first variant,
the total number of preemptions is bounded by some fixed integer.
In the second variant, each job $j$ has an associated parameter that bounds the number of times $j$ can be preempted.
In both variants, the goal is to find a schedule that minimizes the makespan subject to the preemption constraints.
A {\sf PTAS} for GMS with a global preemption bound, and another {\sf PTAS} for GMS with job-wise preemption bound when the number of machines is a fixed constant is presented by \cite{DBLP:journals/algorithmica/ShachnaiTW05}.

For job scheduling with splittings, where different parts of a job may be processed simultaneously on different machines, various objectives have been studied. There are many variants including different setup times. In machine scheduling, the ``setup time'' refers to the time required to prepare a machine for a specific job or task. Sequence-dependent setup time refers to a situation where the setup time required for a job depends on the job that came before it, rather than being constant for all jobs (independent job setup time). 

For identical machines, \cite{DBLP:journals/cor/ShimK08}, suggests a branch and bound algorithm with the objective to minimize the total tardiness with independent job setup time. A heuristic to the same problem, with the objective of minimizing the makespan is proposed by \cite{doi:10.1080/07408170304382} and \cite{doi:10.1080/00207540410001720745} suggests a two-phase heuristic algorithm with the objective of minimizing total tardiness.

For uniform machines, \cite{DBLP:journals/cor/KimL21} studies a variant with dedicated machines (there are some dedicated machines for each job), sequence-dependent setup times, and limited setup resources (jobs require setup operators that are limited) with the objective to minimize the makespan.

\textbf{Fractional bin-packing: } 
There are two main variants of fractional bin-packing. In bin-packing with size-increasing fragmentation (BP-SIF), each item may be fragmented; overhead units are added to the size of every fragment. In bin-packing with size-preserving fragmentation (BP-SPF) each item has a size and a cost; fragmenting an item increases its cost but does not change its size.

Shachnai et al. \cite{DBLP:journals/mst/ShachnaiTY08} develop approximation schemes for BP-SIF and BP-SPF; a dual {\sf PTAS} (a {\sf PTAS} for the dual version of the problem), an asymptotic {\sf PTAS} called {\sf APTAS} and a dual asymptotic {\sf FPTAS} called {\sf AFPTAS} for both versions.

Also, the fractional knapsack problem with penalties is recently introduced by \cite{DBLP:journals/eor/MalagutiMPP19}. They develop an {\sf FPTAS} and a dynamic program for the problem, and they show an extensive computational study comparing the performance of their models.

\section{Technical Details and Omitted Proofs}

\subsection{Variant of the \interval{\asym\decision}{$n=2,v$}{${\cal X}, {\cal R}$} with two $v$-s}\label{sec:vv-variant}
We need two new variants of the \interval{\asym\decision}{$n,v$}{${\cal X}, {\cal R}$} problem. We first define the \interval{\asym\decision}{$n=2, v_1, v_2$}{${\cal X}, {\cal R}$} problem as follows. For two fixed rationals $v_1, v_2$ and a set of ratios ${\cal R}=(r_1,r_2)$ 
such that $r_1v_1+r_2v_2 >0$,
\begin{quote}
	\interval{\asym\decision}{$n=2, v_1, v_2$}{${\cal X}, {\cal R}, m_1, m_2$}: \quad
	Decide if there exists a partition of ${\cal X}$ among two bins in which the relative sum of every bin $i$ is at most $(1+v_i) S$ where bin $q$ contains exactly $m_q$ items.
\end{quote}
If such a partition exists, we call it a \emph{$vv$-feasible} partition. We also need a variant with a critical coordinate,
\begin{quote}
    \interval{\asym\decision}{$n=2, v_1, v_2, i$}{${\cal X}, {\cal R}, m_1, m_2$}: \quad
	Minimize $\max(b_1,b_2)$ subject to $b_i \leq (1+v_i) \cdot S$ where $b_1, b_2$ are bin relative sums in a $2$-way partition of ${\cal X}$ where bin $q$ contains exactly $m_q$ items.
\end{quote}
We prove in 
the supplementary information
that \interval{\asym\decision}{$n=2, v_1, v_2, i$}{${\cal X}, {\cal R}, m_1, m_2$} has an {\sf FPTAS}.
Using similar techniques as in \Cref{alg:UniInterval}, we show a polynomial time algorithm for \interval{\asym\decision}{$n=2, v_1, v_2$}{${\cal X}, {\cal R}, m_1, m_2$}. 

\begin{lemma} \label{lem:twottuniform}
	Let $\epsilon=(r_1v_1+r_2v_2)/(r_1+r_2)>0$.
	If for every $i \in \{1,2\}$ and $j = 2/i$ , {\sf FPTAS}$($\interval{\asym\decision}{$n=2, v_1, v_2, i$}{${\cal X}, {\cal R}, m_1, m_2$}, $\epsilon) > (1+v_j) \cdot S$, then there is no $vv$-feasible $2$-way partition of ${\cal X}$.
\end{lemma}

\begin{proof}
	First set $i=1$ and $j=2$. If {\sf FPTAS}$($\interval{\asym\decision}{$n=2, v_1, v_2, i$}{${\cal X}, {\cal R}, m_1, m_2$}, $\epsilon)$ is greater than $r_2(1+v_2)\cdot S$, by definition of {\sf FPTAS}, 
    \begin{align*}
         \text{{\sf FPTAS}$($\interval{\asym\decision}{$n=2, v_1, v_2, i$}{${\cal X}, {\cal R}, m_1, m_2$}, $\epsilon)$} > r_2(1+v_2)\cdot S / (1+\epsilon) \\
         = r_2(1+v_2)\cdot S / (1+(r_1v_1 + r_2v_2)/(r_1+r_2)) = r_2 \frac{1+v_2}{r_1(1+v_1)+r_2(1+v_2)} \cdot (r_1+r_2)S.
    \end{align*}
	Similarly, for $i=2$ and $j=1$, we have 
    \begin{align*}
         \text{{\sf FPTAS}$($\interval{\asym\decision}{$n=2, v_1, v_2, i$}{${\cal X}, {\cal R}, m_1, m_2$}, $\epsilon)$} \\ > r_1 \frac{1+v_1}{r_1(1+v_1)+r_2(1+v_2)} \cdot (r_1+r_2)S.
    \end{align*}
    Therefore, in any partition, if $r_1b_1 \leq r_1(1+v_1) \cdot S$ then, 
    $$r_2b_2 > r_2 \frac{1+v_2}{r_1(1+v_1)+r_2(1+v_2)} \cdot (r_1+r_2)S,$$
    and if $r_2b_2 \leq r_2(1+v_2) \cdot S$ then, 
    $$r_1b_1 > r_1 \frac{1+v_1}{r_1(1+v_1)+r_2(1+v_2)} \cdot (r_1+r_2)S.$$
    Since a $vv$-feasible partition requires both $r_1b_1 \leq r_1(1+v_1) \cdot S$ and $r_2b_2 \leq r_2(1+v_2) \cdot S$, we must have 
	\begin{align*}
		r_1b_1 + r_2b_2 & > r_1 \frac{1+v_1}{r_1(1+v_1)+r_2(1+v_2)} \cdot (r_1+r_2)S \\
        & ~~~~~~~~~~~~~~~~~~~~~~~~~~~~~~ + r_2 \frac{1+v_2}{r_1(1+v_1)+r_2(1+v_2)} \cdot (r_1+r_2)S \\
		& = \frac{r_1(1+v_1)+r_2(1+v_2)}{r_1(1+v_1)+r_2(1+v_2)} \cdot (r_1+r_2)S \\
		& = (r_1+r_2)S, 
	\end{align*}
	by the problem definition, the sum $r_1b_1 + r_2b_2 = (r_1+r_2)S$, so we reach a contradiction.
\end{proof}

We are ready to design our algorithm.

\begin{algorithm}[!h] \caption{\qquad\interval{\asym\decision}{$n=2, v_1, v_2$}{${\cal X}, {\cal R}, m_1, m_2$}} \label{alg:UnittInterval}
	\begin{algorithmic}[1]
		\For {$i \in \{1, 2\}$} 
		\State Set $j := 2/i$.
		\State 
		Run the FPTAS for \interval{\asym\decision}{$n=2, v_1, v_2, i$}{${\cal X}, {\cal R}, m_1, m_2$} with critical coordinate $i$ and $\epsilon= (r_1v_1+r_2v_2)/(r_1+r_2)$.
		\State 
		Let $b_j$ be the relative sum of bin $j$ as returned by this FPTAS.
		\If{$b_j \le (1+v_j)\cdot S$} 
		\State {\sf return} ``yes''
		\EndIf
		\EndFor
		\State {\sf return} ``no''
	\end{algorithmic}
\end{algorithm}	

\begin{theorem}\label{thm:UnittInterval}
	For any rationals $v_1, v_2$ and every ratios $r_1,r_2$, with $r_1v_1+r_2v_2>0$, \Cref{alg:UnittInterval} solves the problem \interval{\asym\decision}{$n=2, v_1, v_2$}{${\cal X}, {\cal R}, m_1, m_2$} in time $O(\text{poly}(m, \log{(\sumratios{} S)}, (r_1+r_2)/(r_1v_1+r_2v_2)))$, where $m=m_1+m_2$ is the number of input items in ${\cal X}$.
\end{theorem}

\begin{proof}
	The run-time of \Cref{alg:UnittInterval} is dominated by the run-times of the two {\sf FPTAS}-s for the problem \interval{\asym\decision}{$n=2, v_1, v_2, i$}{${\cal X}, {\cal R}, m_1, m_2$}, which by definition of {\sf FPTAS} is in $O(\text{poly}(m,\log{(\sumratios{}S)},1/\epsilon))$, which is  $O(\text{poly}(m,\log{(\sumratios{}S)},(r_1+r_2)/(r_1v_1+r_2v_2)))$ (we show in \Cref{sec:running-time2} that the exact run-time is $O(\frac{m}{r_1v_1+r_2v_2} \log{(\sumratios{}S)})$).
	
	If $b_j$, the returned relative bin sum of one of the {\sf FPTAS}(\interval{\asym\decision}{$n=2, v_1, v_2, i$}{${\cal X}, {\cal R}, m_1, m_2$}, $\epsilon= (r_1v_1+r_2v_2)/(r_1+r_2)$), is at most $(1+v_j)\cdot S$, then the partition found by the {\sf FPTAS} is  $vv$-feasible, so \Cref{alg:UnittInterval} answers ``yes'' correctly (the number of items in the bins is respected by the {\sf FPTAS}).
	Otherwise, by \Cref{lem:twottuniform}, there is no $vv$-feasible partition, 
	and \Cref{alg:UnittInterval} answers ``no'' correctly. 
\end{proof}

\subsection{Properties of instances with at least two almost-full bins (\Cref{lem:structure})}

\subsubsection{Simple example}
\label{sec:example-lem}
In this section, we provide a simple example of \Cref{lem:structure}.

We show an instance with a solution with two almost-full bins, and no solution with at most one almost-full bin.
There are $n=5$ identical bins (so $\sumratios{}=n=5$), $u=n-2=3$, ${\cal X} = (n + 2) \cdot [1] = (1, 1, 1, 1, 1, 1, 1)$. Then, with this setting, $m=7,~ 5M=1,~ 5S=7,~ S+uM = 7/5+3/5=2$. Now, we can compute $v$ and $\epsilon$: $v = u M / S = (3/5) /(7/5)=3/7$ and $\epsilon = v / (4 \cdot m^2 ) = 3/1372$. That is, an almost-full bin is a bin with a sum larger than $(1 + v) \cdot S / (1 + \epsilon) = ((10/7) \cdot (7/5) )/ (1375/1372) = 2/(1375/1372) = 2744/1375 < 2$. Also, a $u$-feasible solution of the instance is a solution where the largest bin sum is at most $(1+v)S = 2$. Therefore, one possible partition is $P_1 = ((1,1),(1,1),(1),(1),(1))$. This partition is $u$-possible since the largest bin sum equals $2$ as required. It has two almost-full bins: bins 1 and 2. 
Note that there is also a partition with three almost-full bins:
$((1,1),(1,1),(1,1),(1),())$.
However, there is no solution with at most one almost-full bin.
Let us see that $P_1$ satisfies all the properties in the structure Lemma.
\begin{enumerate}[label=(\alph*)]
		\item 
		Exactly two bins  are almost-full: bins 1 and 2.
		
		\item 
		The sum of bins 3,4,5 (not-almost-full) satisfies
		\begin{align*}
		\left(1 - \frac{2}{n-2} v- 2 \epsilon \right) \cdot S
		& \leq
		b_3,b_4,b_5
		\leq
		\left(1 - \frac{2}{n-2}v + (n - 1) 2\epsilon \right) \cdot S \\
		(1 - 2/7  - 3/686) \cdot 7/5
		& \leq
		b_3,b_4,b_5
		\leq
		(1 - 2/7  + 6/342) \cdot 7/5
		\\
		0.99387755102
		& \leq
		b_3,b_4,b_5
		\leq
		1.02456140351,
		\end{align*}
		which is true since $b_3=b_4=b_5=1$
		
		\item 
		Every almost-full bin contains only \emph{big-items}, items largest than $ n S (\frac{v}{n-2} - 2 \epsilon) = 7 (1/7 - 3/686) = 0.9693877551$. Since every item equals  1, the claim holds. 

		\item 
		Every not-almost-full bin contains big-items that are larger or equal to every item in bin 1,2, or \emph{small-items}. Since all items are big and have the same size, the claim holds trivially. 
		
		\item 
        All items are big, so $m_b = m = 7$. Denote  by $m_i$ the number of big-items in bin $i$. 
        So $m_1=m_2=\lceil (m_b + n - 3) / \sumratios{} + 1 / \sumratios{} \rceil= \lceil(7 + 5-3) / 5  + 1/5 \rceil = 2$, which is true.
	\end{enumerate}

Keeping this setting, we can increase or decrease the number of bins, recompute every parameter according to the new $n$, and reach the same conclusion. We built a generic example, such that if we modify $n$, all the other parameters can be recomputed because they are only dependent on $n$. For example, for $n=4$ identical bins, we will have $u=n-2=4, {\cal X} = (n + 2) \cdot [1] = [1, 1, 1, 1, 1, 1]$. Therefore, we reach the same conclusions with the solution $[1,1],[1,1],[1], [1]$.

\subsubsection{Detailed proof} \label{sec:appendix-lem-struct}
In this section, we provide a detailed proof of \Cref{lem:structure}.

Recall that, by the lemma assumption, there exists a 
$v$-feasible partition with some $p\geq 2$ almost-full bins $1,\ldots, p$.
We take an arbitrary such partition and convert it using a sequence of transformations to another $v$-feasible partition satisfying the properties stated in the lemma, as explained below.

\begin{proof}[\textbf{Proof of \Cref{lem:structure}, property \ref{prop:2-almost-full}}]

Pick a $v$-feasible partition
with some $p\geq 3$ almost-full bins.
Necessarily $p\leq n-1$, otherwise each bin has a bin relative sum larger than $S$ (since it is an almost-full bin), implying that the sum of the items is larger than $\sumratios{} S$ which is a contradiction.
It is sufficient to show that there exists a $v$-feasible partition with $p-1$ almost-full bins; then we can proceed by induction down to 2.
Call the almost-full bins $1,\dots,p$ 
and the not-almost-full bins  $p+1,\dots,n$. 
By definition of almost-full bins,
\begin{align*}
r_1b_1 + \dots + r_pb_p \geq (r_1+\dots+r_p)(1+ v - 2\epsilon)\cdot S,
\end{align*}
so,
\begin{align*}
	r_{p+1}b_{p+1} + \dots + r_nb_n & = \sumratios{}S - (r_{1}b_{1} + \dots + r_pb_p) \\
    & < \sumratios{}S - (r_1+\dots+r_p)(1+ v - 2\epsilon)\cdot S \\
    & = (r_{p+1} + \dots + r_n - (r_1+\dots+r_p) (v - 2\epsilon))\cdot S.
\end{align*}
Let $\Delta$ be the total free space in the not-almost-full machines:
\begin{align*}
    \Delta & = \sum^n_{i=p+1}{r_i(1+v)S} - (r_{p+1}b_{p+1} + \dots + r_nb_n)  \\
    & > (1+v) S (r_{p+1} + \dots + r_n) - (r_{p+1} + \dots + r_n - (r_1+\dots+r_p) (v - 2\epsilon))\cdot S \\
    & = ((1+v) (r_{p+1} + \dots + r_n) +  (r_1+\dots+r_p) (v - 2\epsilon) - (r_{p+1} + \dots + r_n))\cdot S \\
    & = ((r_{p+1} + \dots + r_n)\cdot v +  (r_1+\dots+r_p) (v - 2\epsilon))\cdot S \\
    & = (\sumratios{} v - (r_1+\dots+r_p) 2\epsilon)\cdot S \\
\end{align*}
Denote by $\Delta_{k}$ the maximum
free space in a non-almost-full bin (breaking ties arbitrarily). 
By the pigeonhole principle, 
\begin{align*}
\Delta_{k} & \geq 
\frac{(\sumratios{} v - (r_1+\dots+r_p) 2\epsilon)\cdot S }{n-p} \\
& \geq 
\frac{(\sumratios{} v - (r_1+\dots+r_p) 2\epsilon)\cdot S}{n-3} && \text{(since $p \geq 3$)} \\
& \geq 
\frac{(\sumratios{} v - (r_1+\dots+r_{n-1}) 2\epsilon)\cdot S}{n-3} && \text{(since $p \leq n-1$)} \\
\end{align*}

Assume that bin $k$ is the bin with the largest free space $\Delta_{k}$. By assumption, each item is at most $\sumratios{} M$. 
Suppose we take one item from bin $p$, and move it to bin $k$. Then after the transformation:
\begin{align*}
b_k r_k & \leq r_k(1+v)S - \Delta_{k} + \sumratios{}M \\ 
& = r_k(1+v)S - \frac{(\sumratios{} v - (r_1+\dots+r_{n-1}) 2\epsilon)\cdot S}{n-3} + \sumratios{}M \\
& \leq r_k(1+v)S - \frac{(\sumratios{} v - (r_1+\dots+r_{n-1}) 2\epsilon)\cdot S}{n-3} + \frac{\sumratios{}vS}{n-2}  \\
& ~~~~~~~~~~~~~~~~~~~~~~~~~~~~~~~~~~~~~~~~~~~~~~~~~~~~~~  \text{(since $uM=vS$ and $u \geq n-2$ )} \\
& = (r_k(1+v) - \frac{(\sumratios{} v - (r_1+\dots+r_{n-1}) 2\epsilon)}{n-3} + \frac{\sumratios{}v}{n-2}) \cdot S \\
& = (r_k(1+v) - \frac{\sumratios{}v}{(n-3)(n-2)} + \frac{r_1+\dots+r_{n-1}}{n-3} 2\epsilon) \cdot S \\
& = r_{k}(1+v - \frac{\sumratios{}v}{r_{k}(n-3)(n-2)} + \frac{r_1+\dots+r_{n-1}+r_{k}(n-3)}{r_{k}(n-3)}  2\epsilon - 2 \epsilon) \cdot S  \\
& < r_{k}(1+v - \frac{\sumratios{}v}{r_{k}(n-3)(n-2)} + \frac{n \sumratios{}}{r_{k}(n-3)}  2\epsilon - 2 \epsilon) \cdot S  \\
& \leq r_{k}(1 + v - \frac{\sumratios{}v}{r_{k}(n-3)(n-2)} + \frac{\sumratios{}v}{r_k(n-3)(n-2)} - 2\epsilon) \cdot S  \\
& ~~~~~~~~~~~~~~~~~~~~ \text{(since $\epsilon = \frac{v}{4n^2m^2} \leq \frac{v}{2n(n-2)}$)} \\
& < r_k(1 + v - 2\epsilon) \cdot S.
\end{align*}
So, bin $k$ remains not-almost-full. If bin $p$ is still almost-full, then we are in the same situation: we have exactly $p$ almost-full bin.
So we can repeat the argument, and move another item from bin $p$ to the (possibly different) bin with the largest free space. Since the number of items in bin $u$ decreases with each step, eventually, it becomes not-almost-full. 
\end{proof}
This process continues until we have a partition with exactly two almost-full bins, which we assume to be bins 1 and 2.
The not-almost-full bins are bins $3,\ldots,n$.
We now transform the partition so that the sum in each bin $3,\ldots,n$ is bounded within a small interval.

\begin{proof}[\textbf{Proof of \Cref{lem:structure}, property \ref{prop:not-almost-full}}]
	
For proving the lower bound, suppose there is some not-almost-full bin $i \in \{3,\ldots,n\}$
with $b_i < \left(r_i + \left(r_i - \frac{\sumratios{}}{n-2} \right)v - r_i 2\epsilon \right) \cdot S$.
Move an item from one almost-full bin to bin $i$. Its new sum satisfies
\begin{align*}
r_i b_i & < \left(r_i + \left( r_1 - \frac{\sumratios{}}{n-2}\right)v - r_i 2\epsilon \right) \cdot S + \sumratios{}M \\
& =  \left(r_i + \left(r_i -  \frac{\sumratios{}}{n-2} \right)v - r_i 2\epsilon \right) \cdot S  + \frac{\sumratios{}vS}{n-2} && \text{(Since $u\leq n-2$)} \\
& \leq \left(r_i + \left(r_i - \frac{\sumratios{}}{n-2} + \frac{\sumratios{}}{n-2}  \right)v - r_i 2\epsilon \right) \cdot S  \\
& = r_i (1 + v - 2\epsilon) \cdot S,
\end{align*}
so bin $i$ is still not almost-full.
Assumption \ref{prop:no-1-almost-full} implies that bins $1$ and $2$ must still be almost-full. So we are in the same situation and can repeat the argument until the lower bound is satisfied.
 
The upper bound on $r_i b_i$ is proved by subtracting the lower bounds of the other $n-1$ bins from the sum of all items.
The total length of items allocated to the $n-3$ other not-almost-full bins is at least:
\begin{align*}
\sum_{j \in R \setminus \{1, 2, i\}} & \left(\left(r_j + \left(r_j - \frac{\sumratios{}}{n-2} \right)v - r_j 2\epsilon \right) \cdot S \right) \\
	& = S \left ( \sum_{j \in R \setminus \{1, 2, i\}} \left(r_j + \left(r_j - \frac{\sumratios{}}{n-2}  \right)v - r_j 2\epsilon \right) \right) \\
	& = S \left ( \sum_{j \in R \setminus \{1, 2, i\}} (r_j + r_j v  - r_j 2\epsilon ) - \sum_{j \in R \setminus \{1, 2, i\}} \frac{\sumratios{} v }{n-2}  \right)  \\
	& = S \left ( \sum_{j \in R \setminus \{1, 2, i\}} r_j (1 + v - 2\epsilon ) - (n - 3) \frac{\sumratios{} v }{n-2}  \right)  \\
	& =  (1 + v - 2\epsilon ) \cdot S \sum_{j \in R \setminus \{1, 2, i\}} r_j  -   (n - 3) \frac{\sumratios{} v }{n-2} \cdot S \\
	& =  (1 + v - 2\epsilon ) \cdot S \cdot (\sumratios{}-r_1-r_2-r_i)  -   (n - 3) \frac{\sumratios{} v }{n-2} \cdot S 
\end{align*}
And the sum of items allocated to the two almost-full bins is:
\begin{align*}
r_1b_1 + r_2b_2
>(r_1 + r_2)(1 + v - 2 \epsilon) \cdot S && \text{(by almost-full bin sum)} 
\end{align*}
Therefore, 
\begin{align*}
	r_i b_i & \leq  \sumratios{}S - (r_1 + r_2)(1 + v - 2 \epsilon) \cdot S  +  (n - 3) \frac{\sumratios{} v }{n-2} \cdot S \\
    & ~~~~~~~~~~~~~~~~~~~~~~~~~~~~~~~~~~~~~~~~~~~~ - (1 + v - 2\epsilon )  (\sumratios{} - (r_1 + r_2 + r_i)) \cdot S \\
	& =  (\sumratios{} - (r_1 + r_2)(1 + v - 2 \epsilon)  + (n - 3) \frac{\sumratios{} v }{n-2} - (1 + v - 2\epsilon )  (\sumratios{} - (r_1 + r_2 + r_i)) )\cdot S \\
    & =  (\sumratios{} - (r_1 + r_2 + \sumratios{} - (r_1 + r_2 + r_i))(1 + v - 2 \epsilon)  + (n - 3) \frac{\sumratios{} v }{n-2})\cdot S \\
    & =  (\sumratios{} - (\sumratios{} - r_i)(1 + v - 2 \epsilon)  + (n - 3) \frac{\sumratios{} v }{n-2})\cdot S \\
    & =  (r_i + (r_i - \sumratios{} + \frac{(n - 3)\sumratios{}}{n-2})v + (\sumratios{} - r_i) 2 \epsilon) \cdot S \\
    & = (r_i + \left(r_ i - \frac{\sumratios{}}{n-2} \right) v + (\sumratios{} - r_i)  2\epsilon)\cdot S,    
	\end{align*}
	so property \ref{prop:not-almost-full} is satisfied. 
\end{proof}


Next, we transform the partition such that the almost-full bins $1$ and $2$ contain only ``big-items'', which we defined as items larger than 
$\sumratios{}S \left (\frac{v}{n-2} - 2 \epsilon \right)$.

\begin{proof}[\textbf{Proof of \Cref{lem:structure}, property \ref{prop:big-items}}]
Suppose bin 1 or 2 contains an item smaller than $(\frac{v}{n-2} - 2\epsilon) \cdot \sumratios{} S$. We move it to some not-almost-full bin $i$, $i\geq 3$. By property \ref{prop:not-almost-full}, the new sum of items allocated to bin $i$ satisfies
	\begin{align*}
	r_i b_i & \leq 
    \left(r_i + \left(r_i - \frac{\sumratios{}}{n-2} \right) v + (\sumratios{} - r_i) 2\epsilon\right)\cdot S + (\frac{v}{n-2} - 2\epsilon) \cdot \sumratios{} S \\
    & = \left(r_i + \left(r_i - \frac{\sumratios{}}{n-2} + \frac{\sumratios{}}{n-2} \right) v + (\sumratios{} - \sumratios{} - r_i) 2\epsilon\right) \cdot S \\
	& = (1+v-2\epsilon) S,
	\end{align*}
	so bin $i$ remains not-almost-full. 
	Assumption \ref{prop:no-1-almost-full} implies that bins 1 and 2 remain almost-full. So properties \ref{prop:2-almost-full}, \ref{prop:not-almost-full} still hold, and we can repeat the argument.
	Each move increases the number of items on the not-almost-full bins, so the process must end, and eventually all items smaller than $(\frac{v}{n-2} - 2\epsilon) \cdot \sumratios{} S$ are allocated to bins $3,\ldots,n$.
\end{proof}

Next, we transform the partition so that the not-almost-full bins contain only the largest big-items and all the small-items.

\begin{proof}[\textbf{Proof of \Cref{lem:structure}, property \ref{prop:split-big-small-items}}]	
    Suppose there exists some $i \in \{3, \dots, n\}$ for which bin $i$ contains an item $x$ such that $x \geq 2 \epsilon \sumratios{}S$, but $x < y$ for some item $y$ contained by bin $1$ or $2$. We exchange $x$ and $y$.
    The sum of the items allocated to bin $i$ increases by at most $\sumratios{}M - 2\epsilon \sumratios{} S$:
    \begin{align*}
    	r_i b_i & \leq (r_i + \left(r_i - \frac{\sumratios{}}{n-2} \right) v + (\sumratios{} - r_i) 2\epsilon)\cdot S + \sumratios{}M - 2 \epsilon \sumratios{} S \\
    	& \leq (r_i + \left(r_i - \frac{\sumratios{}}{n-2} \right) v + (\sumratios{} - r_i) 2\epsilon) \cdot S + \frac{\sumratios{}vS}{n-2} - 2 \sumratios{} S \epsilon \\
        & ~~~~~~~~~~~~~~~~~~~~~~~~~~~~~~~~~~~~~~~~~~~~~~~~~~~~~~~~ \text{(since $M = vS/u$ and $u \geq n-2$)} \\
    	& = (r_i + \left(r_i - \frac{\sumratios{}}{n-2} + \frac{\sumratios{}}{n-2} \right) v + (\sumratios{} - r_i - \sumratios{}) 2\epsilon) \cdot S \\
    	& = r_i \cdot (1 + v - 2\epsilon)\cdot S,
    \end{align*}
    so bin $i$ is still not almost-full. So by \ref{prop:no-1-almost-full}, bins $1$ and $2$ must remain almost-full,
    and we can repeat the argument.
    Each move decreases the sum of the items at bins $1,2$, so the process must end.
\end{proof}

\begin{proof}[\textbf{Proof of \Cref{lem:structure}, property \ref{prop:litems}}]

By definition of almost-full, the sum of the items in bin $i$ is at least:
\begin{align*}
r_i (1+v- 2\epsilon) S = 
r_i (S + (n-2)(\frac{v}{n-2}-\frac{2\epsilon}{n-2})S) \geq
r_i (S + (n-2)(\frac{v}{n-2}-2\epsilon)S).
\end{align*}
The sum of all items must be at least the number of big-items times the smallest big-item, so,
\begin{align*}
	S \geq m_b \cdot (\frac{v}{n-2}-2\epsilon)S.
\end{align*}
Substituting above gives that the total sum of the items in an almost full bin is at least:
\begin{align*}
r_i (m_b\cdot (\frac{v}{n-2}-2\epsilon)S + (n-2)(\frac{v}{n-2}-2\epsilon)S)
= 
r_i (\frac{v}{n-2}-2\epsilon)(m_b+n-2)S.
\end{align*}
Since the largest item is $\sumratios{} M \leq \frac{v \sumratios{} S}{n-2}$, the number of items contained in an almost full bin is at least:

\begin{align*}
m_i \geq &
\frac{r_i (\frac{v}{n-2}-2\epsilon)(m_b+n-2)S}{\frac{v \sumratios{} S}{n-2}}
=
\frac{r_i (v-2\epsilon(n-2))}{v \sumratios{}} (m_b+n-2)
\\
=&
\frac{r_i}{\sumratios{}}(1-\frac{n-2}{2n^2m^2})(m_b+n-2) ~~~~~~~~~~~~~~~~~~~~ \text{(since $v = 4 \epsilon n^2 m^2$)}
\\
=& 
\frac{r_i}{\sumratios{}}(m_b+n-2)
- 
\frac{r_i}{\sumratios{}}\frac{(m_b+n-2)(n-2)}{2 n^2m^2}
\\
\geq & 
\frac{r_i}{\sumratios{}}(m_b+n-2)
- 
\frac{r_i}{\sumratios{}}\frac{1}{ m}  ~~~~~~~~~~ \text{(since $m_b+n-2\leq 2 m$ and $n-2 \leq n^2$)}
\\
= & 
\frac{r_i}{\sumratios{}}(m_b+n-2 - 1/m)
\\
 > & 
\frac{r_i}{\sumratios{}}(m_b+n-3),
\end{align*}
so $m_i \sumratios{} > r_i (m_b+n-3)$.
As both sides are integers, 
$m_i \sumratios{} \geq r_i (m_b+n-3) + 1$,
so $m_i \geq r_i (m_b+n-3)/\sumratios{} + 1/\sumratios{}$.

Based on the previous properties, we know that every not-almost-full bin will become almost-full if we add into it a single big-item. Therefore,
\begin{align*}
	m_1 &\geq r_1 (m_b+n-3) /\sumratios{} + 1/\sumratios{}
	\\
	m_2 &\geq r_2 (m_b+n-3) /\sumratios{} + 1/\sumratios{}
	\\
	m_i+1 &\geq r_i (m_b+n-3) /\sumratios{} + 1/\sumratios{}&& \text{(for every $i \in \{3, \dots, n\}$)}
\end{align*}
Since $m_b = m_1 + m_2 + \dots+ m_n$, we have:
\begin{align*}
   m_1 + m_2 & = m_b - m_3 - \dots - m_n \\
   & \leq m_b - \frac{r_3}{\sumratios{}} (m_b+n-3)- 1/\sumratios{} + 1 - \dots - \frac{r_n}{\sumratios{}} (m_b+n-3) - 1/\sumratios{} + 1 \\
   & = m_b - \frac{r_3}{\sumratios{}} m_b - \dots - \frac{r_n}{\sumratios{}} m_b  - \frac{r_3}{\sumratios{}} (n-3) - \dots - \frac{r_n}{\sumratios{}} (n-3) - \frac{n-2}{\sumratios{}} + (n-2) \\
   & = m_b\left(1 - \frac{r_3}{\sumratios{}} - \dots - \frac{r_n}{\sumratios{}}\right) + (n-3) \left(-\frac{r_3}{\sumratios{}}- \dots - \frac{r_n}{\sumratios{}}\right) - \frac{n-2}{\sumratios{}} + (n-3) +1\\
   & = m_b \left(\frac{\sumratios{} - r_3 - \dots - r_n}{\sumratios{}}\right) + (n-3) \left(\frac{\sumratios{}-r_3 - \dots - r_n}{\sumratios{}}\right) - (n-2)/\sumratios{} + 1 \\
   & = m_b \left(\frac{r_1+r_2}{\sumratios{}}\right) + (n-3) \left(\frac{r_1+r_2}{\sumratios{}}\right) - \frac{n-2}{\sumratios{}} +1 \\
   & = (m_b + n -3) \left(\frac{r_1+r_2}{\sumratios{}}\right) - \frac{n-2}{\sumratios{}} +1.
\end{align*}
Therefore,
\begin{align*}
    m_1 & \leq (m_b + n -3) \left(\frac{r_1+r_2}{\sumratios{}}\right) - \frac{n-2}{\sumratios{}} +1 - m_2 \\
    & \leq (m_b + n -3) \left(\frac{r_1+r_2}{\sumratios{}}\right) + \frac{2-n}{\sumratios{}} +1 - r_2 (m_b+n-3) /\sumratios{} - 1/\sumratios{} \\
    & = (m_b + n - 3) \frac{r_1}{\sumratios{}} + \frac{1-n}{\sumratios{}} +1 \\
    & = (m_b + n - 3) \frac{r_1}{\sumratios{}} + \frac{1}{\sumratios{}} + \frac{\sumratios{}-n}{\sumratios{}}.
\end{align*}
So, we have:
\begin{align*}
r_1 (m_b+n-3) /\sumratios{} + 1/\sumratios{} 
\leq
m_1
\leq
r_1(m_b + n - 3)/\sumratios{} + 1/\sumratios{} + \frac{\sumratios{}-n}{\sumratios{}}.
\end{align*}
Similarly,
\begin{align*}
r_2 (m_b+n-3) /\sumratios{} + 1/\sumratios{} 
\leq
m_2
\leq
r_2(m_b + n - 3)/\sumratios{} + 1/\sumratios{} + \frac{\sumratios{}-n}{\sumratios{}}.
\end{align*}
The difference between the upper bound and the lower bound is $\frac{\sumratios{}-n}{\sumratios{}}$, which is less than 1. Therefore, since $m_1$ and $m_1$ are natural, there is at most one possible value, proving the claim.
\end{proof}

\section{Running time}\label{sec:running-time}

We compute in this section the running time of \Cref{alg:UniInterval}, \Cref{alg:tfeasible} and \Cref{alg:UnittInterval}. 

The running time of each {\sf FPTAS} is linear in the number of states in the dynamic program. The number of states in the dynamic program is bounded by $L^n$, where $L$ is the number of possible values in every element of the vector. In \cite{woeginger2000does}[Section 3], Woeginger proves that $L \leq \lceil (1+\frac{2 g m}{\epsilon}) \pi_1(m, \log_2 \Bar{x}) \rceil$, 
where $m \in \mathbb{N}$ is the number of items,
$\Bar{x}$ is the sum of all item sizes (which in our case is just $\sumratios{}S$),
and $\pi_1$ is a
function describing the binary length of the values in the state vectors (in our case it is a linear function).

In each {\sf FPTAS} settings, $g = 1$, $\Bar{x} = \sum_{k=1}^m \sum_{i=1}^{\alpha} x_{i,k}$, where $\alpha = 1$, then, $\Bar{x}$ is the sum of all the items. Because every value in the state-vectors is at most the sum of all values, we have the linear function $\pi_1(m, \log_2 \Bar{x}) = \log_2 \Bar{x}$. 

\subsection{\Cref{alg:UniInterval} running time}\label{sec:running-time1}

\Cref{alg:UniInterval} runs an {\sf FPTAS} for 
\partition{\asym\minmax}{$n=2, v, i$}{${\cal X}, {\cal R}$} with  $\epsilon = v/2$ and the sum of the items $\Bar{x}$, is equal to $\sumratios{}S$. Therefore, we have the following:
$L \leq \lceil (1+\frac{2m}{v/2}) \log_2 (\sumratios{}S) \rceil = \lceil (1 + \frac{4m}{v}) \log_2 (\sumratios{}S) \rceil \in O(\frac{m}{v} \log{\sumratios{}S})$.

Therefore, \Cref{alg:UniInterval} running time is $O(\frac{m}{v} \log{\sumratios{}S})$.

\subsection{\Cref{alg:tfeasible} running time}\label{sec:running-time1.1}

\Cref{alg:tfeasible} runs three {\sf FPTAS}, one for \partition{\asym\minmax}{$n$}{${\cal X}, {\cal R}$}, one for \partition{\asym\minmax}{$n, v, i$}{${\cal X}, {\cal R}$} and one for \partition{\asym\minmax}{$n=2, v_1, v_2, i$}{${\cal X}, {\cal R}$} with different settings. 
\begin{itemize}
    \item For \partition{\asym\minmax}{$n$}{${\cal X}, {\cal R}$} and \partition{\asym\minmax}{$n, v, i$}{${\cal X}, {\cal R}$}, $\epsilon = v/(4n^2m^2)$ and the sum of the items $\Bar{x}$, is equal to $\sumratios{}S$.
    \item For \partition{\asym\minmax}{$n=2, v_1, v_2, i$}{${\cal X}, {\cal R}$}, $\epsilon = (r_1v_1+r_2v_2)/(r_1+r_2)$ and the sum of the items $\Bar{x}$, is equal to $\sumratios{}S$.
\end{itemize}
Therefore, we have the following:
\begin{itemize}
    \item For \partition{\asym\minmax}{$n$}{${\cal X}, {\cal R}$} and \partition{\asym\minmax}{$n, v, i$}{${\cal X}, {\cal R}$}, $L \leq \lceil (1+\frac{2m}{v/(4n^2m^2)}) \log_2 (\sumratios{}S) \rceil = \lceil (1 + \frac{8n^2m^3}{v}) \log_2 (\sumratios{}S) \rceil \allowbreak \in O(\frac{m^3}{v} \log{\sumratios{}S})$, resulting in the running time 
    \begin{align*}
    O(\frac{m^3}{v} \log{\sumratios{}S}) 
    & = O(\frac{m^3 \cdot S}{uM} \log{\sumratios{}S})
    && \text{since $v=uM/S$}
    \\
    & = O(\frac{m^3 \cdot S}{M} \log{\sumratios{}S})
    && \text{since $u \geq n-2$}
    \\ 
    & = O(\frac{m^3 \cdot mM}{M} \log{\sumratios{}S})
    &&  \text{since $m M \geq S$}
    \\ 
    & = O(m^4 \log{\sumratios{}S}).
    \end{align*}
    \item For \partition{\asym\minmax}{$n=2, v_1, v_2, i$}{${\cal X}, {\cal R}$}, the running time is computed in \Cref{sec:running-time2} and it is $O(\frac{m}{r_1v_1+r_2v_2} \log{\sumratios{}S})$.
\end{itemize}

Therefore, \Cref{alg:tfeasible} running time is $O(m^4 \log{\sumratios{}S})$.

\subsection{\Cref{alg:UnittInterval} running time}\label{sec:running-time2}

\Cref{alg:UnittInterval} runs an {\sf FPTAS} for \partition{\asym\minmax}{$n=2, v_1, v_2, i$}{${\cal X}, {\cal R}$}. In \partition{\asym\minmax}{$n=2, v_1, v_2, i$}{${\cal X}, {\cal R}$}, $\epsilon = (r_1v_1+r_2v_2)/(r_1+r_2)$ and the sum of the items $\Bar{x}$, is equal to $\sumratios{}S$. Therefore, we have the following:
$L \leq \lceil (1+\frac{2m}{(r_1v_1+r_2v_2)/(r_1+r_2)}) \log_2 (\sumratios{}S) \rceil = \lceil (1 + \frac{2(r_1+r_2)m}{(r_1v_1+r_2v_2)}) \log_2 (\sumratios{} S) \rceil \in O(\frac{m}{r_1v_1+r_2v_2} \log{\sumratios{} S})$.

Therefore, \Cref{alg:UnittInterval} running time is $O(\frac{m}{r_1v_1+r_2v_2} \log{\sumratios{} S})$.

\end{appendices}

\bibliography{references}


\begin{thebibliography}{47}
\ifx \bisbn   \undefined \def \bisbn  #1{ISBN #1}\fi
\ifx \binits  \undefined \def \binits#1{#1}\fi
\ifx \bauthor  \undefined \def \bauthor#1{#1}\fi
\ifx \batitle  \undefined \def \batitle#1{#1}\fi
\ifx \bjtitle  \undefined \def \bjtitle#1{#1}\fi
\ifx \bvolume  \undefined \def \bvolume#1{\textbf{#1}}\fi
\ifx \byear  \undefined \def \byear#1{#1}\fi
\ifx \bissue  \undefined \def \bissue#1{#1}\fi
\ifx \bfpage  \undefined \def \bfpage#1{#1}\fi
\ifx \blpage  \undefined \def \blpage #1{#1}\fi
\ifx \burl  \undefined \def \burl#1{\textsf{#1}}\fi
\ifx \doiurl  \undefined \def \doiurl#1{\url{https://doi.org/#1}}\fi
\ifx \betal  \undefined \def \betal{\textit{et al.}}\fi
\ifx \binstitute  \undefined \def \binstitute#1{#1}\fi
\ifx \binstitutionaled  \undefined \def \binstitutionaled#1{#1}\fi
\ifx \bctitle  \undefined \def \bctitle#1{#1}\fi
\ifx \beditor  \undefined \def \beditor#1{#1}\fi
\ifx \bpublisher  \undefined \def \bpublisher#1{#1}\fi
\ifx \bbtitle  \undefined \def \bbtitle#1{#1}\fi
\ifx \bedition  \undefined \def \bedition#1{#1}\fi
\ifx \bseriesno  \undefined \def \bseriesno#1{#1}\fi
\ifx \blocation  \undefined \def \blocation#1{#1}\fi
\ifx \bsertitle  \undefined \def \bsertitle#1{#1}\fi
\ifx \bsnm \undefined \def \bsnm#1{#1}\fi
\ifx \bsuffix \undefined \def \bsuffix#1{#1}\fi
\ifx \bparticle \undefined \def \bparticle#1{#1}\fi
\ifx \barticle \undefined \def \barticle#1{#1}\fi
\bibcommenthead
\ifx \bconfdate \undefined \def \bconfdate #1{#1}\fi
\ifx \botherref \undefined \def \botherref #1{#1}\fi
\ifx \url \undefined \def \url#1{\textsf{#1}}\fi
\ifx \bchapter \undefined \def \bchapter#1{#1}\fi
\ifx \bbook \undefined \def \bbook#1{#1}\fi
\ifx \bcomment \undefined \def \bcomment#1{#1}\fi
\ifx \oauthor \undefined \def \oauthor#1{#1}\fi
\ifx \citeauthoryear \undefined \def \citeauthoryear#1{#1}\fi
\ifx \endbibitem  \undefined \def \endbibitem {}\fi
\ifx \bconflocation  \undefined \def \bconflocation#1{#1}\fi
\ifx \arxivurl  \undefined \def \arxivurl#1{\textsf{#1}}\fi
\csname PreBibitemsHook\endcsname

\bibitem[\protect\citeauthoryear{Bismuth
  et~al.}{2024}]{DBLP:conf/isaac/BismuthMSS24}
\begin{bchapter}
\bauthor{\bsnm{Bismuth}, \binits{S.}},
\bauthor{\bsnm{Makarov}, \binits{V.}},
\bauthor{\bsnm{Segal{-}Halevi}, \binits{E.}},
\bauthor{\bsnm{Shapira}, \binits{D.}}:
\bctitle{Partitioning problems with splittings and interval targets}.
In: \beditor{\bsnm{Mestre}, \binits{J.}},
\beditor{\bsnm{Wirth}, \binits{A.}} (eds.)
\bbtitle{35th International Symposium on Algorithms and Computation, {ISAAC}
  2024, December 8-11, 2024, Sydney, Australia}.
\bsertitle{LIPIcs},
vol. \bseriesno{322},
pp. \bfpage{12}--\blpage{11215}
(\byear{2024}).
\doiurl{10.4230/LIPICS.ISAAC.2024.12} .
\burl{https://doi.org/10.4230/LIPIcs.ISAAC.2024.12}
\end{bchapter}
\endbibitem

\bibitem[\protect\citeauthoryear{Storer
  et~al.}{1996}]{DBLP:journals/anor/StorerFW96}
\begin{barticle}
\bauthor{\bsnm{Storer}, \binits{R.H.}},
\bauthor{\bsnm{Flanders}, \binits{S.W.}},
\bauthor{\bsnm{Wu}, \binits{S.D.}}:
\batitle{Problem space local search for number partitioning}.
\bjtitle{Ann. Oper. Res.}
\bvolume{63}(\bissue{4}),
\bfpage{463}--\blpage{487}
(\byear{1996})
\doiurl{10.1007/BF02156630}
\end{barticle}
\endbibitem

\bibitem[\protect\citeauthoryear{Pedroso and
  Kubo}{2010}]{DBLP:journals/eor/PedrosoK10}
\begin{barticle}
\bauthor{\bsnm{Pedroso}, \binits{J.P.}},
\bauthor{\bsnm{Kubo}, \binits{M.}}:
\batitle{Heuristics and exact methods for number partitioning}.
\bjtitle{Eur. J. Oper. Res.}
\bvolume{202}(\bissue{1}),
\bfpage{73}--\blpage{81}
(\byear{2010})
\doiurl{10.1016/J.EJOR.2009.04.027}
\end{barticle}
\endbibitem

\bibitem[\protect\citeauthoryear{Kalczynski
  et~al.}{2023}]{DBLP:journals/orf/KalczynskiGD23}
\begin{barticle}
\bauthor{\bsnm{Kalczynski}, \binits{P.J.}},
\bauthor{\bsnm{Goldstein}, \binits{Z.}},
\bauthor{\bsnm{Drezner}, \binits{Z.}}:
\batitle{An efficient heuristic for the k-partitioning problem}.
\bjtitle{Oper. Res. Forum}
\bvolume{4}(\bissue{4}),
\bfpage{70}
(\byear{2023})
\doiurl{10.1007/S43069-023-00249-W}
\end{barticle}
\endbibitem

\bibitem[\protect\citeauthoryear{Anily and
  Federgruen}{1991}]{DBLP:journals/ior/AnilyF91}
\begin{barticle}
\bauthor{\bsnm{Anily}, \binits{S.}},
\bauthor{\bsnm{Federgruen}, \binits{A.}}:
\batitle{Structured partitioning problems}.
\bjtitle{Oper. Res.}
\bvolume{39}(\bissue{1}),
\bfpage{130}--\blpage{149}
(\byear{1991})
\doiurl{10.1287/OPRE.39.1.130}
\end{barticle}
\endbibitem

\bibitem[\protect\citeauthoryear{Faria
  et~al.}{2021}]{DBLP:journals/cor/FariaSS21}
\begin{barticle}
\bauthor{\bsnm{Faria}, \binits{A.F.}},
\bauthor{\bsnm{Souza}, \binits{S.R.}},
\bauthor{\bsnm{S{\'{a}}}, \binits{E.M.}}:
\batitle{A mixed-integer linear programming model to solve the multidimensional
  multi-way number partitioning problem}.
\bjtitle{Comput. Oper. Res.}
\bvolume{127},
\bfpage{105133}
(\byear{2021})
\doiurl{10.1016/J.COR.2020.105133}
\end{barticle}
\endbibitem

\bibitem[\protect\citeauthoryear{Garey and Johnson}{1979}]{gareycomputers}
\begin{bbook}
\bauthor{\bsnm{Garey}, \binits{M.R.}},
\bauthor{\bsnm{Johnson}, \binits{D.S.}}:
\bbtitle{Computers and Intractability: {A} Guide to the Theory of
  NP-Completeness}.
\bpublisher{W. H. Freeman},
\blocation{New York}
(\byear{1979})
\end{bbook}
\endbibitem

\bibitem[\protect\citeauthoryear{Garey and
  Johnson}{1975}]{DBLP:journals/siamcomp/GareyJ75}
\begin{barticle}
\bauthor{\bsnm{Garey}, \binits{M.R.}},
\bauthor{\bsnm{Johnson}, \binits{D.S.}}:
\batitle{Complexity results for multiprocessor scheduling under resource
  constraints}.
\bjtitle{{SIAM} J. Comput.}
\bvolume{4}(\bissue{4}),
\bfpage{397}--\blpage{411}
(\byear{1975})
\doiurl{10.1137/0204035}
\end{barticle}
\endbibitem

\bibitem[\protect\citeauthoryear{Woeginger}{2000}]{woeginger2000does}
\begin{barticle}
\bauthor{\bsnm{Woeginger}, \binits{G.J.}}:
\batitle{When does a dynamic programming formulation guarantee the existence of
  a fully polynomial time approximation scheme ({FPTAS})?}
\bjtitle{{INFORMS} J. Comput.}
\bvolume{12}(\bissue{1}),
\bfpage{57}--\blpage{74}
(\byear{2000})
\end{barticle}
\endbibitem

\bibitem[\protect\citeauthoryear{Brams and Taylor}{1996}]{Brams1996Fair}
\begin{bbook}
\bauthor{\bsnm{Brams}, \binits{S.J.}},
\bauthor{\bsnm{Taylor}, \binits{A.D.}}:
\bbtitle{Fair Division: From Cake Cutting to Dispute Resolution}.
\bpublisher{Cambridge University Press},
\blocation{Cambridge UK}
(\byear{1996})
\end{bbook}
\endbibitem

\bibitem[\protect\citeauthoryear{Brams and Taylor}{2000}]{brams2000winwin}
\begin{bbook}
\bauthor{\bsnm{Brams}, \binits{S.J.}},
\bauthor{\bsnm{Taylor}, \binits{A.D.}}:
\bbtitle{The Win-win Solution - Guaranteeing Fair Shares to Everybody}.
\bpublisher{W. W. Norton {\&} Company},
\blocation{New York}
(\byear{2000})
\end{bbook}
\endbibitem

\bibitem[\protect\citeauthoryear{Brams and Togman}{1996}]{Brams1996Camp}
\begin{barticle}
\bauthor{\bsnm{Brams}, \binits{S.J.}},
\bauthor{\bsnm{Togman}, \binits{J.M.}}:
\batitle{{Camp David: Was The Agreement Fair?}}
\bjtitle{Conflict Management and Peace Science}
\bvolume{15}(\bissue{1}),
\bfpage{99}--\blpage{112}
(\byear{1996})
\end{barticle}
\endbibitem

\bibitem[\protect\citeauthoryear{Massoud}{2000}]{Massoud2000Fair}
\begin{barticle}
\bauthor{\bsnm{Massoud}, \binits{T.G.}}:
\batitle{{Fair Division, Adjusted Winner Procedure (AW), and the
  Israeli-Palestinian Conflict}}.
\bjtitle{Journal of Conflict Resolution}
\bvolume{44}(\bissue{3}),
\bfpage{333}--\blpage{358}
(\byear{2000})
\end{barticle}
\endbibitem

\bibitem[\protect\citeauthoryear{Schneider and
  Kr\"{a}mer}{2004}]{Schneider2004Limitations}
\begin{barticle}
\bauthor{\bsnm{Schneider}, \binits{G.}},
\bauthor{\bsnm{Kr\"{a}mer}, \binits{U.S.}}:
\batitle{{The Limitations of Fair Division}}.
\bjtitle{Journal of Conflict Resolution}
\bvolume{48}(\bissue{4}),
\bfpage{506}--\blpage{524}
(\byear{2004})
\end{barticle}
\endbibitem

\bibitem[\protect\citeauthoryear{Daniel and Parco}{2005}]{Daniel2005Fair}
\begin{barticle}
\bauthor{\bsnm{Daniel}, \binits{T.}},
\bauthor{\bsnm{Parco}, \binits{J.}}:
\batitle{{Fair, Efficient and Envy-Free Bargaining: An Experimental Test of the
  Brams-Taylor Adjusted Winner Mechanism}}.
\bjtitle{Group Decision and Negotiation}
\bvolume{14}(\bissue{3}),
\bfpage{241}--\blpage{264}
(\byear{2005})
\end{barticle}
\endbibitem

\bibitem[\protect\citeauthoryear{Wilson}{1998}]{wilson1998fair}
\begin{botherref}
\oauthor{\bsnm{Wilson}, \binits{S.J.}}:
Fair division using linear programming.
preprint, Departement of Mathematics, Iowa State University
(1998)
\end{botherref}
\endbibitem

\bibitem[\protect\citeauthoryear{Pazner and
  Schmeidler}{1978}]{pazner1978egalitarian}
\begin{barticle}
\bauthor{\bsnm{Pazner}, \binits{E.A.}},
\bauthor{\bsnm{Schmeidler}, \binits{D.}}:
\batitle{{Egalitarian Equivalent Allocations: A New Concept of Economic
  Equity}}.
\bjtitle{Quarterly Journal of Economics}
\bvolume{92}(\bissue{4}),
\bfpage{671}--\blpage{687}
(\byear{1978})
\end{barticle}
\endbibitem

\bibitem[\protect\citeauthoryear{Bei et~al.}{2021}]{bei2020fair}
\begin{barticle}
\bauthor{\bsnm{Bei}, \binits{X.}},
\bauthor{\bsnm{Li}, \binits{Z.}},
\bauthor{\bsnm{Liu}, \binits{J.}},
\bauthor{\bsnm{Liu}, \binits{S.}},
\bauthor{\bsnm{Lu}, \binits{X.}}:
\batitle{Fair division of mixed divisible and indivisible goods}.
\bjtitle{Artif. Intell.}
\bvolume{293},
\bfpage{103436}
(\byear{2021})
\end{barticle}
\endbibitem

\bibitem[\protect\citeauthoryear{Bei
  et~al.}{2023}]{DBLP:journals/corr/abs-2310-00976}
\begin{botherref}
\oauthor{\bsnm{Bei}, \binits{X.}},
\oauthor{\bsnm{Liu}, \binits{S.}},
\oauthor{\bsnm{Lu}, \binits{X.}}:
Fair division with subjective divisibility.
CoRR
\textbf{abs/2310.00976}
(2023)
\doiurl{10.48550/ARXIV.2310.00976}
{\href{https://arxiv.org/abs/2310.00976}{{2310.00976}}}
\end{botherref}
\endbibitem

\bibitem[\protect\citeauthoryear{Bismuth et~al.}{2024}]{segal-halevi2019fair}
\begin{bchapter}
\bauthor{\bsnm{Bismuth}, \binits{S.}},
\bauthor{\bsnm{Bliznets}, \binits{I.}},
\bauthor{\bsnm{Segal-Halevi}, \binits{E.}}:
\bctitle{Fair division with bounded sharing: Binary and non-degenerate
  valuations}.
In: \beditor{\bsnm{Sch{\"a}fer}, \binits{G.}},
\beditor{\bsnm{Ventre}, \binits{C.}} (eds.)
\bbtitle{Algorithmic Game Theory},
pp. \bfpage{89}--\blpage{107}.
\bpublisher{Springer},
\blocation{Cham}
(\byear{2024})
\end{bchapter}
\endbibitem

\bibitem[\protect\citeauthoryear{Sandomirskiy and
  Segal{-}Halevi}{2022}]{sandomirskiy2019fair}
\begin{barticle}
\bauthor{\bsnm{Sandomirskiy}, \binits{F.}},
\bauthor{\bsnm{Segal{-}Halevi}, \binits{E.}}:
\batitle{Efficient fair division with minimal sharing}.
\bjtitle{Oper. Res.}
\bvolume{70}(\bissue{3}),
\bfpage{1762}--\blpage{1782}
(\byear{2022})
\doiurl{10.1287/opre.2022.2279}
\end{barticle}
\endbibitem

\bibitem[\protect\citeauthoryear{Chakraborty
  et~al.}{2021}]{DBLP:journals/teco/ChakrabortyISZ21}
\begin{barticle}
\bauthor{\bsnm{Chakraborty}, \binits{M.}},
\bauthor{\bsnm{Igarashi}, \binits{A.}},
\bauthor{\bsnm{Suksompong}, \binits{W.}},
\bauthor{\bsnm{Zick}, \binits{Y.}}:
\batitle{Weighted envy-freeness in indivisible item allocation}.
\bjtitle{{ACM} Trans. Economics and Comput.}
\bvolume{9}(\bissue{3}),
\bfpage{18}--\blpage{11839}
(\byear{2021})
\doiurl{10.1145/3457166}
\end{barticle}
\endbibitem

\bibitem[\protect\citeauthoryear{Aziz
  et~al.}{2020}]{DBLP:journals/orl/AzizMS20}
\begin{barticle}
\bauthor{\bsnm{Aziz}, \binits{H.}},
\bauthor{\bsnm{Moulin}, \binits{H.}},
\bauthor{\bsnm{Sandomirskiy}, \binits{F.}}:
\batitle{A polynomial-time algorithm for computing a pareto optimal and almost
  proportional allocation}.
\bjtitle{Oper. Res. Lett.}
\bvolume{48}(\bissue{5}),
\bfpage{573}--\blpage{578}
(\byear{2020})
\doiurl{10.1016/J.ORL.2020.07.005}
\end{barticle}
\endbibitem

\bibitem[\protect\citeauthoryear{Farhadi
  et~al.}{2019}]{DBLP:journals/jair/FarhadiGHLPSSY19}
\begin{barticle}
\bauthor{\bsnm{Farhadi}, \binits{A.}},
\bauthor{\bsnm{Ghodsi}, \binits{M.}},
\bauthor{\bsnm{Hajiaghayi}, \binits{M.T.}},
\bauthor{\bsnm{Lahaie}, \binits{S.}},
\bauthor{\bsnm{Pennock}, \binits{D.M.}},
\bauthor{\bsnm{Seddighin}, \binits{M.}},
\bauthor{\bsnm{Seddighin}, \binits{S.}},
\bauthor{\bsnm{Yami}, \binits{H.}}:
\batitle{Fair allocation of indivisible goods to asymmetric agents}.
\bjtitle{J. Artif. Intell. Res.}
\bvolume{64},
\bfpage{1}--\blpage{20}
(\byear{2019})
\doiurl{10.1613/JAIR.1.11291}
\end{barticle}
\endbibitem

\bibitem[\protect\citeauthoryear{Brams and
  Taylor}{1996}]{DBLP:books/daglib/0017730}
\begin{bbook}
\bauthor{\bsnm{Brams}, \binits{S.J.}},
\bauthor{\bsnm{Taylor}, \binits{A.D.}}:
\bbtitle{Fair Division - from Cake-cutting to Dispute Resolution},
(\byear{1996})
\end{bbook}
\endbibitem

\bibitem[\protect\citeauthoryear{Cseh and
  Fleiner}{2020}]{DBLP:journals/talg/CsehF20}
\begin{barticle}
\bauthor{\bsnm{Cseh}, \binits{{\'{A}}.}},
\bauthor{\bsnm{Fleiner}, \binits{T.}}:
\batitle{The complexity of cake cutting with unequal shares}.
\bjtitle{{ACM} Trans. Algorithms}
\bvolume{16}(\bissue{3}),
\bfpage{29}--\blpage{12921}
(\byear{2020})
\doiurl{10.1145/3380742}
\end{barticle}
\endbibitem

\bibitem[\protect\citeauthoryear{Su}{2000}]{DBLP:journals/tamm/Su00}
\begin{barticle}
\bauthor{\bsnm{Su}, \binits{F.E.}}:
\batitle{Cake-cutting algorithms: Be fair if you can. by jack robertson;
  william webb}.
\bjtitle{Am. Math. Mon.}
\bvolume{107}(\bissue{2}),
\bfpage{185}--\blpage{188}
(\byear{2000})
\end{barticle}
\endbibitem

\bibitem[\protect\citeauthoryear{Segal-Halevi}{2019}]{SEGALHALEVI2019123382}
\begin{barticle}
\bauthor{\bsnm{Segal-Halevi}, \binits{E.}}:
\batitle{Cake-cutting with different entitlements: How many cuts are needed?}
\bjtitle{Journal of Mathematical Analysis and Applications}
\bvolume{480}(\bissue{1}),
\bfpage{123382}
(\byear{2019})
\doiurl{10.1016/j.jmaa.2019.123382}
\end{barticle}
\endbibitem

\bibitem[\protect\citeauthoryear{McNaughton}{1959}]{10.2307/2627472}
\begin{barticle}
\bauthor{\bsnm{McNaughton}, \binits{R.}}:
\batitle{Scheduling with deadlines and loss functions}.
\bjtitle{Management Science}
\bvolume{6}(\bissue{1}),
\bfpage{1}--\blpage{12}
(\byear{1959}).
Accessed 2022-12-11
\end{barticle}
\endbibitem

\bibitem[\protect\citeauthoryear{Xing and
  Zhang}{2000}]{DBLP:journals/dam/XingZ00}
\begin{barticle}
\bauthor{\bsnm{Xing}, \binits{W.}},
\bauthor{\bsnm{Zhang}, \binits{J.}}:
\batitle{Parallel machine scheduling with splitting jobs}.
\bjtitle{Discret. Appl. Math.}
\bvolume{103}(\bissue{1-3}),
\bfpage{259}--\blpage{269}
(\byear{2000})
\doiurl{10.1016/S0166-218X(00)00176-1}
\end{barticle}
\endbibitem

\bibitem[\protect\citeauthoryear{Shchepin and Vakhania}{2005}]{ShchepinArticle}
\begin{bchapter}
\bauthor{\bsnm{Shchepin}, \binits{E.}},
\bauthor{\bsnm{Vakhania}, \binits{N.}}:
\bctitle{New tight np-hardness of preemptive multiprocessor and open-shop
  scheduling}.
In: \bbtitle{Proceedings of 2nd Multidisciplinary International Conference on
  Scheduling: Theory and Applications MISTA 2005},
pp. \bfpage{606}--\blpage{629}
(\byear{2005})
\end{bchapter}
\endbibitem

\bibitem[\protect\citeauthoryear{Gonzalez and
  Sahni}{1978}]{DBLP:journals/jacm/GonzalezS78}
\begin{barticle}
\bauthor{\bsnm{Gonzalez}, \binits{T.F.}},
\bauthor{\bsnm{Sahni}, \binits{S.}}:
\batitle{Preemptive scheduling of uniform processor systems}.
\bjtitle{J. {ACM}}
\bvolume{25}(\bissue{1}),
\bfpage{92}--\blpage{101}
(\byear{1978})
\doiurl{10.1145/322047.322055}
\end{barticle}
\endbibitem

\bibitem[\protect\citeauthoryear{Shachnai
  et~al.}{2005}]{DBLP:journals/algorithmica/ShachnaiTW05}
\begin{barticle}
\bauthor{\bsnm{Shachnai}, \binits{H.}},
\bauthor{\bsnm{Tamir}, \binits{T.}},
\bauthor{\bsnm{Woeginger}, \binits{G.J.}}:
\batitle{Minimizing makespan and preemption costs on a system of uniform
  machines}.
\bjtitle{Algorithmica}
\bvolume{42}(\bissue{3-4}),
\bfpage{309}--\blpage{334}
(\byear{2005})
\doiurl{10.1007/s00453-005-1171-0}
\end{barticle}
\endbibitem

\bibitem[\protect\citeauthoryear{Mandal et~al.}{1998}]{MANDAL199891}
\begin{barticle}
\bauthor{\bsnm{Mandal}, \binits{C.A.}},
\bauthor{\bsnm{Chakrabarti}, \binits{P.P.}},
\bauthor{\bsnm{Ghose}, \binits{S.}}:
\batitle{Complexity of fragmentable object bin packing and an application}.
\bjtitle{Computers and Mathematics with Applications}
\bvolume{35}(\bissue{11}),
\bfpage{91}--\blpage{97}
(\byear{1998})
\doiurl{10.1016/S0898-1221(98)00087-X}
\end{barticle}
\endbibitem

\bibitem[\protect\citeauthoryear{Menakerman and
  Rom}{2001}]{DBLP:conf/wads/MenakermanR01}
\begin{bchapter}
\bauthor{\bsnm{Menakerman}, \binits{N.}},
\bauthor{\bsnm{Rom}, \binits{R.}}:
\bctitle{Bin packing with item fragmentation}.
In: \beditor{\bsnm{Dehne}, \binits{F.K.H.A.}},
\beditor{\bsnm{Sack}, \binits{J.}},
\beditor{\bsnm{Tamassia}, \binits{R.}} (eds.)
\bbtitle{Algorithms and Data Structures, 7th International Workshop, {WADS}
  2001, Providence, RI, USA, August 8-10, 2001, Proceedings}.
\bsertitle{Lecture Notes in Computer Science},
vol. \bseriesno{2125},
pp. \bfpage{313}--\blpage{324}
(\byear{2001}).
\doiurl{10.1007/3-540-44634-6\_29} .
\burl{https://doi.org/10.1007/3-540-44634-6\_29}
\end{bchapter}
\endbibitem

\bibitem[\protect\citeauthoryear{Bertazzi
  et~al.}{2019}]{DBLP:journals/dam/BertazziGW19}
\begin{barticle}
\bauthor{\bsnm{Bertazzi}, \binits{L.}},
\bauthor{\bsnm{Golden}, \binits{B.L.}},
\bauthor{\bsnm{Wang}, \binits{X.}}:
\batitle{The bin packing problem with item fragmentation: {A} worst-case
  analysis}.
\bjtitle{Discret. Appl. Math.}
\bvolume{261},
\bfpage{63}--\blpage{77}
(\byear{2019})
\doiurl{10.1016/j.dam.2018.08.023}
\end{barticle}
\endbibitem

\bibitem[\protect\citeauthoryear{Kellerer et~al.}{2003}]{kellerer2003efficient}
\begin{barticle}
\bauthor{\bsnm{Kellerer}, \binits{H.}},
\bauthor{\bsnm{Mansini}, \binits{R.}},
\bauthor{\bsnm{Pferschy}, \binits{U.}},
\bauthor{\bsnm{Speranza}, \binits{M.G.}}:
\batitle{An efficient fully polynomial approximation scheme for the subset-sum
  problem}.
\bjtitle{Journal of Computer and System Sciences}
\bvolume{66}(\bissue{2}),
\bfpage{349}--\blpage{370}
(\byear{2003})
\end{barticle}
\endbibitem

\bibitem[\protect\citeauthoryear{Caprara et~al.}{2000}]{caprara2000multiple}
\begin{barticle}
\bauthor{\bsnm{Caprara}, \binits{A.}},
\bauthor{\bsnm{Kellerer}, \binits{H.}},
\bauthor{\bsnm{Pferschy}, \binits{U.}}:
\batitle{The multiple subset sum problem}.
\bjtitle{SIAM Journal on Optimization}
\bvolume{11}(\bissue{2}),
\bfpage{308}--\blpage{319}
(\byear{2000})
\end{barticle}
\endbibitem

\bibitem[\protect\citeauthoryear{Schreiber et~al.}{2018}]{SchreiberOptimal}
\begin{barticle}
\bauthor{\bsnm{Schreiber}, \binits{E.L.}},
\bauthor{\bsnm{Korf}, \binits{R.E.}},
\bauthor{\bsnm{Moffitt}, \binits{M.D.}}:
\batitle{Optimal multi-way number partitioning}.
\bjtitle{J. {ACM}}
\bvolume{65}(\bissue{4}),
\bfpage{24}--\blpage{12461}
(\byear{2018})
\doiurl{10.1145/3184400}
\end{barticle}
\endbibitem

\bibitem[\protect\citeauthoryear{Korf}{1998}]{KorfCGA}
\begin{barticle}
\bauthor{\bsnm{Korf}, \binits{R.E.}}:
\batitle{A complete anytime algorithm for number partitioning}.
\bjtitle{Artif. Intell.}
\bvolume{106}(\bissue{2}),
\bfpage{181}--\blpage{203}
(\byear{1998})
\doiurl{10.1016/S0004-3702(98)00086-1}
\end{barticle}
\endbibitem

\bibitem[\protect\citeauthoryear{Goldberg et~al.}{2022}]{goldberg2020consensus}
\begin{barticle}
\bauthor{\bsnm{Goldberg}, \binits{P.W.}},
\bauthor{\bsnm{Hollender}, \binits{A.}},
\bauthor{\bsnm{Igarashi}, \binits{A.}},
\bauthor{\bsnm{Manurangsi}, \binits{P.}},
\bauthor{\bsnm{Suksompong}, \binits{W.}}:
\batitle{Consensus halving for sets of items}.
\bjtitle{Math. Oper. Res.}
\bvolume{47}(\bissue{4}),
\bfpage{3357}--\blpage{3379}
(\byear{2022})
\doiurl{10.1287/moor.2021.1249}
\end{barticle}
\endbibitem

\bibitem[\protect\citeauthoryear{Shim and
  Kim}{2008}]{DBLP:journals/cor/ShimK08}
\begin{barticle}
\bauthor{\bsnm{Shim}, \binits{S.}},
\bauthor{\bsnm{Kim}, \binits{Y.}}:
\batitle{A branch and bound algorithm for an identical parallel machine
  scheduling problem with a job splitting property}.
\bjtitle{Comput. Oper. Res.}
\bvolume{35}(\bissue{3}),
\bfpage{863}--\blpage{875}
(\byear{2008})
\doiurl{10.1016/j.cor.2006.04.006}
\end{barticle}
\endbibitem

\bibitem[\protect\citeauthoryear{Yalaoui and
  Chu}{2003}]{doi:10.1080/07408170304382}
\begin{barticle}
\bauthor{\bsnm{Yalaoui}, \binits{F.}},
\bauthor{\bsnm{Chu}, \binits{C.}}:
\batitle{An efficient heuristic approach for parallel machine scheduling with
  job splitting and sequence-dependent setup times}.
\bjtitle{IIE Transactions}
\bvolume{35}(\bissue{2}),
\bfpage{183}--\blpage{190}
(\byear{2003})
\doiurl{10.1080/07408170304382}
{\href{https://arxiv.org/abs/https://doi.org/10.1080/07408170304382}{{https://doi.org/10.1080/07408170304382}}}
\end{barticle}
\endbibitem

\bibitem[\protect\citeauthoryear{Kim
  et~al.}{2004}]{doi:10.1080/00207540410001720745}
\begin{barticle}
\bauthor{\bsnm{Kim}, \binits{Y.-D.}},
\bauthor{\bsnm{Shim}, \binits{S.-O.}},
\bauthor{\bsnm{Kim}, \binits{S.-B.}},
\bauthor{\bsnm{Choi}, \binits{Y.-C.}},
\bauthor{\bsnm{Yoon}, \binits{H.M.}}:
\batitle{Parallel machine scheduling considering a job-splitting property}.
\bjtitle{International Journal of Production Research}
\bvolume{42}(\bissue{21}),
\bfpage{4531}--\blpage{4546}
(\byear{2004})
\doiurl{10.1080/00207540410001720745}
{\href{https://arxiv.org/abs/https://doi.org/10.1080/00207540410001720745}{{https://doi.org/10.1080/00207540410001720745}}}
\end{barticle}
\endbibitem

\bibitem[\protect\citeauthoryear{Kim and Lee}{2021}]{DBLP:journals/cor/KimL21}
\begin{barticle}
\bauthor{\bsnm{Kim}, \binits{H.}},
\bauthor{\bsnm{Lee}, \binits{J.}}:
\batitle{Scheduling uniform parallel dedicated machines with job splitting,
  sequence-dependent setup times, and multiple servers}.
\bjtitle{Comput. Oper. Res.}
\bvolume{126},
\bfpage{105115}
(\byear{2021})
\doiurl{10.1016/j.cor.2020.105115}
\end{barticle}
\endbibitem

\bibitem[\protect\citeauthoryear{Shachnai
  et~al.}{2008}]{DBLP:journals/mst/ShachnaiTY08}
\begin{barticle}
\bauthor{\bsnm{Shachnai}, \binits{H.}},
\bauthor{\bsnm{Tamir}, \binits{T.}},
\bauthor{\bsnm{Yehezkely}, \binits{O.}}:
\batitle{Approximation schemes for packing with item fragmentation}.
\bjtitle{Theory Comput. Syst.}
\bvolume{43}(\bissue{1}),
\bfpage{81}--\blpage{98}
(\byear{2008})
\doiurl{10.1007/s00224-007-9082-x}
\end{barticle}
\endbibitem

\bibitem[\protect\citeauthoryear{Malaguti
  et~al.}{2019}]{DBLP:journals/eor/MalagutiMPP19}
\begin{barticle}
\bauthor{\bsnm{Malaguti}, \binits{E.}},
\bauthor{\bsnm{Monaci}, \binits{M.}},
\bauthor{\bsnm{Paronuzzi}, \binits{P.}},
\bauthor{\bsnm{Pferschy}, \binits{U.}}:
\batitle{Integer optimization with penalized fractional values: The knapsack
  case}.
\bjtitle{Eur. J. Oper. Res.}
\bvolume{273}(\bissue{3}),
\bfpage{874}--\blpage{888}
(\byear{2019})
\doiurl{10.1016/j.ejor.2018.09.020}
\end{barticle}
\endbibitem

\end{thebibliography}

\end{document}


\title[Asymmetric Number Partitioning with Splitting and Interval Targets]{Asymmetric Number Partitioning with Splitting and Interval Targets\footnote{This is an extension of a paper that has been presented at the 35th International Symposium on Algorithms and Computation, (ISAAC) and appeared in its proceedings \cite{DBLP:conf/isaac/BismuthMSS24}. The earlier version focused solely on positive results concerning identical bins. We present here a generalization of these results, for asymmetric bins. This version also contains all proofs
omitted from the previous version.}}

\subtitle{Supplementary Information}


\author*[1]{\fnm{Samuel} \sur{Bismuth}}\email{samuelbismuth101@gmail.com}
\equalcont{These authors contributed equally to this work.}


\author[1]{\fnm{Erel} \sur{Segal-Halevi}}\email{erelsgl@gmail.com}
\equalcont{These authors contributed equally to this work.}

\author[1]{\fnm{Dana} \sur{Shapira}}\email{shapird@g.ariel.ac.il}
\equalcont{These authors contributed equally to this work.}

\affil*[1]{\orgdiv{Department of Computer Science}, \orgname{Ariel University}, \orgaddress{\city{Ariel}, \postcode{40700},  \country{Israel}}}





\maketitle

\section{FPTAS-s for Various Partitioning Problems}\label{app:fptas-woeginger}
Our algorithms use FPTAS-s for several variants of the $n$-way number partitioning problem. These FPTAS-s are developed using a general technique described by \cite{woeginger2000does}.
In this section, we briefly describe this technique.
 
The first step is to develop a dynamic programming (DP) algorithm that solves the problem exactly. The algorithm should have a specific format called \emph{simple DP}, defined by the following parameters:
\begin{itemize}
\item [--] The size of the input vectors, $\alpha\in \mathbb{N}$ (in our settings, usually $\alpha=1$ since each input is a single integer);
\item [--] The size of the \emph{state vectors}, $\beta\in \mathbb{N}$  (in our settings, each state represents a partition of a subset of the inputs);
\item [--] A set of initial states $V_0$ (in our settings, $V_0$ usually contains a single state (the zero vector) representing the empty partition);
\item [--] A set of \emph{transition functions} $F$; each function in $F$ accepts a state and an input, and returns a new state (in our settings, each function in $F$ corresponds to putting the new input in one of the machines);
\item [--] A set of \emph{filter functions} $H$; each function $h_j\in H$ corresponds to a function $f_j\in F$. It accepts a state and an input and returns a positive value if the new state returned by $f_j$ is infeasible (infeasible states are kept out of the state space).
\item [--] An \emph{objective function} $G$, that maps a state to a numeric value.
\end{itemize}
The DP algorithm processes the inputs one by one. For each input $k$, it applies every transition function in $F$ to every state in $V_{k-1}$, to produce the new state-set $V_k$. Finally, it picks the state in $V_m$ that minimizes the objective function. Formally:
\begin{enumerate}
\item Let $V_0$ be the set of initial states.
\item For $k := 1,\ldots,m$:
\begin{align*}
V_k := \{
f_j(s,x) 
| 
f_j\in F,~ s\in V_{k-1},~ h_j(s,x)\leq 0
\}
\end{align*}
\item Return 
\begin{align*}
\min\{G(s) | s\in V_m\}
\end{align*}
\end{enumerate}
Every DP in this format can be converted to an {\sf FPTAS} if it satisfies a condition called \emph{critical-coordinate-benevolence} (CC-benevolence).
To prove that a dynamic program is CC-benevolent, we need to define a \emph{degree vector} $D$ of size $\beta$, which determines how much each state is allowed to deviate from the optimal state. 
The functions in $F$, $G$, $H$ and the vector $D$ should satisfy several conditions listed in Lemma 6.1 of \cite{woeginger2000does}.
These conditions use the term \emph{$[D,\Delta]$-close}, defined as follows.

For a real number $\delta > 1$ and two vectors 
$V = [w_1,\dots, w_{\beta}]$ and $V' = [w'_1,\dots, w'_{\beta}]$, we say that $V$ 
is $[D,\Delta]-$close to $V'$ if 
\begin{equation*}
\Delta^{-d_\ell} \cdot w_\ell ~\le~ w'_\ell ~\le~ \Delta^{d_\ell} \cdot w_\ell, \text{ for all } \ell=1,\dots,\beta
\end{equation*}
that is, each coordinate in $V'$ deviates from the corresponding coordinate in $V$ by a multiplicative factor determined by $\Delta$ and by the degree vector $D$.
The $\Delta$ is a factor determined by the required approximation accuracy $\epsilon$.
Note that if some coordinate $\ell$ in $D$ is $0$, then the definition of $[D,\Delta]$-close requires that the coordinate $\ell$ in $V'$ is equal to coordinate $\ell$ in $V$ (no deviation is allowed).

We now use Lemma 6.1 of \cite{woeginger2000does} to prove that the problems \partition{\asym\minmax}{$n,v,i$}{${\cal X}, {\cal R}$} and \partition{\asym\minmax}{$n=2,v_1, v_2,i$}{${\cal X},{\cal R}, m_1, m_2$} are CC-benevolent, and therefore there is an {\sf FPTAS} for solving these problems.


\subsection{\partition{\asym\minmax}{$n,v,i$}{${\cal X}, {\cal R}$}}\label{subsec:appendix-n-way-partition-cc}

Recall that the objective of \partition{\asym\minmax}{$n,v,i$}{${\cal X}, {\cal R}$} is to find an $n$-way partition of ${\cal X}$ minimizing the largest bin relative sum, subject to the constraint that the sum of bin $i$ is at most $(1+v)\cdot S$, where $v\in\mathbb{Q}$ is given in the input and $S = (\sum_{x\in {\cal X}} x)/\sum$ is the average bin relative sum.

\begin{proof}
\textbf{The dynamic program}. 
We define a simple DP with $\alpha = 1$ (the size of the input vectors) and $\beta=n$ (the size of the state vectors).
For $k=1,\dots,m$ define the input vector
$U_k=[x_k]$
where $x_k$ is the size of item $k$. A state $V=[b_1,b_2,\dots, b_n]$ in $V_k$ encodes a partial allocation for the first $k$ items, 
where $b_i$ is the sum of items in bin $i\in\{1,2\dots,n\}$ in the partial allocation. The set $F$ contains $n$ transition functions $f_1,  \dots, f_n$:
\begin{align*}
	& f_1(x_k, b_1, b_2, \dots, b_n) = [b_1+x_k/r_1, b_2, \dots, b_n] \\
	& f_2(x_k, b_1, b_2, \dots, b_n) = [b_1, b_2+x_k/r_2, \dots, b_n] \\
	& \dots \\
	& f_n(x_k, b_1, b_2, \dots, b_n) = [b_1, b_2, \dots, b_n+x_k/r_n] 
\end{align*}
Intuitively, the function $f_i$ corresponds to putting item $k$ in bin $i$. 
The set $H$ contains a function $h_1(x_k, b_1, b_2,\dots,b_n) = b_1 + x_k/r_i - (1+v)\cdot S$,
and functions $h_j(x_k, b_1, b_2, \dots, b_n) \equiv 0$ for $j \in \{n\} \setminus i$.
These functions represent the fact that the relative sum of the bin $i$ must always be at most $(1+v)\cdot S$ (if it becomes larger than $(1+v)\cdot S$, then $h_i$ will return a positive value and the new state will be filtered out).
There are no constraints on the sums of bins that are not $i$.

The initial state space $V_0$ contains a single vector $\{[0,0,\dots,0]\}$. The minimization objective is
\begin{align*}
G(b_1,b_2,\dots,b_n) = \max\{b_1,b_2,\dots,b_n\}.
\end{align*}
\textbf{Benevolence}. 
We show that our problem is \emph{CC-benevolent}, as defined at \cite{woeginger2000does}[Section 6].
We use the degree vector $D=[1,1,\dots,1]$,
where $i$ is the critical coordinate.

All the transition functions are polynomials of degree 1.
The value of the function $f_1(U, V)$ only depends on $x_k$ which is on $U$, and $i$ which is our critical-coordinate.
So, C.1(i) on the function set $F$ is fulfilled.
	
The functions $h_1, h_2, \dots, h_n$ are polynomials; the monomials do not depend on relative sums that are not $b_i$, and the monomial that depends on $b_i$ has a positive coefficient. 
So, C.2(i) on the function set $H$ is fulfilled.
		
If a state $[b_1,b_2,\dots,b_n]$ is $[D,\Delta]$-close to another state $[b_1',b_2',\dots,b_n']$,
then by $[D,\Delta]$-close definition, we have 
$w'_1 \leq \Delta \cdot w_1$,
$w'_2 \leq \Delta \cdot w_2$,
 		\ldots,
$w'_n \leq \Delta \cdot w_n$, so 
$$\max\{b_1',b_2',\dots,b_n'\} \leq \Delta\cdot \max\{b_1,b_2,\dots,b_n\},$$ and 
$G(V') \leq \Delta \cdot G(V)$
Therefore, C.3(i) on the function $G$ is fulfilled
(with degree $g=1$).

The definition in \cite{woeginger2000does} also allows a \emph{domination relation}, but we do not need it in our case (formally, our domination relation $\preceq_{dom}$ is the trivial relation).
So, the statement conditions C.1(ii), C.2(ii), C.3(ii) are fulfilled.

Condition C.4 (i) holds since all functions in $F$ can be evaluated in polynomial time.
C.4 (ii) holds since the cardinality of $F$ is a constant.
C.4 (iii) holds since the cardinality of $V_0$ is a constant.
C.4(iv) is satisfied since the value of the coordinates is upper bounded by the sum of the items $\sumratios{} S$, so their logarithm is bounded by the size of the input.
Hence, our problem is CC-benevolent. 
By the main theorem of \cite{woeginger2000does}, it has an {\sf FPTAS}.
\end{proof}

\subsection{\partition{\asym\minmax}{$n=2,v_1, v_2,i$}{${\cal X},{\cal R}, m_1, m_2$}}\label{subsec:appendix-card-partition-cc}

Recall that the objective is to find a $2$-way partition of ${\cal X}$ minimizing the largest bin relative sum, subject the the constraint that the sum of bin $i$ is at most $(1+v_i)\cdot S$, and additionally, bin $q$ must contain exactly $m_q$ items.

\begin{proof}
\textbf{The dynamic program}. We define a simple DP with $\alpha = 1$ (the size of the input vectors) and $\beta=4$ (the size of the state vectors).
For $k=1,\dots,m$ define the input vector 
$U_k=[x_k]$
, where $x_k$ is the size of item $k$. A state 
$V = [b_1,b_2, l_1, l_2]$ in $V_k$
encodes a partial allocation for the first $k$ items, 
where $b_i$ is the relative sum of items in bin $i\in\{1,2\}$ in the partial allocation and $l_i$ is the number of items in bin $i$. The set $F$ contains two transition functions $f_1$ and $f_2$:
\begin{align*}
	f_1(x_k, b_1, b_2, l_1, l_2) = [b_1+x_k/r_1, b_2, l_1 + 1, l_2] \\
	f_2(x_k, b_1, b_2, l_1, l_2) = [b_1, b_2+x_k/r_2, l_1, l_2 + 1]
\end{align*}
Intuitively, the function $f_i$ corresponds to putting item $k$ in bin $i$. 
The set $H$ contains a function $h_i(x_k, b_1, b_2, l_1, l_2) = b_1 + x_k/r_i - (1+v)S$ and $h_j(x_k, b_1, b_2, l_1, l_2) \equiv 0$ for $j \in [2] \setminus i$. These functions represent the fact that the relative sum of bin $i$ must always be at most $(1+v)S$ (if it becomes larger than $(1+v)S$, then $h_i$ will return a positive value and the new state will be filtered out). There are no constraints on the relative sum of bins that are not $i$.

The initial state space $V_0$ contains a single vector $\{[0,0,\dots,0]\}$.
Finally, set the minimization objective to
\begin{align*}
	G(b_1,b_2,l_1,l_2) = 
	\begin{cases}
	\max\{b_1,b_2\},& \text{if } l_1 = m_1 \text{ and } l_2 = m_2\\
	\infty,              & \text{otherwise}
	\end{cases}
\end{align*}

\textbf{Benevolence}. 
We show that our problem is \emph{CC-benevolent}, as defined at \cite{woeginger2000does}[Section 6].

We define the degree vector as $D=[1,1,0,0]$.  Note that the third and fourth coordinates correspond to the number of items in each bin, for which we need an exact number and not an approximation. 

Lemma 4.1(i) is satisfied since if a state $[b_1,b_2,l_1,l_2]$ is $[D,\Delta]$-close to another state $[b_1',b_2',l_1',l_2']$, by definition of $[D,\Delta]$-close [Section 2](2.1), we 
must have $l_1=l_1'$ and $l_2=l_2'$ because the degree of coordinates 3 and 4 is $0$. 
At coordinates 1 and 2, both transition functions are polynomials of degree 1.
Furthermore, the value of the function $f_1(U, V)$ only depends on $x_k$ which is on $U$, and $b_1$ which is our critical-coordinate.
So, C.1(i) on the function set $F$ is fulfilled.
	
The functions $h_1$ and $h_2$ are polynomials; the monomials do not depend on $b_2$, and the monomial that depends on $b_1$ has a positive coordinate.
So, C.2(i) on the function set $H$ is fulfilled.
	
If a state $[b_1,b_2,l_1,l_2]$ is $[D,\Delta]$-close to another state $[b_1',b_2',l_1',l_2']$ (where $\Delta$ is a factor determined by the required approximation accuracy $\epsilon$), by $[D,\Delta]$-close definition, we have 
$l_1=l_1'$ and $l_2=l_2'$ because the degree of coordinates 3 and 4 is $0$, also
$b'_1 \leq \Delta b_1$ and 
$b'_2 \leq \Delta b_2$, so 
$\max\{b_1',b_2'\} \leq \Delta \max\{b_1,b_2\}$, so
$G(V') \leq \Delta^{g} G(V)$.
Therefore, C.3(i) on the function $G$ is fulfilled
(with degree $g=1$).

Again we do not need a domination relation, so we use the trivial relation $\preceq_{dom}$. 
Then the statement conditions C.1(ii), C.2(ii), C.3(ii) are fulfilled.

Condition C.4 (i) holds since all functions in $F$ can be evaluated in polynomial time.
C.4 (ii) holds since the cardinality of $F$ is a constant.
C.4 (iii) holds since the cardinality of $V_0$ is a constant.
C.4(iv) is satisfied for coordinates 1, 2 since their value is upper bounded by the sum of the items $\sumratios{} S$, so their logarithm is bounded by the size of the input.
For coordinates 3, 4 (whose degree is 0), the condition is satisfied since their value is upper bounded by the number of items $m$.
Hence, our problem is CC-benevolent. 
By the main theorem of \cite{woeginger2000does}, it has an {\sf FPTAS}.
\end{proof}

\section{Additional Plots}
\label{sec:appendix-plots}
We present the plots obtained for different values of the parameters presented in the main paper experiments. All the plots behave similarly, showing that the difference between the optimal and the perfect partition decreases when the number of split items increases.

\begin{figure}[H]
\begin{center}
\includegraphics[scale=0.42]{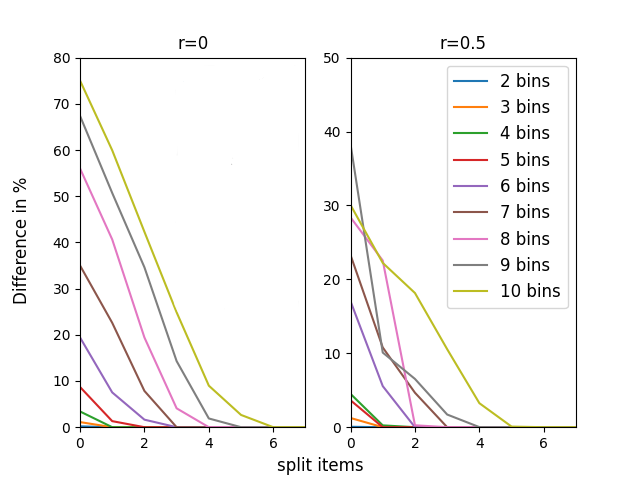}
\end{center}
\caption{Uniform distribution with $m=10$ and 16 bits.
}
\end{figure}

\begin{figure}[H]
\begin{center}
\includegraphics[scale=0.42]{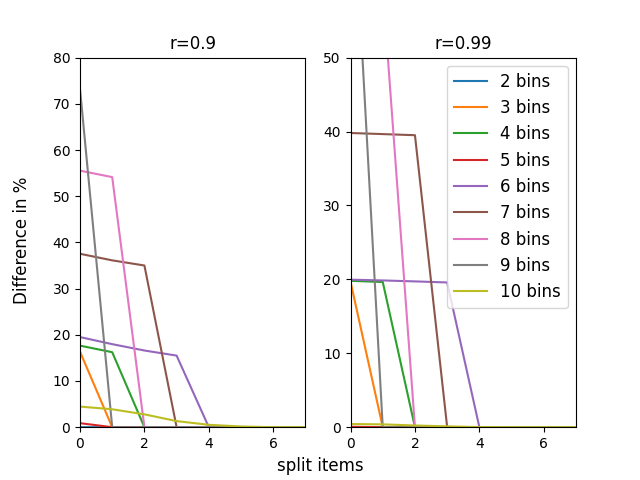}
\end{center}
\caption{Uniform distribution with $m=10$ and 16 bits.
}
\end{figure}

\begin{figure}[H]
\begin{center}
\includegraphics[scale=0.42]{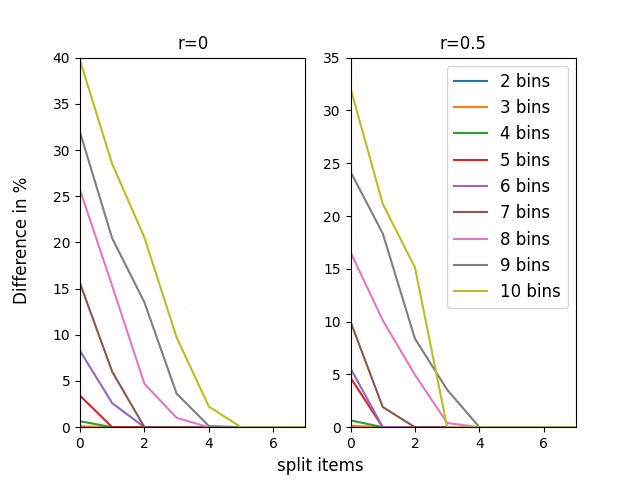}
\end{center}
\caption{Uniform distribution with $m=13$ and 16 bits.
}
\end{figure}

\begin{figure}[H]
\begin{center}
\includegraphics[scale=0.42]{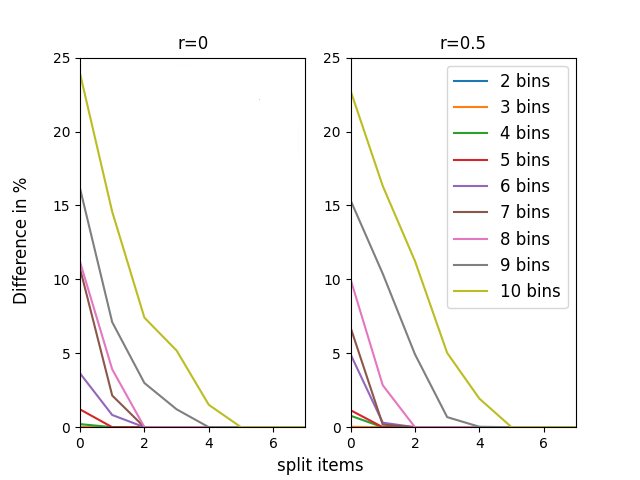}
\end{center}
\caption{Uniform distribution with $m=15$ and 16 bits.
}
\end{figure}

\begin{figure}[H]
\begin{center}
\includegraphics[scale=0.42]{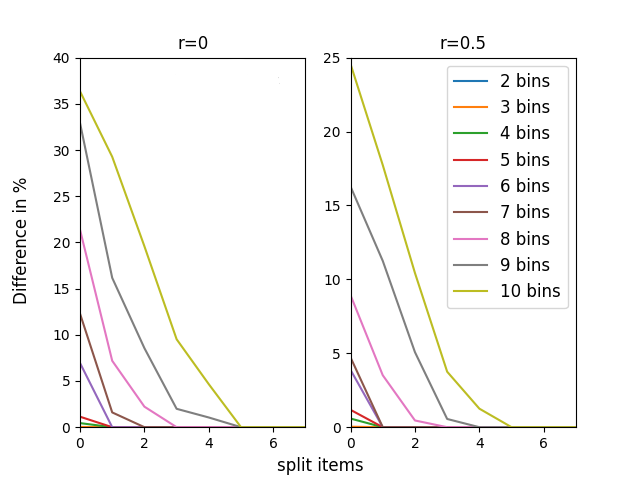}
\end{center}
\caption{Uniform distribution with $m=15$ and 32 bits.
}
\end{figure}

\begin{figure}[H]
\begin{center}
\includegraphics[scale=0.42]{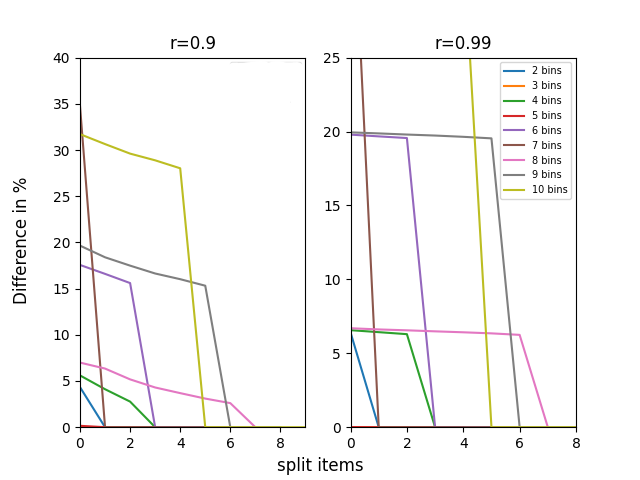}
\end{center}
\caption{Uniform distribution with $m=15$ and 32 bits.
}
\end{figure}

\begin{figure}[H]
\begin{center}
\includegraphics[scale=0.42]{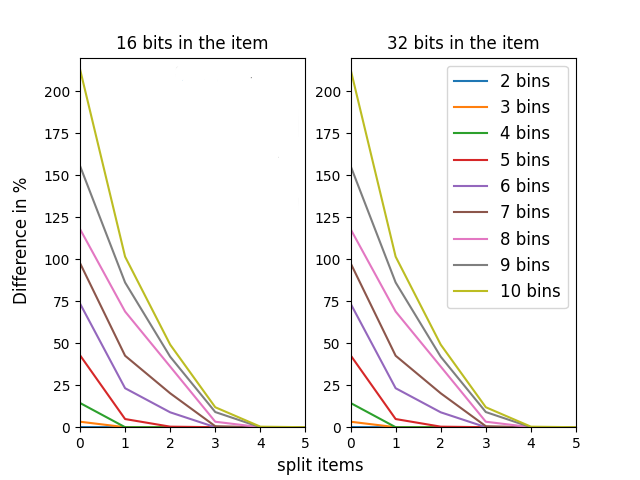}
\end{center}
\caption{Exponential distribution with $m=10$.
}
\end{figure}

\begin{figure}[H]
\begin{center}
\includegraphics[scale=0.42]{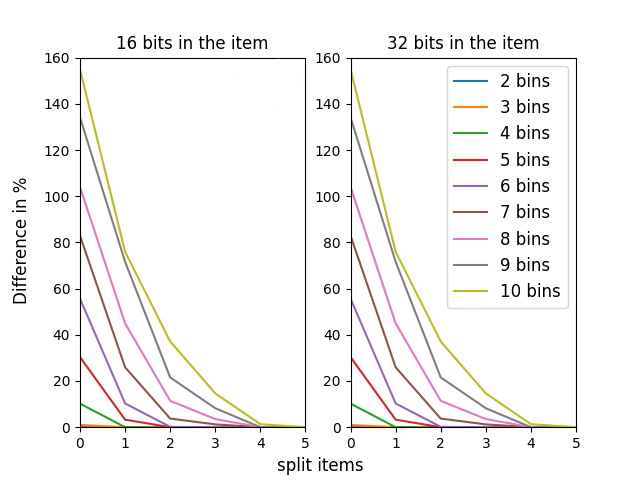}
\end{center}
\caption{Exponential distribution with $m=13$.
}
\end{figure}

\begin{figure}[H]
\begin{center}
\includegraphics[scale=0.42]{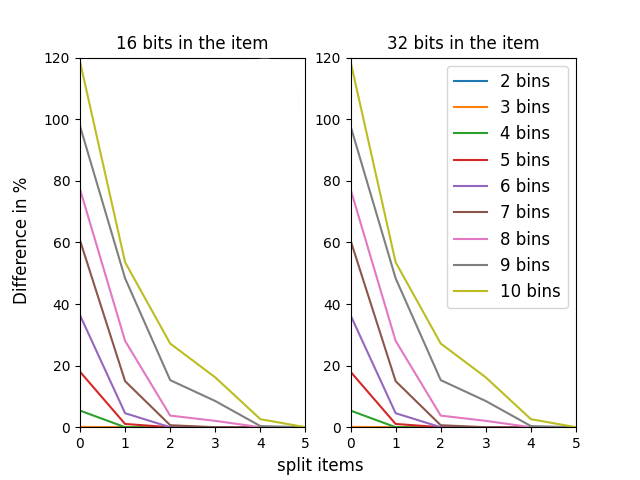}
\end{center}
\caption{Exponential distribution with $m=15$.
}
\end{figure}

\bibliography{references}